\documentclass[reqno]{amsart}
\usepackage{url}
\usepackage{amsmath}
\usepackage{mathtools}

\DeclareMathOperator*{\argmin}{arg\,min}
\DeclarePairedDelimiter\floor{\lfloor}{\rfloor}
\usepackage{amsfonts}
\usepackage{amssymb}
\usepackage{xcolor}
\usepackage{graphicx}
\usepackage{float}
\usepackage{enumerate}
\usepackage{multirow}

\newcommand{\norm}[1]{\left\lVert#1\right\rVert}

\usepackage{mathtools}
\usepackage{algorithmic, algorithm} 
\newtheorem{theorem}{Theorem}
\newtheorem{corollary}{Corollary}

\newtheorem{definition}{Definition}
\newtheorem{lemma}{Lemma}
\newtheorem{proposition}{Proposition}
\newtheorem{remark}{Remark}

\usepackage[export]{adjustbox}
\usepackage{commath}
\usepackage{hyperref}
\hypersetup{
	colorlinks,
	linkcolor=black,
	anchorcolor=black,
	citecolor=black
}
\newcommand*\diff{\mathop{}\!\mathrm{d}}

\begin{document}

\title[Ramanujan meets time-frequency analysis]{When Ramanujan meets time-frequency analysis in complicated time series analysis}

\author{Ziyu Chen}
\address{Department of Mathematics, Duke University, Durham, NC, USA}
\email{ziyu@math.duke.edu}

\author{Hau-Tieng Wu}
\address{Department of Mathematics and Department of Statistical Science, Duke University, Durham, NC, USA. Mathematics Division, National Center for Theoretical Sciences, Taipei, Taiwan}
\email{hauwu@math.duke.edu}

\maketitle

\begin{abstract}
To handle time series with complicated oscillatory structure, we propose a novel time-frequency (TF) analysis tool that fuses the short time Fourier transform (STFT) and periodic transform (PT).
Since many time series oscillate with time-varying frequency, amplitude and non-sinusoidal oscillatory pattern, a direct application of PT or STFT might not be suitable. However, we show that by combining them in a proper way, we obtain a powerful TF analysis tool.
We first combine the Ramanujan sums and $l_1$ penalization to implement the PT. We call the algorithm Ramanujan PT (RPT). The RPT is of its own interest for other applications, like analyzing short signal composed of components with integer periods, but that is not the focus of this paper. 
Second, the RPT is applied to modify the STFT and generate a novel TF representation of the complicated time series that faithfully reflect the instantaneous frequency information of each oscillatory components. 
We coin the proposed TF analysis the Ramanujan de-shape (RDS) and vectorized RDS (vRDS).
In addition to showing some preliminary analysis results on complicated biomedical signals, we provide theoretical analysis about RPT. Specifically, we show that the RPT is robust to three commonly encountered noises, including envelop fluctuation, jitter and additive noise. 
\end{abstract}

\section{Introduction}\label{Sect:Introduction}

Oscillatory signals are ubiquitous in our life. For example, signals in seismology, finance and medicine are usually oscillatory. They are often ``non-stationary''\footnote{Technically speaking, a random process that is not stationary in the wide sense is non-stationary. In our setup, the meaning of non-stationary is more qualitative and colloquial. A rigorous quantification of our model will be shown later.}, composed of multiple oscillatory components, and each oscillatory component carries time-varying features, like time-varying frequency, time-varying amplitude, or even time-varying oscillatory pattern. In practice, these time-varying quantities encode abundant information of the underlying system. A common interest in signal processing is extracting those time-varying quantities from the given signal. With these quantities, researchers can infer the underlying system behavior, or forecast the upcoming events.

To appreciate the challenge of achieving the above-mentioned signal processing interests, see the trans-abdominal maternal electrocardiogram (ta-mECG) signal shown in Figure \ref{fig:1}.
In this example, the signal is clearly oscillatory, and we briefly summarize its characteristics. 
\begin{enumerate}
\item [(C1)] It is composed of two oscillations; one is the maternal ECG (mECG), and the other one is the fetal ECG (fECG). Moreover, there is a slowly varying trend, which is usually called the baseline wandering.

\item [(C2)] Although it is not easy to identify from the signal, according to physiological knowledge, the cycle-to-cycle period changes from time to time, both in mECG and fECG. This variation is usually known as the {\em heart rate variability} (HRV) \cite{electrophysiology1996heart}. 

\item [(C3)] The oscillation is far from sinusoidal, whose ``shape'' encodes the heart's electrophysiological information, like arrhythmia. Moreover, due to the nonlinear relationship between the cycle-to-cycle period and the ventricular electrophysiological dynamics, the cycle morphology changes from cycle to cycle as well \cite{malik2002relation}. This variation can be quantified as the QT interval variability \cite{baumert2016qt}. 

\item [(C4)] The amplitude of each cycle changes due to the respiration-induced impedance variation \cite{schmidt2017ecg}. 
\end{enumerate}
These four characteristics are time-varying, and quantifications of these characteristics allow researchers or clinicians to understand more health-related information.
There are many other signals sharing similar complicated structures, not only from medicine. We need a systematic signal processing tool to extract the information we have interest.
In this paper, we propose a novel time-frequency (TF) analysis tool aiming to extract how fast the signal oscillates at each moment.

\begin{figure}[!htbp]
\includegraphics[trim=50 210 50 220,clip,width=1\textwidth]{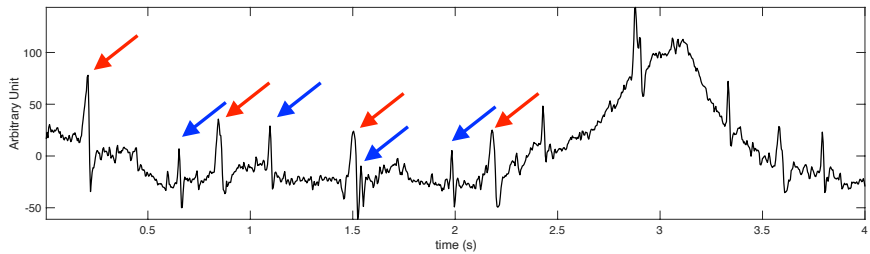}
\caption{The ta-mECG signal of 4 seconds long. The signal is sampled at 250Hz. The spikes indicated by red arrows are mECG, and those spikes indicated by blue arrows are fECG.}
\label{fig:1}
\end{figure}

TF analysis extracts the temporal and spectral information from the signal simultaneously. The short-time Fourier transform (STFT) and continuous wavelet transform (CWT) are probably the most well known TF analysis tools. Usually, a TF analysis tool converts the given signal, a function defined on $\mathbb{R}$, which parametrizes time, into a TF representation (TFR), a function defined on $\mathbb{R}^2$, which parametrizes time and frequency. The user can then read the oscillatory information from the TFR; for example, how fast and strong the signal oscillates at a given time, or if there is a singularity. Usually, the modulation of the TFR is represented as an image for users' visual inspection; for example, the squared modulation of STFT is the widely applied spectrogram, and the squared modulation of CWT is the scalogram.

In this paper, we focus on the common challenge all these TF tools face when the signal does not oscillate sinusoidally, as is illustrated in (C3) above. Since the oscillation is not sinusoidal, in the frequency domain the oscillatory pattern is ``automatically'' decomposed into multiple oscillatory components via the Fourier series, and this ``automatic'' decomposition generates a lot of trouble when we interpret the output TFR. 
Let us take a look at this fact mathematically. Consider a function that oscillates regularly with a non-sinusoidal pattern that repeats every $1/\xi_0$ seconds, where $\xi_0>0$ can be viewed as the frequency of the oscillation, and the strength is $A>0$. We represent the non-sinusoidal oscillatory pattern by a $1$-periodic function $s$, and to simplify the discussion, we assume it has sufficient regularity so that it can be pointwisely expanded by its Fourier series; that is, $s(t)=\sum\limits_{n=0}^{\infty}a_n\cos(2\pi n t + \alpha_n )$, where $a_n\geq 0$, and $\alpha_n$ is the phase and $\alpha_0 = 0$. In this setup, the oscillatory signal can be written as
\[
f(t) = A s(\xi_0 t) = \sum_{n=0}^{\infty}Aa_n\cos(2\pi n \xi_0 t + \alpha_n ),
\]
where each term $Aa_n\cos(2\pi n \xi_0 t + \alpha_n )$ can be viewed as an oscillatory component, $n\geq 1$. By a simple calculation, we get the STFT with a window function $g$:
\begin{equation}\label{Equation STFT first example}
V_f^{(g)}(t,\xi)=Aa_0 e^{-2\pi i t \xi}\hat{g}(\xi)+\frac{A}{2}\sum_{n\in \mathbb{Z},n\neq 0} a_n e^{i\text{sign}(n)\alpha_n} \hat{g}(\xi- n\xi_0)e^{2\pi i t(\xi-n\xi_0)}\,.
\end{equation}
Clearly, we will see multiple components in the spectrogram, the squared magnitude of $V_f^{(g)}(t,\xi)$.
Figure \ref{fig:shape} gives an example of the spectrogram of a signal $f(t) = s(\phi(t))$, where $s$ is a non-sinusoidal function and $\phi(t)$ is a smooth and monotonically increasing function. We can clearly observe multiple components. If our interest is extracting how fast or how large the signal oscillates at each time, this TFR might be misleading and difficult to interpret.

\begin{figure}[!htbp]
	\includegraphics[trim=60 20 70 20,clip,width=1\textwidth]{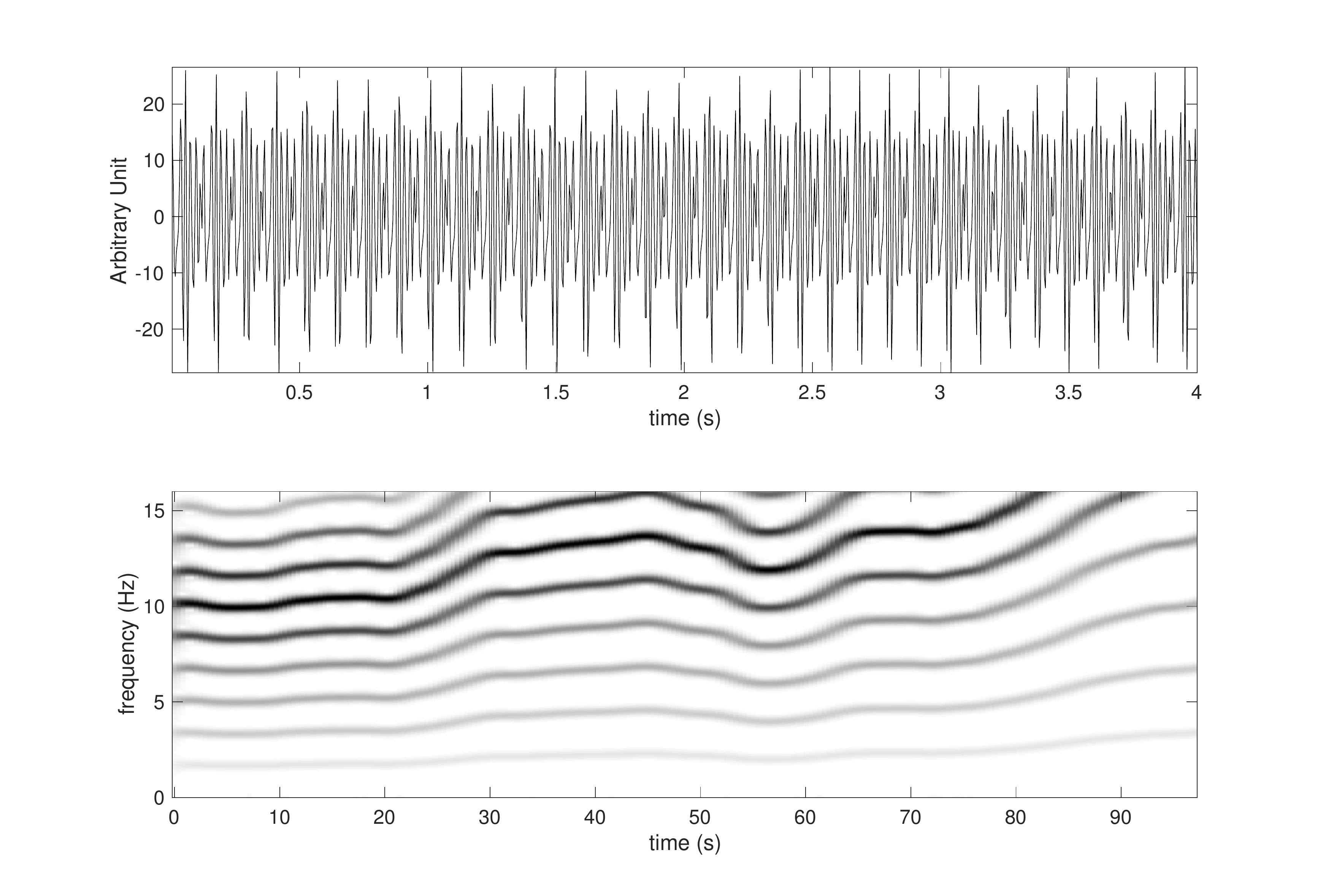}
	\caption{Top: plot of $f(t) = s(\phi(t))$ up to 4 seconds; bottom: the spectrogram of $f(t)$, where $s(t) = \cos(2\pi t) + 2\cos(4\pi t+2) + 3\cos(6\pi t+1) + 4\cos(8\pi t) + 7\cos(10\pi t) - 10\cos(12\pi t) + 8\cos(14\pi t) + 6\cos(16\pi t-1.5) + 3\cos(18\pi t+0.8) + \cos(20\pi t-0.2)$.}
	\label{fig:shape}
\end{figure}

To appreciate how this trouble impacts the interpretation for real data, see Figure \ref{fig:2} for the spectrogram of the ta-mECG shown in Figure \ref{fig:1}. It is clear that there are various ``curves'' in the TF domain. Since we have ground truth information for this signal (the direct contact fECG is available in this database), we know that those curves indicated by blue arrows are associated with the mECG, and those light curves indicated by red arrows are associated with the fECG. These curves come from the non-sinusoidal oscillatory patterns of mECG and fECG.
However, if we are interested in getting the maternal instantaneous heart rate (IHR) and the fetal IHR for various clinical applications, like maternal stress detection \cite{lobmaier2019fetal}, it is difficult to directly extract this information from the TFR. Note that IHR quantifies how fast the signal oscillates at each moment.
Although not shown, we mention that the scalogram and other TFRs face the same challenge. A summary and comparison of different techniques of handling the single-channel ta-mECG is out of the scope of this paper, and we refer readers with interest to \cite{su2019recovery} and citations therein.

A natural question to ask is if it is possible to come out with a TF analysis tool that is immune to the non-sinusoidal oscillatory patterns? In \cite{Deshape}, the {\em de-shape} algorithm was introduced to handle this challenge. It modifies the well known cepstrum idea, called {\em inverse cepstrum}, to eliminate the impact of non-sinusoidal oscillatory pattern in the TFR. While the de-shape algorithm has been applied to study various biomedical signals, like the ta-mECG \cite{deshapeapp}, the respiratory signal \cite{recycling}, and the sawtooth artifacts in the pulse transit time \cite{lin2019unexpected}, it has several limitations. 
For example, when the oscillatory pattern is also time varying, the inverse cepstrum might not perform well; due to the interaction of the inverse cepstrum and the non-sinusoidal oscillation, the TFR from the de-shape algorithm might not correctly encode the oscillatory strength information. We will elaborate its limitations in more detail later in Section \ref{Section: dsAlgorithm Intro}.

To resolve these limitations, we propose a novel TF analysis algorithm inspired by the periodic transform (PT) that has been actively studied in the signal processing society \cite{sethares1999,PPV1,PPV2,PPV3,unified,iMUSIC}. Based on the nice properties of the Ramanujan sum and Ramanujan subspace \cite{PPV1,PPV2,PPV3,unified}, we propose a $l_1$ minimization approach for the PT under the Ramanujan subspace, which we call the Ramanujan periodic transform (RPT). Based on the RPT, and a critical observation of the periodic structure in the spectrogram, we propose a novel TF analysis algorithm called the {\em Ramanujan de-shape} (RDS), and its vectorization version called the {\em vectorized Ramanujan de-shape} (vRDS). We also provide theoretical support for the RPT, particularly three robustness results. Specifically, we show that the RPT is robust to envelope fluctuation, jitter, and additive noise. These results explain the performance of RDS and  vRDS. We also demonstrate its performance by showing some preliminary results on real biomedical signals.

We will use following notations throughout this paper. Denote $\mathbb{N}^+=\{1,2,\ldots\}$ be the set of positive natural numbers, and $\mathbb{N}=\mathbb{N}^+\cup \{0\}$. $(p,q)$ denotes the {\em greatest common divisor} of two integers $p$ and $q$ and $lcm$ stands for {\em least common multiple}. For a real number $x$, $\floor*{x}$ denotes the largest integer no greater than $x$. For two integers $p$ and $q$, $p|q$ denotes that $p$ is a divisor of $q$. Given two sequences $f(n)$ and $g(n)\geq 0$, we say $f(n)=O(g(n))$ if there exists a constant $c_1>0$ such that $|f(n)| \leq c_1 g(n)$; $f(n)=\Omega(g(n))$ if there exists $c_2>0$ such that $|f(n)| \geq c_2 g(n)$; $f(n)=\Theta(g(n))$ if there exists $c_1$ and $c_2$ such that $c_1|g(n)|\leq |f(n)|\leq c_2|g(n)|$.

The paper is organized in the following way. In Section \ref{Section mathematical background}, we summarize necessary mathematical background, including the mathematical model we consider for the complicated time series and the basic framework for the PT. In Section \ref{Section Proposed time-frequency analysis algorithm}, we briefly review the de-shape algorithm and then introduce our RDS and  vRDS algorithms. We illustrate the effectiveness of RDS and  vRDS by several numerical examples in Section \ref{Section Numerical results}. We leave the proofs of the theoretical foundations of RDS in Section \ref{Section Theoretical results}. 
Section \ref{Section Discussion} summarizes the paper.
\section{Mathematical background}\label{Section mathematical background}

In this section, we summarize the model we choose to model high-frequency oscillatory time series. Then we summarize an existing algorithm aiming to handle this kind of signal, and describe its limitation as a motivation for the proposed algorithm in this study.

\subsection{Adaptive harmonic model}

In \cite{SST}, motivated by capturing time-varying amplitude and frequency of a given time series, Daubechies et al. consider the {\em intrinsic mode type} (IMT) function:
\begin{equation}\label{AHM} 
f(t) = A(t)\cos(2\pi\varphi(t)),
\end{equation}
where $A(t)>0$ is a smooth function that represents the time-varying amplitude of the signal, and $\varphi(t)$ is a smooth monotonically increasing function that describes the phase of the signal. We call $A(t)$ the \textit{amplitude modulation} (AM) and $\varphi'(t)$ the \textit{instantaneous frequency} (IF). In practice, $\frac{1}{\varphi'(t)}$ can be understood as the {\em instantaneous period} (IP) of the oscillation at time $t$. We need some conditions for $A$ and $\varphi$, otherwise we cannot identify $A$ and $\varphi$ in general. A common condition is assuming that $A$ and $\varphi'$ changes {\em slowly} in the following sense. Fix a small $\epsilon>0$, we assume $|A'(t)|\leq \epsilon \varphi'(t)$ and $|\varphi''(t)|\leq \epsilon \varphi'(t)$ for all $t\in \mathbb{R}$. Under this condition, the identifiability of $A$, $\varphi$ and $\varphi'$ of the IMT function has been proved in \cite{Chen_Cheng_Wu:2014}.
A signal satisfies the \textit{adaptive harmonic model (AHM)} if it is composed of one or multiple IMT functions, 
\begin{equation}\label{AHM2}
f(t) = \sum_{l=1}^Lf_l(t)+\Phi(t),
\end{equation}
where $L\in \mathbb{N}^{+}$, $f_l=A_l(t)\cos(2\pi\varphi_l(t))$ is an IMT function, with the assumption that $\varphi_l'(t)>\varphi_{l-1}'(t)$ and the separation of $\varphi_l'(t)$ and $\varphi_{l-1}'(t)$ is uniformly bounded from below, and $\Phi(t)$ is a mean zero random process that models noise or other stochastic quantity that we have interest. Note that $\Phi$ does not need to be stationary, and we may even consider the piecewise locally stationary random process \cite{zhou2013heteroscedasticity} to more realistically model noise. See \cite{Chen_Cheng_Wu:2014} for more discussion when the noise is nonstationary. For readers with interest in the identifiability issue of the AHM, see \cite{Chen_Cheng_Wu:2014} for details and proof.
The AHM has been proved useful to model various types of time series from different fields, including instantaneous heart rate \cite{SSTAPP2}, vibration signal \cite{SSTAPP3}, and seismic signal \cite{SSTAPP6}. 

\subsection{Adaptive non-harmonic model}
The AHM has been considered in various applied problems. However, there are many oscillatory signals that cannot be {\em satisfactorily} modeled by the AHM, since their oscillatory pattern is non-sinusoidal. The ta-mECG shown in Figure \ref{fig:1} is a typical example that cannot be satisfactorily modeled by the AHM. 
The main observation that the oscillatory pattern is non-sinusoidal motivates us to replace the cosine function in the IMT function by a 1-periodic function to capture this non-sinusoidal oscillatory pattern \cite{Wu2013}:
\begin{equation}\label{ANHM} 
f(t) = A(t)s(\varphi(t)),
\end{equation}
where $A(t)$ and $\varphi(t)$ are the same as those of (\ref{AHM}), and $s(t)$ is a real $1$-periodic function. Here, $s(t)$ is called the \textit{wave-shape function}. Usually we assume that the wave-shape function has zero mean and unitary $L^2$ norm to avoid the identifiability issue. In \cite{Wu2013}, the first Fourier coefficient of $s$ is assumed to be $\hat{s}(1) \neq 0$, and this condition is later relaxed in \cite{Yang2}.
We mention that for a reasonable $s$, $f$ could be expanded by Fourier series as:
\begin{equation}\label{eq3} 
f(t) = \sum_{n=0}^{\infty}A(t)a_n\cos(2\pi n \varphi(t)+ \alpha_n),
\end{equation}
where $a_n \geq 0$ are associated with the Fourier coefficients of $s$ and $\alpha_n \in [0,2\pi)$ ($\alpha_0 = 0$). Note that the equality in \eqref{eq3} in general has to be understood in the distribution sense, and this equality is in the pointwise sense if $s$ is smooth enough. 
Therefore, instead of viewing $f$ in (\ref{ANHM}) as an oscillatory signal with single oscillatory component with non-sinusoidal oscillation, we could also view $f$ as an oscillatory signal with multiple components with sinusoidal oscillation pattern as in (\ref{AHM}). Under this interpretation, we call the first oscillatory component $A(t)a_1 \cos(2\pi \varphi(t) + \alpha_1)$ the \textit{fundamental component} and $A(t)a_n\cos(2\pi n \varphi(t)+ \alpha_n)$, $n\geq 2$ the \textit{$n$-th multiple} (or harmonic) of the 
fundamental component. We call the IF $\varphi'(t)$ of the fundamental component the {\em fundamental IF} of $f$ , and clearly, the IF of the $n$th multiple is $n$-times the fundamental frequency. 
A signal satisfies the \textit{adaptive non-harmonic model (ANHM)} if it can be written as
\begin{equation}\label{ANHM2}
f(t)=\sum_{l=1}^Lf_l(t)+\Phi,
\end{equation}
where $L\in \mathbb{N}^{+}$, $f_l=A_l(t)s_l(\varphi_l(t))$ satisfies \eqref{eq3}, and $\Phi$ is the same as that in \eqref{AHM2}. This model is also considered in \cite{HouShi}. To model the ta-mECG shown in Figure \ref{fig:1}, we have $L=2$, and $s_l$ captures the typical ``P-QRS-T pattern'' of the ECG signal. With this model, the signal processing challenge can be itemized to the following -- 
\begin{enumerate}
\item [(Q1)] how to estimate the maternal IHR $\varphi_1'(t)$ and fetal IHR $\varphi_2'(t)$ from $f(t)$? 

\item [(Q2)] how to estimate the maternal AM $A_1(t)$ and fetal AM $A_2(t)$ from $f(t)$? 

\item [(Q3)] how to estimate the maternal wave-shape dynamics $s_1$ and fetal wave-shape dynamics $s_2$ from $f(t)$? 

\item [(Q4)] may we even decompose the ta-mECG ($f(t)$) to the mECG ($f_1(t)$) and the mECG ($f_2(t)$)? This question is called the {\em single channel blind source separation} (scBSS) problem since we only have one channel.
\end{enumerate}
Note that (Q1)-(Q4) are typical questions we may ask for any other oscillatory time series, while the interpretation of these quantities depends on the background knowledge and purposes.
In this study, we focus on answering question (Q1), which serves as the initial point to answer questions (Q2), (Q3) and (Q4) \cite{su2019recovery,lin2019wave}.

\begin{remark}
We have several remarks.
\begin{enumerate}
\item We mentioned that the ta-mECG may be {\em satisfactorily} modeled by the AHM, like the expansion shown in \eqref{eq3}. The main reason about the {\em satisfaction} depends on the background knowledge. In many problems like the ECG analysis, how the signal oscillates provides abundant information about the system. Specifically, it is the oscillatory pattern of an ECG signal that a physician diagnose arrhythmia or other health issues. Therefore, for these signals, we should model how fast the signal oscillates and the oscillatory pattern separately, instead of using the AHM. 

\item The model in (\ref{ANHM}) has been further generalized to capture the ``time-varying oscillatory pattern'' in \cite{Deshape,Yang2,lin2019wave}. See \cite{Deshape,lin2019wave} for more discussion about its physiological motivation and generalization. Among these generalized models, the one proposed in \cite{lin2019wave} is designed to further quantify the wave-shape dynamics mentioned in (Q3), and it has been applied to study the ultra-long time physiological time series \cite{wang2019novel}. Since taking these generalizations into account will not make a conceptual difference, we focus on the simple ANHM model in this paper.

\item In practice we have trend in the signal that needs to be modeled \cite{Chen_Cheng_Wu:2014}. In this study, to simplify the discussion and focus on the main idea beyond the algorithm for (Q1), we focus on the ANHM considered in (\ref{ANHM2}) and assume that the trend can be removed by traditional filtering techniques. 
\end{enumerate}
\end{remark}

\subsection{Periodicity transform: old and new}

Before handling signals described in the previous section, 
we detour to review a relevant signal processing tool, the periodic transform (PT) \cite{sethares1999}, and propose a new implementation of the PT. In a nutshell, the PT is a time domain technique aiming to quantifying the oscillatory behavior of a given signal. It helps us obtain the periods of each oscillatory components. It is particularly useful when the periodicity is integer or can be well approximated by an integer \cite{iMUSIC}, and when we have a short recording \cite{PPV3} so that frequency domain or TF domain techniques are limited. 
It has been applied to study various signals, ranging from music \cite{music2001}, video \cite{video}, to DNA sequence analysis \cite{DNA1, DNA2, DNA3}.

The reason we detour to discuss PT comes from the critical observation in Figure \ref{fig:shape} in the Introduction. Recall that there is a repetition pattern in the frequency axis of spectrogram shown in Figure \ref{fig:shape}, where the period of the repetition pattern equals the fundamental frequency. This motivates us to consider the PT that aims to obtain the hidden periodicity, and hence the fundamental frequency. Specifically, as we will show in the next section, we will analyze the spectrogram at time $t$, $|V_f^{(h)}(t,\cdot)|$ shown in \eqref{Equation STFT first example}, by the PT. 

Intuitively, if we can find a good dictionary encoding the periodicity information, then by decomposing the signal via the associated space in some special way, we can determine if the signal contains an oscillatory component and decide its periodicity. Such intuition is now known as the \textit{Basis Pursuit} \cite{bp}. 
To realize this intuition, we need to find the space encoding the periodicity information. First, recall that a vector $v\in \mathbb{R}^n$ is called $p$-periodic, where $p\leq n$, if $p$ is the smallest positive integer so that $v(j+p)=v(j)$ for any $j=1,\ldots,n-p$. An intuitive space that encodes the information of period $p$ can be defined in the following way.

\begin{definition}[Periodic space \cite{sethares1999}]
Fix $n\in \mathbb{N}^{+}$ and take the period $p\leq n$. The {\em periodic subspace} of period $p$ of $\mathbb{R}^n$, denoted as
$\mathcal{P}_p$, consists of $p$-periodic vectors, denoted as $\delta_{p}^{s}$, $s = 0, 1, 2, \dots, p-1$, and defined by 
\[
\delta_{p}^{s}(j) = \begin{cases}
1, & \text{$(j-s)$ mod $p = 0$};\\
0, & \text{otherwise}.
\end{cases}
\] 
We call $\delta_{p}^{s}$ the $p$-periodic basis vector.
\end{definition}
Note that this basis of $p$-periodic subspace $\mathcal{P}_p$ is redundant in that $\mathcal{P}_{m_1p} \cap \mathcal{P}_{m_2p} = \mathcal{P}_{p}$, if $m_1$ and $m_2$ are mutually prime. This could make the estimation of periodicity fail since there is a strong correlation and overlap between subspaces $\mathcal{P}_{m_1p}$ and $\mathcal{P}_{m_2 p}$. In \cite[IV.A]{PPV3}, this basis is called the {\em natural basis}, which is a special case of {\em nested periodic subspaces}. 

To tackle this problem, Vaidyanathan \cite{PPV1,PPV2} considered the {\em Ramanujan sums} and the {\em $p$-periodic Ramanujan subspace} so that a $p$-periodic sequence can be uniquely decomposed into periodic components in the $p$-periodic Ramanujan subspace and those Ramanujan subspaces with periods being factors of $p$.  
We mention that Ramanujan sums has been applied to various problems, ranging from low frequency noise analysis \cite{Planat2002}, $1/f$ noise analysis \cite{Planat2009}, electrocardiogram signal \cite{mainardi2008analysis}, amino acid sequences \cite{mainardi2007application}, etc. Moreover, the Ramanujan subspace is also a special case of nested periodic subspaces, which contains several nice properties \cite{PPV3} compared with other nested periodic subspaces. For more theoretical details about nested periodic subspaces and its relationship with several other existing work, we refer readers with interest to \cite{unified}. We now briefly review the theory of Ramanujan sums and Ramanujan subspaces. More details about Ramanujan sums and Ramanujan subspaces and historical notes can be found in \cite{PPV1, PPV2, PPV3} and the citation therein. 

\begin{definition}[Euler totient function]
The {\em Euler totient function} $\phi(n)$, where $n\in\mathbb{N}^{+}$, is defined as the number of positive integers $m\in\mathbb{N}^{+}$, $m\leq n $, such that $m$ is coprime to $n$; that is, $(m,n) = 1$. The {\em totient summation function} is defined as 
\[
\Phi(n) := \sum_{i=1}^{n}\phi(i).
\]
\end{definition}

\begin{definition}[Ramanujan sum]
For $p\in \mathbb{N}^{+}$, the $p$th Ramanujan sum is defined as 
\[
c_{p}(n) = \sum_{\substack{{k=1}\\ {(k,p)=1}}}^{p} e^{\frac{2\pi i k n}{p}}, 
\]
where $n\in \mathbb{Z}$.
\end{definition}
Note that there are $\phi(p)$ terms in the summation of $c_{p}(n)$  for all $n\in \mathbb{Z}$. Also, it is clear from the definition that $c_{p}(n) = c_{p}(n+p)$ for all $n\in \mathbb{Z}$. We list some useful properties of Ramanujan sums here: 

\begin{proposition}[Some properties of the Ramanujan sums] \label{Proposition properties RS}
The Ramanujan sums satisfy the following properties:
\begin{enumerate}
\item \cite[Corollary 1]{PPV1} $c_p(n)$ is an integer and $c_p(n)\leq\phi(p)$ for any $p$ and $n$ ; 

\item \cite[Equation (9)]{PPV1} $c_p(n) = c_p(-n)$ for any $p$ and $n$ ; 

\item (Autocorrelation) \cite[Theorem 7]{PPV1} $\sum\limits_{n=0}^{p-1}c_p(n)c_p(n-l) = p c_p(l)$ for any $l\in \mathbb{Z}$;

\item (Orthogonality) \cite[Equation (15)]{PPV1} $\sum\limits_{n=0}^{lcm(p,q)-1}c_p(n)c_q(n-l) = 0$  for any $l\in \mathbb{Z}$ when $p\neq q$.
\end{enumerate}
\end{proposition}

\begin{definition}
Define the $k$-th circular shift sequence of $c_p(n)$ as $c_p^{(k)}(n)$:
\[
c_{p}^{(k)}(n) := c_{p}(n-k),
\]
where $n, k\in \mathbb{Z}$.
\end{definition}

To detect a $p$-periodic component in a given signal, we consider the following $p\times p$ circulant matrix:
\begin{align*}
B_{p} 
 := &\begin{bmatrix}
c_{p}(0) & c_{p}(p-1) & c_{p}(p-2) & \cdots & c_{p}(1)\\
c_{p}(1) & c_{p}(0) & c_{p}(p-1) & \cdots & c_{p}(2)\\
c_{p}(2) & c_{p}(1) & c_{p}(0) & \cdots & c_{p}(3)\\
\vdots & \vdots & \vdots & \ddots & \vdots\\
c_{p}(p-1) & c_{p}(p-2) & c_{p}(p-3) & \cdots & c_{p}(0)
\end{bmatrix} \\
=&\begin{bmatrix}
c_p^{(0)}(0) & c_p^{(1)}(0) & c_p^{(2)}(0) & \cdots & c_p^{(p-1)}(0) \\
c_p^{(0)}(1) & c_p^{(1)}(1) & c_p^{(2)}(1) & \cdots & c_p^{(p-1)}(1) \\
\vdots & \vdots & \vdots & \ddots & \vdots \\
c_p^{(0)}(p-1) & c_p^{(1)}(p-1) & c_p^{(2)}(p-1) & \cdots & c_p^{(p-1)}(p-1) \\
\end{bmatrix}.
\end{align*}
Clearly $B_p$ is a symmetric matrix due to Proposition \ref{Proposition properties RS} (2).

\begin{proposition}[Properties of circulant matrix \cite{PPV1}]\label{prop2}
The circulant matrix $B_p\in\mathbb{Z}^{p\times p}$ has some nice properties:
\begin{enumerate}
\item  $\text{rank}(B_p) = \phi(p)$ 

\item any $\phi(p)$ consecutive columns of $B_p$ are linearly independent. 

\item the period of any vector in the column space of $B_p$ is exactly $p$ and cannot be smaller. 

\item The column space of $B_p$ is the same as that of the subspace spanned by $\phi(p)$ Fourier columns: 
\[  
\Big\{\begin{bmatrix}
1 \\
\omega_p^{k_i} \\
\omega_p^{2k_i} \\
\vdots \\
\omega_p^{(p-1)k_i}
\end{bmatrix}
\Big\}_{i},
\quad  (k_i,p)=1,\quad 1\leq k_i\leq p,
\]
where $\omega_p = e^{\frac{-2\pi i}{p}}$.
\end{enumerate}
\end{proposition}

Proposition \ref{prop2} reminds us the usual definition of the discrete Fourier transform (DFT) matrix, which we remind the readers below. 

\begin{definition} The $p\times p$ DFT matrix is defined as 
\[
W_p =
\begin{bmatrix}
1 & 1 & 1 & \cdots & 1\\
1 & \omega_p^{1} & \omega_p^{2} & \cdots & \omega_p^{p-1}\\
1 & \omega_p^{2} & \omega_p^{2\times 2} & \cdots &   \omega_p^{2(p-1)}\\
1 & \vdots & \vdots & \ddots & \vdots\\
1 & \omega_p^{(p-1)} & \omega_p^{(p-1)\times 2} & \cdots & \omega_p^{(p-1)(p-1)}
\end{bmatrix}.
\]
	
\end{definition}

\begin{definition}
We need the following notations.
\begin{enumerate}
\item Denote the matrix formed by the first $\phi(p)$ columns of $B_p$ by $C_p$. 

\item Fix $N\in \mathbb{N}^{+}$. Let $C_{p,N}$ be an $N\times\phi(p)$ matrix defined as
\[
C_{p,N} = 
\begin{bmatrix}
C_p \\
\vdots \\
C_p \\
R
\end{bmatrix},
\]
where $R$ is the first $N-p\floor*{\frac{N}{p}}$ rows of $C_p$ if $N -p \floor*{\frac{N}{p}}\neq 0$, and $R$ is not needed if $N -p \floor*{\frac{N}{p}}= 0$. 

\item Denote the columns of $C_{p,N}$ by $c_{p,N}^{(0)}, c_{p,N}^{(1)}, \cdots, c_{p,N}^{(\phi(p)-1)}$ accordingly.
\end{enumerate} 
\end{definition}

Note that by a direct calculation, we have $B_p = V_p V_p^{*}$, where $V_p$ is a $p\times\phi(p)$ submatrix of $W_p$ whose columns are those $\phi(p)$ Fourier columns shown in Proposition \ref{prop2} (4) \cite{PPV1}. By (2) of Proposition \ref{prop2}, $C_p$ has full column rank and has the same column space as $V_p$. Moreover, $C_p^T C_p\in \mathbb{Z}^{\phi(p)\times\phi(p)}$ is positive definite.
$C_{p,N}$ is designed to capture the $p$-periodic component inside a signal of length $N$ by periodically extending $C_p$ to length $N$ with $C_p$ repeated for $\floor*{\frac{N}{p}}$ times padded with $R$ if needed.
Note that $C_{p,N}^T C_{p,N}$ is always positive definite when $N\geq p$.

\begin{definition}[Ramanujan subspace \cite{PPV1,PPV3}]
	The {\em $p$-periodic Ramanujan subspace} of length $N$, denoted by $\mathcal{R}_{p,N}$, is defined to be the column space $C_{p,N}$; that is, the span of $\{c_{p,N}^{(0)}, c_{p,N}^{(1)}, \cdots, c_{p,N}^{(\phi(p)-1)}\}$.
\end{definition} 

By construction, it is clear that the $p$-periodic Ramanujan subspace $\mathcal{R}_{p,N}$ is formed by fewer basis vectors than $\mathcal{P}_{p}$. We have the following important property of the Ramanujan subspace by (4) in Proposition \ref{prop2}.
\begin{proposition}[\cite{PPV1, PPV2}]
$\mathcal{R}_{p,N}$ and $\mathcal{R}_{q,N}$ are ``asymptotically orthogonal'' for $p\neq q$ in the following sense: for any sequence $x(n)\in\mathcal{R}_{p,N}$ and $y(n)\in\mathcal{R}_{q,N}$ we unifomly have
\[
\lim_{N\to \infty}\frac{1}{N}\sum_{n=0}^{N-1}x(n)y(n)=0.
\]
\end{proposition}

To investigate a $p$-periodic signal, it is sufficient to look at the case when $N=p$,
\begin{definition}[\cite{PPV1,PPV3}]\label{Definition Fp}
\[ 
F_{p} := 
\begin{bmatrix}
C_{p_1,p} & C_{p_2,p} & \cdots & C_{p_n,p}
\end{bmatrix},
\]
where $p_i$'s are divisors of $p$ with $1 = p_1<p_2<\cdots<p_n = p$.
\end{definition}

By a well known property that $\sum\limits_{p_i | p} \phi(p_i) = p$, we know that $F_p$ is a $p\times p$ matrix. 
 
\begin{proposition}[\cite{PPV1,PPV3}]\label{prop4}
We have that $\text{rank}(F_p) = p$.
Moreover, $F_p$ has the same column space as the $p\times p$ DFT matrix. 
\end{proposition}

Thus, any signal with period $p$ can be uniquely expressed by a linear combination of columns of $F_p$. In other words, any $p$-periodic signal can be uniquely decomposed into Ramanujan subspaces corresponding to period $p$ and its divisors. Therefore, it is likely to observe periods $p_i$ that are divisors of $p$, when we decompose a $p$-periodic signal. Moreover, by the orthogonality property of the Ramanujan sums, we immediately have:
\begin{corollary}
Following the notation in Definition \ref{Definition Fp}, when $N$ is a multiple of $p$, the column spaces associated with $C_{p_1,N}$, $C_{p_2,N}$, $\dots$, $C_{p_n,N}$ are mutually orthogonal.
\end{corollary}

\subsubsection{Existing algorithms for PT}

There have been many algorithms proposed for the PT. Those algorithms can be roughly classified into two categories. The first one is projecting the signal into the prescribed dictionary, and the second one includes those greedy-based algorithm \cite{sethares1999}, such as Small to Large algorithm, M-best algorithm, Best Correlation algorithm. In \cite{sethares1999}, the natural periodic subspaces $\mathcal{P}_p$ were chosen. The same algorithm can also be adapted to the setup of Ramanujan subspaces $\mathcal{R}_{p,N}$. 

For the comparison purpose in the numerical section, we summarize these algorithms in terms of $\mathcal{R}_{p,N}$. 
Take a signal $x\in \mathbb{R}^N$ and a predetermined longest periods $P_{max}$.
Given a predetermined threshold $T$, the Small to Large algorithm iteratively checks the existence of periods, starting from the shortest period $p=1$, and finish up to the maximal period $P_{max}$. During the iteration, if there exists $p\geq 1$ so that the orthogonal projection of $x$ onto $\mathcal{R}_{p,N}$ has magnitude greater than $T$ times the magnitude of $x$, we remove that component from $x$, get the residue $r$, and continue the algorithm with $r$. Note that $p=1$ is related to the constant trend. 
To run the M-best algorithm, we need a prior estimate of the number of periods, $M$, that constitute $x$. The algorithm maintains a list of the $M$ best periods and the corresponding basis elements. When a new (sub)period is detected that removes more energy from the signal than one currently on the list, the new one replaces the old, and the algorithm iterates until convergence.
The Best Correlation algorithm also needs a predetermined number of periods $M$. This algorithm projects $x$ onto all $\mathcal{R}_{1,N}, \dots, \mathcal{R}_{P_{max},N}$, essentially measuring the correlation between $x$ and the individual periodic basis elements. The period $p$ with the largest correlation is then used for the projection. The selection of $M$ is critical for the M-best algorithm and Best Correlation algorithm.

Vaidyanathan et al. also proposed to use the $l^2$ norm minimization to replace the $l^1$ norm in (\ref{Optimization setup 1}) \cite{dictionary, PPV3, unified}. In \cite{iMUSIC}, the {\em integer MUltiple SIgnal Classification} (iMUSIC) algorithm is proposed, which follows the idea of minimizing the $l^2$ norm of the correlation vector of the periodic component with the noise eigenspace of the sample autocorrelation matrix of $x$ \cite{iMUSIC}. However, these methods cannot guarantee a sparse estimation and spurious periods are possible. This might impede identifying the underlying periods.

\subsubsection{Our proposed algorithm for PT}

We now describe our approach of PT. Our approach is based on the following observation and assumption. 
Since many signals of practical interest consist of few periodic components, the {\em ideal} output of the PT encoding the underlying periodicity of the signal should be sparse. 
This sparsity requirement motivates us to implement the PT as a $l_1$-norm minimization problem. This sparsity idea has been mentioned in \cite{dictionary, PPV3, unified,SuLi}, while our implementation is different from what have been proposed. Take a signal $y\in\mathbb{R}^{N}$. Fix a predetermined upper bound of estimated periods, denoted as $P_{max}\in \mathbb{N}$. Consider
\begin{equation}\label{Optimization setup 1}
\min\norm{Dx}_1 \quad s.t. \quad y=Ax,
\end{equation}
where $A$ is an $N\times \Phi(P_{max})$ dictionary defined as
\[
A := \begin{bmatrix}
C_{1,N} & C_{2,N} & \cdots & C_{P_{max},N}
\end{bmatrix},
\]
$D$ is a diagonal penalty matrix whose $i$-th diagonal entry is $\zeta(P_i)$, where $P_i$ is the period associated with the $i$-th column of $A$, and $\zeta$ is the chosen {\em positive} function describing how the optimization is penalized. We call $\zeta$ the {\em periodicity penalization function}.
In \cite[(32)]{PPV3}, due to the inclination of PT in favor of long period via the dictionary approach, the quadratic penalization, $\zeta(x)=x^2$, is suggested.
In other words, $D$ suppresses high-period components since the components with longer periods are penalized more. 
In practice, if we have any prior knowledge about the period range, we can design a suitable periodicity penalization function. 
How to incorporate the prior knowledge depend on the application, and it is out of the scope of this work. 
Note (\ref{Optimization setup 1}) is similar to Basis Pursuit proposed in \cite{bp}.

Next, we shall relax the condition $y=Ax$, particularly when the signal is corrupted by noise. So, instead of considering \eqref{Optimization setup 1}, we focus on the following program
\begin{equation}\label{lasso}
 \bar{x} \in \argmin\limits_{x\in \mathbb{R}^{\Phi(P_{max})}} \frac{1}{2}\norm{y - Bx }_2^2 + \lambda \norm{ x }_1,
\end{equation}                                                                                                                                         
where $B := AD^{-1}$ and $\lambda>0$ is a tuning parameter. This is the $l_1$ penalized linear regression, also known as \textit{Basis Pursuit Denoising} (BPDN) or Lasso, which was originally proposed in \cite{tibshirani1996, bp}. This type of regression has been extensively studied, for example, in \cite{zhaoyu, wainwright, tropp, tibshirani2013lasso}. We call the proposed algorithm for PT the {\em Ramanujan PT} (RPT). Note that a similar program involving the Ramanujan sums has also been experimented in \cite{SuLi} for music analysis with satisfactory results but with a more redundant dictionary and a different penalty matrix. 
%

Recall that $z\in \mathbb{R}^{\Phi(P_{max})}$ is a subgradient of the $l_1$-norm $\norm{\cdot}_1$ at $x\in \mathbb{R}^{\Phi(P_{max})}$, $i.e.$, $z\in\partial\norm{x}_1$, if 
\[
  z_i=\left\{\begin{array}{ll}
               1, & x_i >0\,;\\
               -1, & x_i <0\,;\\
               \in [-1, 1], & x_i =0\,,
            \end{array}\right.
\]
where $i=1,\ldots,\Phi(P_{max})$.
A subgradient for the function (\ref{lasso}) being minimized is therefore $-B^{T}(y-Bx)+\lambda z$ for some $z\in\partial\norm{x}_1$. Moreover,  $\bar{x}\in \mathbb{R}^{\Phi(P_{max})}$ is a solution to (\ref{lasso}) if and only if there exists a $\bar{z}\in\partial\norm{\bar{x}}_1$, such that
\begin{equation}\label{characterization}
B^T B\bar{x} - B^T y + \lambda \bar{z} = 0.
\end{equation}
We know from \cite[Lemma 1]{tibshirani2013lasso} that if $x_1$ and $x_2$ are both solutions to (\ref{lasso}), then $Bx_1 = Bx_2$. Therefore, if $x_{1,i} = 0$ for $i\in G \subseteq \{1,2, \dots, \Phi(P_{max})\}$ and there exists a $z_1 \in \partial \norm{x_1}_1$ such that $|z_{1,i}|<1$ for $i\in G$, then $z_1 \in \partial \norm{x_2}_1$ by (\ref{characterization}), and we have $x_{2,i} = 0$ for all $i\in G$.

Since each Ramanujan subspace of period $p$ has $\phi(p)$ basis vectors, we have the following definitions.
The first definition is from \cite[equation (22)]{PPV3},
\begin{definition}[Energy of period]
The energy of period of a vector $x\in \mathbb{R}^{\Phi(P_{max})}$, denoted as $EOP_x\in\mathbb{R}^{P_{max}}$, is defined as 
\[
EOP_x(p) = \sum_{i=\Phi(p-1)+1}^{\Phi(p)}x_{i}^{2}\,.  
\]
\end{definition}
Essentially, $EOP_x(p)$ summarizes how strong the signal $x$ oscillates with the period $p$, since the $(\Phi(p-1)+1)$-th to the $\Phi(p)$-th bases captures the component oscillating with the period $p$. In this paper, we also propose the following quantities for the upcoming algorithm and analysis.
\begin{definition}[Intrinsic period]
The intrinsic period (IP) of a vector $x\in \mathbb{R}^{\Phi(P_{max})}$ is a set defined as
\[
IP_x = \{ p: EOP_x(p) > 0\}. 
\]
\end{definition}
We define the IP so that the solution $x$ to (\ref{lasso}) tells us that the signal $y$ contains an oscillatory component with $p$-periodicity if  $p\in IP_x$. We further define their corresponding indices in the dictionary
\begin{definition}[Support of periodicity]
The support of periodicity (SOP) of a vector $x\in \mathbb{R}^{\Phi(P_{max})}$ is defined as
\[
SOP_x = \bigcup\limits_{p\in IP_x } \{ i: \Phi(p-1)+1 \leq i \leq\Phi(p)\}.
\]
\end{definition}

Finally, we mention that when we have multiple realizations of the signal we have interest, we can easily generalize \eqref{lasso} to take all realizations into account. Suppose we have $k$ signals, $y_1,\ldots,y_k\in \mathbb{R}^N$ with the same periodic components but different noise realizations, we can estimate the periodicity by
\[
\bar{x}\in \argmin\limits_{x\in \mathbb{R}^{\Phi(P_{max})}} \frac{1}{2}\norm{Y- Bx\mathbf{1}^T }_{F}^2 + \lambda \norm{ x }_1\,,
\]
where $\mathbf{1}$ is a $k$-dim vector with all entries $1$ and $Y=[y_1,\ldots,y_k]\in \mathbb{R}^{N\times k}$. We call this generalized algorithm vectorized RPT (vRPT).

\section{Proposed time-frequency analysis algorithm}\label{Section Proposed time-frequency analysis algorithm}

In this section, we summarize the existing de-shape algorithm and discuss its limitation. Then we introduce our proposed RDS and  vRDS algorithms by taking the PT into account.

\subsection{A critical observation of spectrogram}\label{subsection:critical observation of STFT}

We start from recalling the STFT. Given a window function $h$, such as a Gaussian function centered at the origin, the STFT of $f$ is defined as:
\[
V_f^{(h)}(t,\xi) = \int_{\mathbb{R}}f(\tau)h(\tau-t)e^{-2\pi i \xi\tau} \diff\tau,
\]
where $t\in \mathbb{R}$ and $\xi\in\mathbb{R}$ indicate time and frequency respectively. We call $V_f^{(h)}:\mathbb{R}^2\to\mathbb{C}$ a TFR of $f$. Usually, people call $|V_f^{(h)}(t,\xi)|^2$ the {\em spectrogram} of $f$, which is another TFR. The TFR depicts the spectral dynamics of the signal as time $t$ evolves. 

If we apply the STFT to analyze an oscillatory signal with non-sinusoidal oscillatory pattern modeled by \eqref{ANHM}, a critical observation is that there exists an oscillatory pattern in the frequency domain. A specific example is shown in \eqref{Equation STFT first example}, and see Figure \ref{fig:shape} for an example of a simulated signal, and Figure \ref{fig:2} for an example of the fetal ECG signal shown in Figure \ref{fig:1}. 
Recall that this pattern comes from the Fourier series expansion of the non-sinusoidal wave-shape function shown in \eqref{eq3}.
As is discussed in the Introduction, it is not easy to extract useful information from the TFR, even when we know the ground truth, not to say if we are blind to the ground truth.

\begin{figure}[!htbp]
	\includegraphics[trim=20 0 40 30,clip,width=0.49\textwidth]{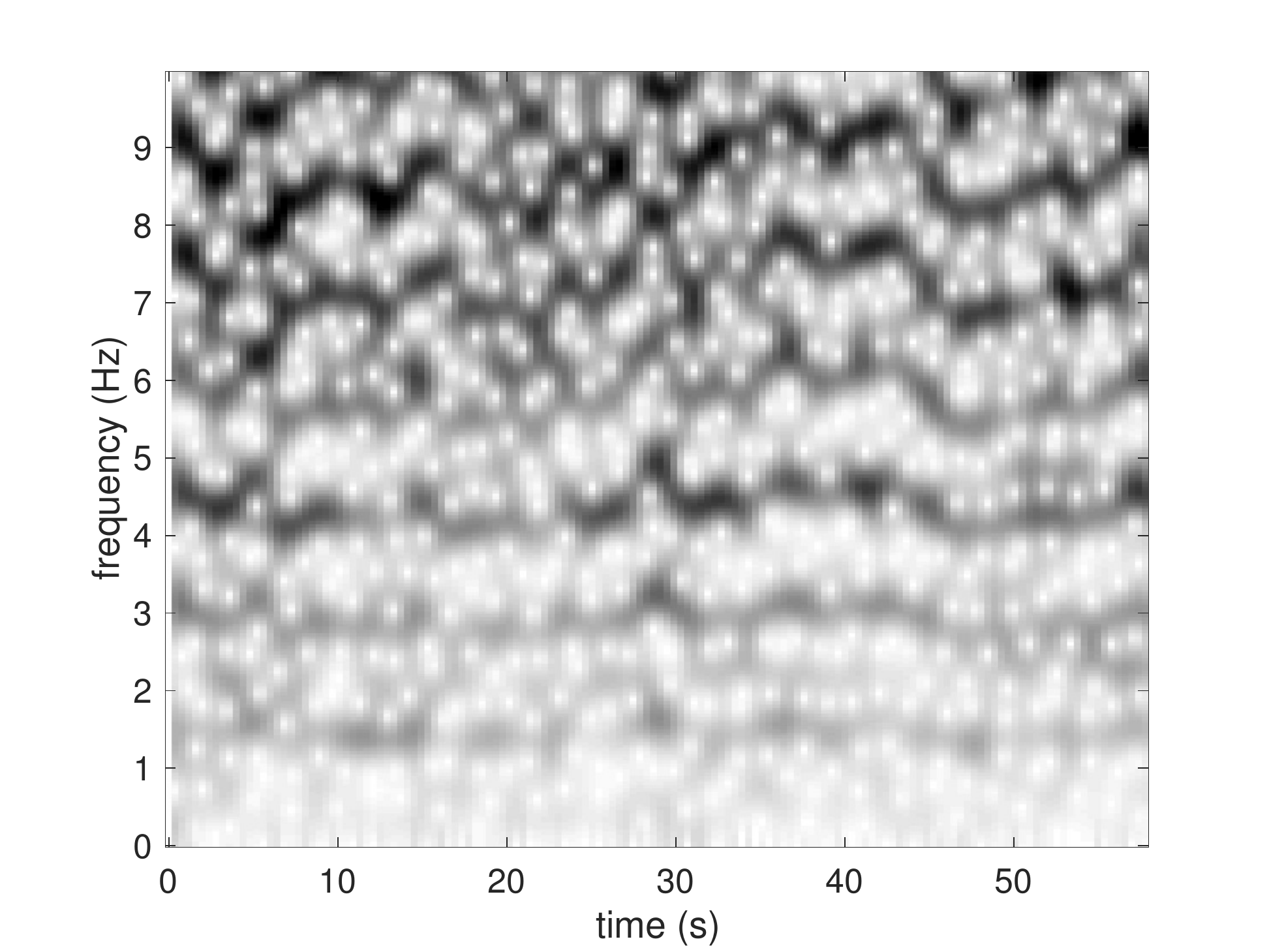}
	\includegraphics[trim=20 0 40 30,clip,width=0.49\textwidth]{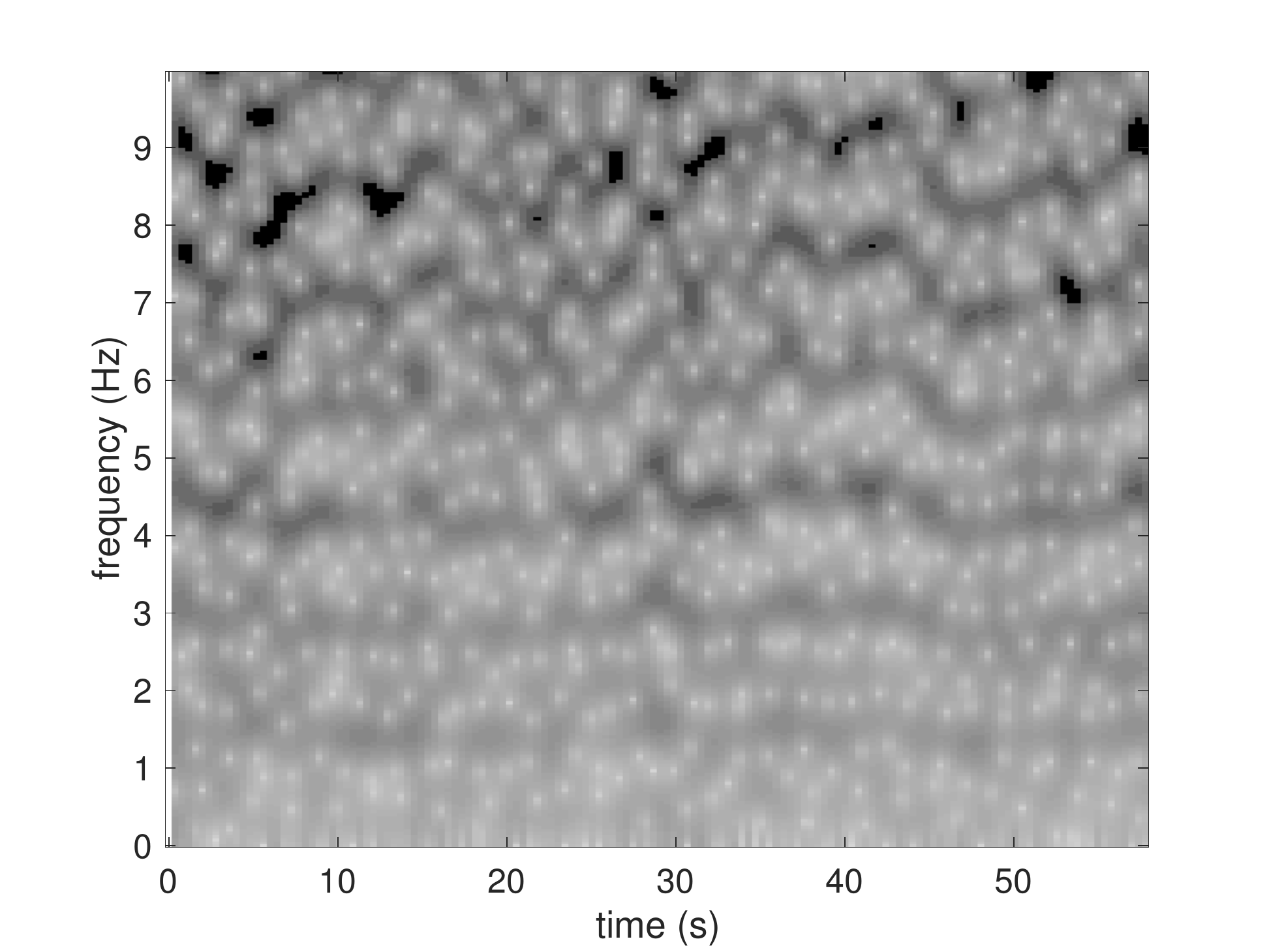}\\
	\begin{minipage}{19.5 cm}
	\hspace{-90pt}\includegraphics[trim=20 0 40 30,clip,width=0.32\textwidth]{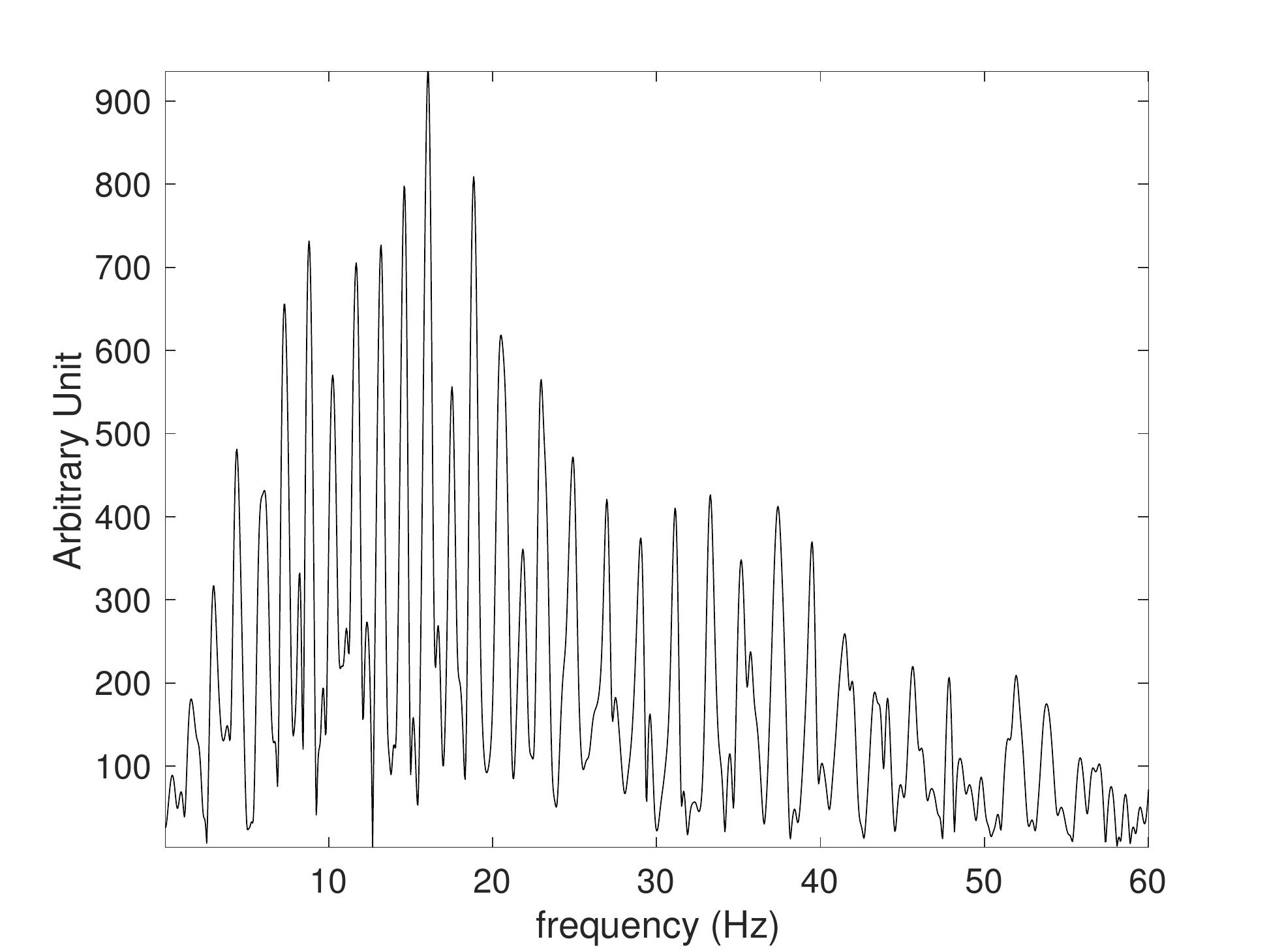}
	\includegraphics[trim=20 0 40 30,clip,width=0.32\textwidth]{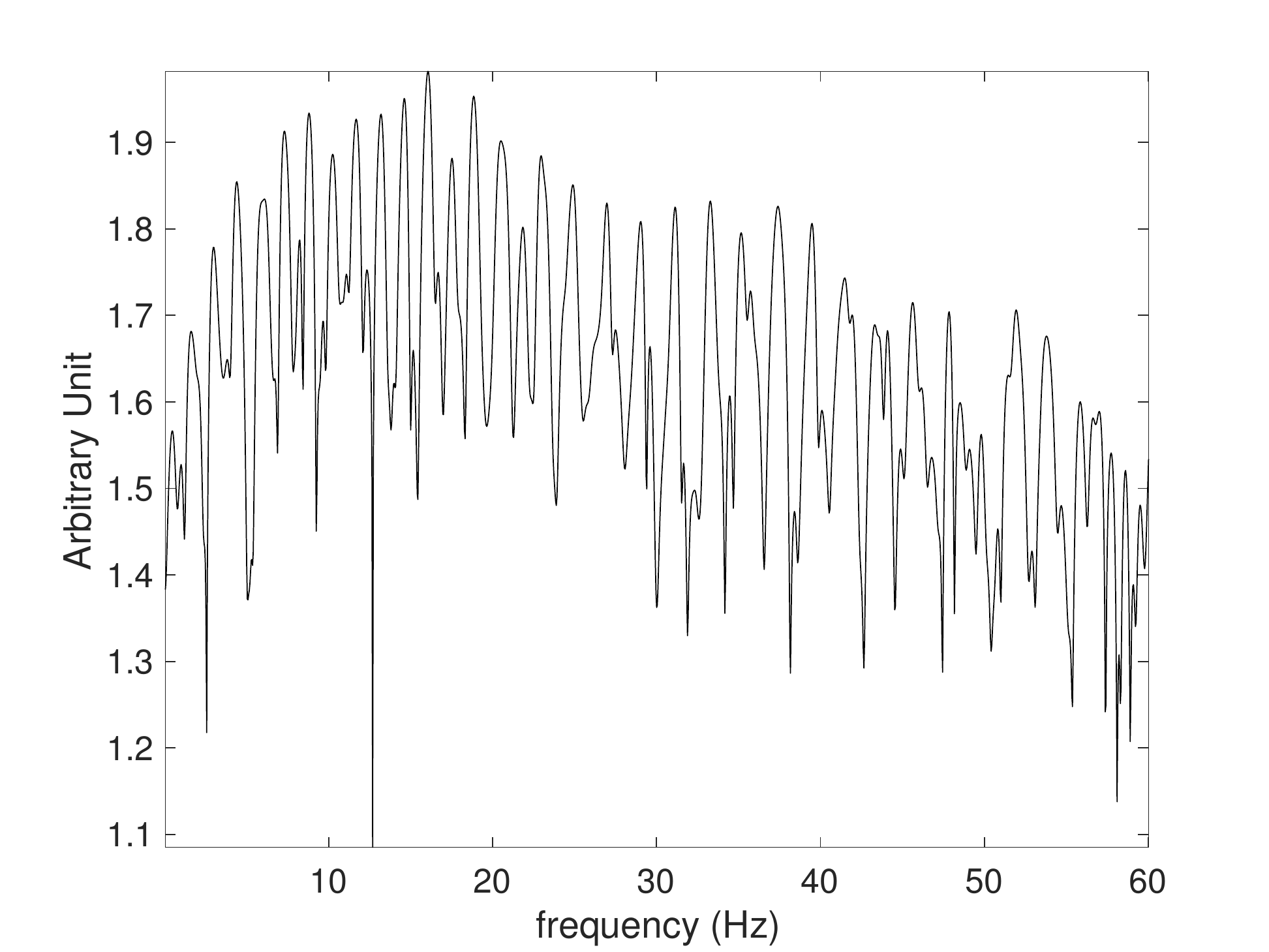}
	\includegraphics[trim=20 0 40 30,clip,width=0.32\textwidth]{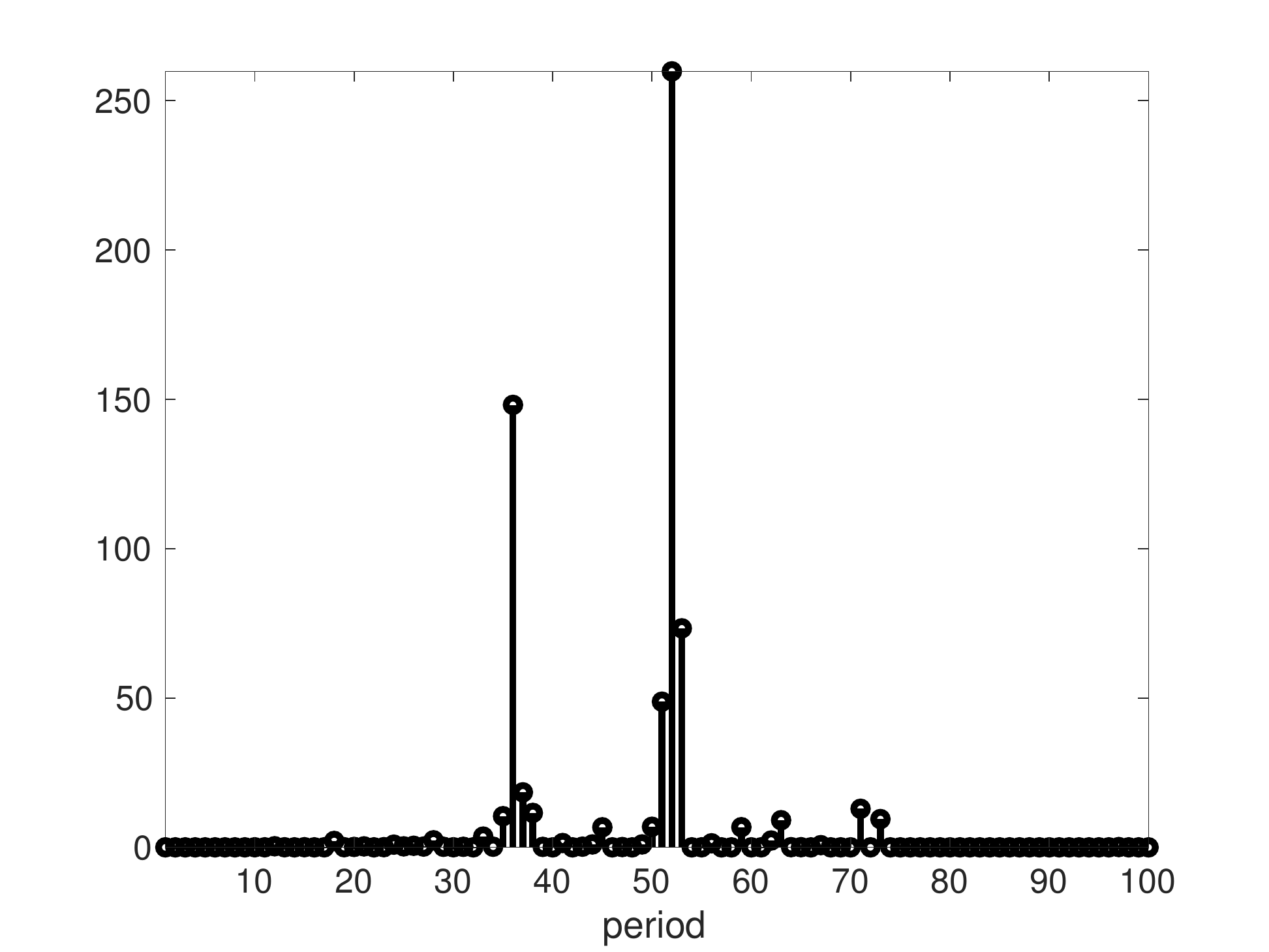}
	\end{minipage}
	\caption{Top left: the magnitude of the STFT of the ta-mECG signal $f$ shown in Figure \ref{fig:1}, $|V_f^{(h)}(t,\xi)|$. The blue arrows indicate the fundamental frequency and its multiples of the mECG, while the red arrows indicate the fundamental frequency and its multiples of the fECG. It is clear that the information associated with the fECG is weak and hard to identify. Moreover, the spectral information of the mECG and fECG is mixed up, which makes it more challenging to proceed. Top right: $|V_f^{(h)}(t,\xi)|^{0.1}$. The line associated the fECG indicated by read arrows are more visible, but still challenging. Middle left: $|V_f^{(h)}(t_0,\cdot)|$, where $t_0 = 24.24$s. Middle middle: $|V_f^{(h)}(t_0,\cdot)|^{0.1}$. Clearly, the envelope is ``flattened'' after taking a fractional power, as is discussed in the main article. Bottom right: result of the proposed RPT on $|V_f^{(h)}(t_0,\cdot)|^{0.1}$, where we see two dominant periods predicted by RPT. Since the frequency bin is $0.04$Hz, these two periods are associated with $35\times 0.04=1.4$ Hz and $52\times 0.04=2.08$ Hz.} 
	\label{fig:2}
\end{figure}

\subsection{Existing approach -- De-shape algorithm}\label{Section: dsAlgorithm Intro}

To handle the challenge shown above, Lin etc. \cite{Deshape} proposed the de-shape algorithm to eliminate the impact of non-sinusoidal oscillation of $f$ satisfying (\ref{ANHM}). 
In a nutshell, by modifying the {\em cepstrum} idea \cite{cepreview} and incorporating it into the STFT, the de-shape algorithm modifies the STFT so that the resulting TFR looks like if the input signal oscillates with a sinusoidal oscillation. Below, we recall the cepstrum and summarize the de-shape algorithm.

Cepstrum is a widely-used technique that has been applied to various signal processing problems, such as pitch detection, deconvolution and speech recognition since its invention in 1963 \cite{cepstrum}. The cepstrum $f^C$ for a proper signal $f$ is defined as 
\[
 f^C(q) = \int_{\mathbb{R}} \log|\hat{f}(\xi)| e^{2\pi i q\xi} \diff\xi,
\]
where $\hat{f}(\xi)$ is the Fourier transform of $f$ and $q\in \mathbb{R}$ is called \textit{quefrency}, whose unit is the same as the unit of the original signal $f$. We refer interested readers to \cite{cepreview} for a review of cepstrum.

Clearly, like Fourier transform, cepstrum is a global operator, and the local dynamics cannot be directly captured by the cepstrum. Thus, in \cite{Deshape}, the cepstrum is generalized to the \textit{short-time cepstral transform (STCT)}, which is defined as 
\[
C_f^{(h)}(t,q) = \int_{\mathbb{R}}\log|V_f^{(h)}(t,\xi)| e^{2\pi i q\xi} \diff\xi\,.
\]
Since taking logarithm transforms the multiplication operation to the addition operation, taking logarithm decouples the amplitude modulation from the oscillation in $\log|V_f^{(h)}(t,\cdot)|$. Here, the amplitude modulation in $|V_f^{(h)}(t,\cdot)|$ comes from the Fourier series coefficients. 
To elaborate this important fact, take \eqref{Equation STFT first example} as an example. 
Assume $\texttt{supp}\hat{g}$ in \eqref{Equation STFT first example} is sufficiently small so that $\texttt{supp}\hat{g}(\xi- n\xi_0)\cap\texttt{supp}\hat{g}(\xi- m\xi_0)=\emptyset$ when $n\neq m$. Thus, 
\[
|V_f^{(g)}(t,\xi)|=Aa_0|\hat{g}(\xi)|+\frac{A}{2}\sum_{n\in \mathbb{Z},n\neq 0} |a_n| |\hat{g}(\xi- n\xi_0)|\,.
\] 
This can be understood as
\[
|V_f^{(g)}(t,\xi)|\approx E(\xi)\Big[\sum_{n\in \mathbb{Z}}|\hat{g}(\xi- n\xi_0)|\Big],
\]
where $E$ is a ``low frequency'' positive function so that $E(n\xi_0)=\frac{A|a_n|}{2}$ for $n\neq 0$ and $E(0)=A|a_0|$. We thus call $E(\xi)$ the {\em spectral envelope} of the wave-shape function. A rigorous treatment of this argument can be found in \cite{Deshape}, particularly when $f$ satisfies the ANHM.
As a result, $\log|V_f^{(h)}(t,\xi)|$ becomes a summation of the logarithm of the spectral envelope $E(\xi)$, which oscillates slowly, and the logarithm of $\sum_{n\in \mathbb{Z}}|\hat{g}(\xi- n\xi_0)|$, which oscillates fast. In other words, $\log E(\xi)$ is in the low-quefrency domain and $\log \big[\sum_{n\in \mathbb{Z}}|\hat{g}(\xi- n\xi_0)|\big]$ is in the high-quefrency domain. However, since taking the natural logarithm might be unstable numerically, the $\gamma$-power of $|V_f^{(h)}(t,\cdot)|$ is introduced to approximate the logarithm, where $\gamma>0$ is a small constant, so that we can compute it in a more numerically stable way \cite{Deshape, cepstrum1, cepstrum2, cepstrum3}; that is, we consider
\[
C_f^{(h,\gamma)}(t,q) = \int_{\mathbb{R}}|V_f^{(h)}(t,\xi)|^\gamma e^{2\pi i q\xi} \diff\xi\,,
\]
where $\gamma>0$ is a small constant chosen by the user.
As a result, if we ignore the low-quefrency content, we would obtain the fundamental IP (the inverse of the fundamental IF) and its multiples in $C_f^{(h,\gamma)}(t,\cdot)$. Note that in general the Fourier transform of $|V_f^{(h)}(t,\xi)|^\gamma$ should be understood in the distribution sense. See \cite{Deshape} for mathematical details.

Next, in order to extract the fundamental IF, we construct the \textit{inverse short-time cepstral transform (iSTCT)} defined as
\[ 
U_f^{(h,\gamma)}(t,\xi) = C_f^{(h,\gamma)}(t,1/ \xi),
\]
where the unit of $\xi$ is in the frequency domain. The main motivation of this conversion comes from the relationship between period and frequency -- the inverse of period is frequency. Hence, the inverse of the $k$-th multiple of the fundamental IP is the fundamental IF divided by $k$. Specifically, at time $t$, if there are peaks around $\xi_0$, $2\xi_0$, $3\xi_0$, etc, in $|V_f^{(h)}(t,\cdot)|^\gamma$, then we would observe peaks around $1/ \xi_0$, $2/ \xi_0$, $3/ \xi_0$, etc, in $C_f^{(h,\gamma)}(t,\cdot)$, and hence $ \xi_0$, $\xi_0/2$, $\xi_0/3$, etc, in $U_f^{(h,\gamma)}(t,\cdot)$.
%
Consequently, the information shared by the STFT and the iSTCT is the fundamental IF. Thus, we can use $U_f^{(h,\gamma)}(t,\xi)$ as a mask for the spectrogram. Motivated by this fact, the de-shape STFT is defined as
\begin{equation}\label{Definition deshape STFT}
W_f^{(h,\gamma)}(t,\xi) = V_f^{(h)}(t,\xi)U_f^{(h,\gamma)}(t,\xi),
\end{equation}
and the final TFR containing only the  fundamental frequency is given by $|W_f^{(h,\gamma)}(t,\xi)|^2$. Note that in general \eqref{Definition deshape STFT} should be understood in the distribution sense. For mathematical details, we refer readers with interest to \cite{Deshape}.

Details of numerically implementation of de-shape STFT can be found in \cite{recycling,deshapeapp}. 
Applications of the de-shape algorithm can be found, for example, in \cite{cicone2017nonlinear,recycling, deshapeapp,lin2019unexpected,li2019non}. 
We mention that based on the STFT of $f$, one may further apply the sychrosqueezing transform \cite{SST} to sharpen the TFR $W_f^{(h,\gamma)}(t,\xi)$ or even concentration of frequency and time (ConceFT) \cite{ConceFT} to alleviate the noise impact. We skip these steps to simplify the discussion.

\begin{figure}[!htbp]
	\includegraphics[trim=20 0 50 0,clip,width=.49\textwidth]{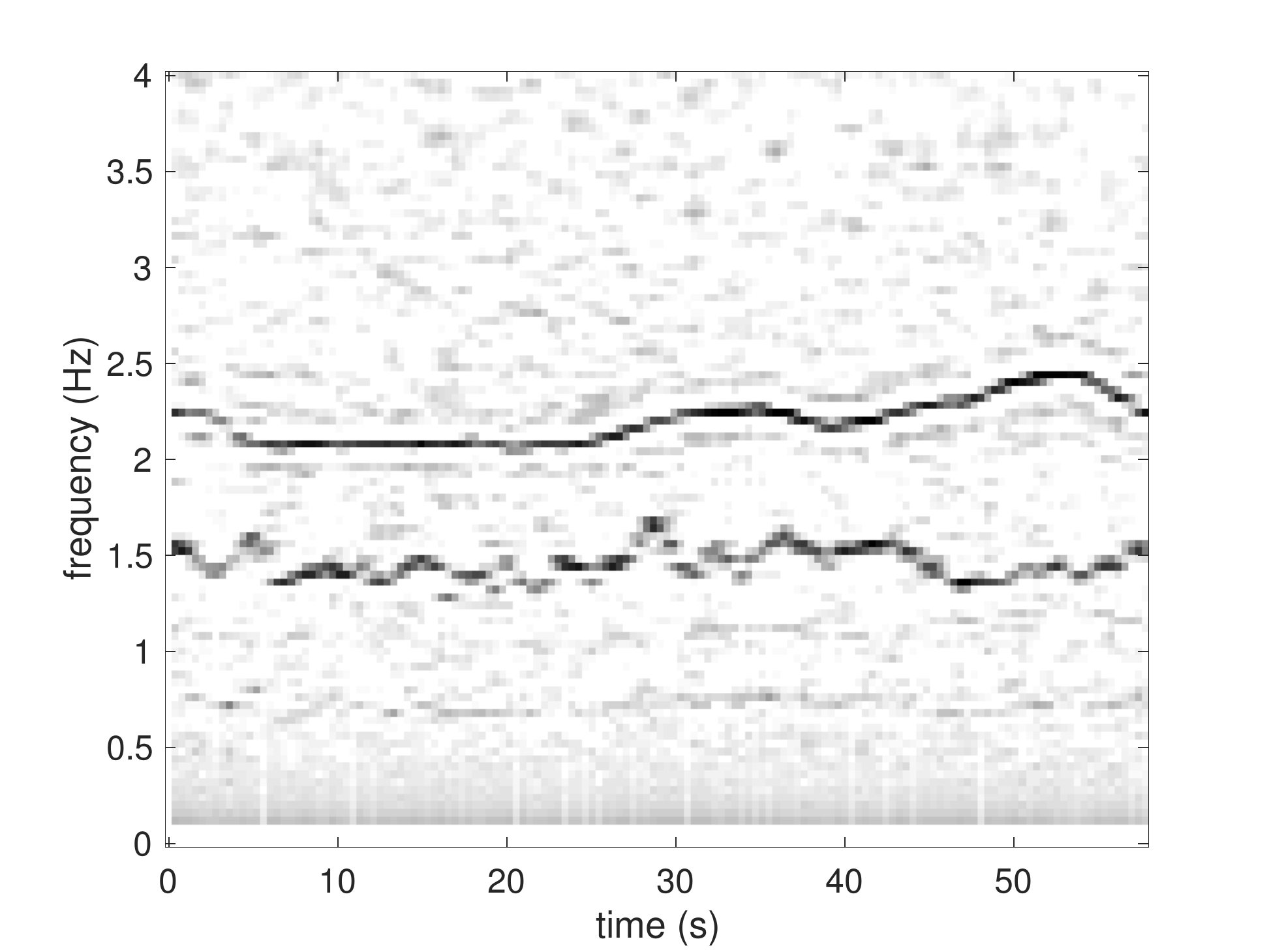}
	\includegraphics[trim=20 0 50 0,clip,width=.49\textwidth]{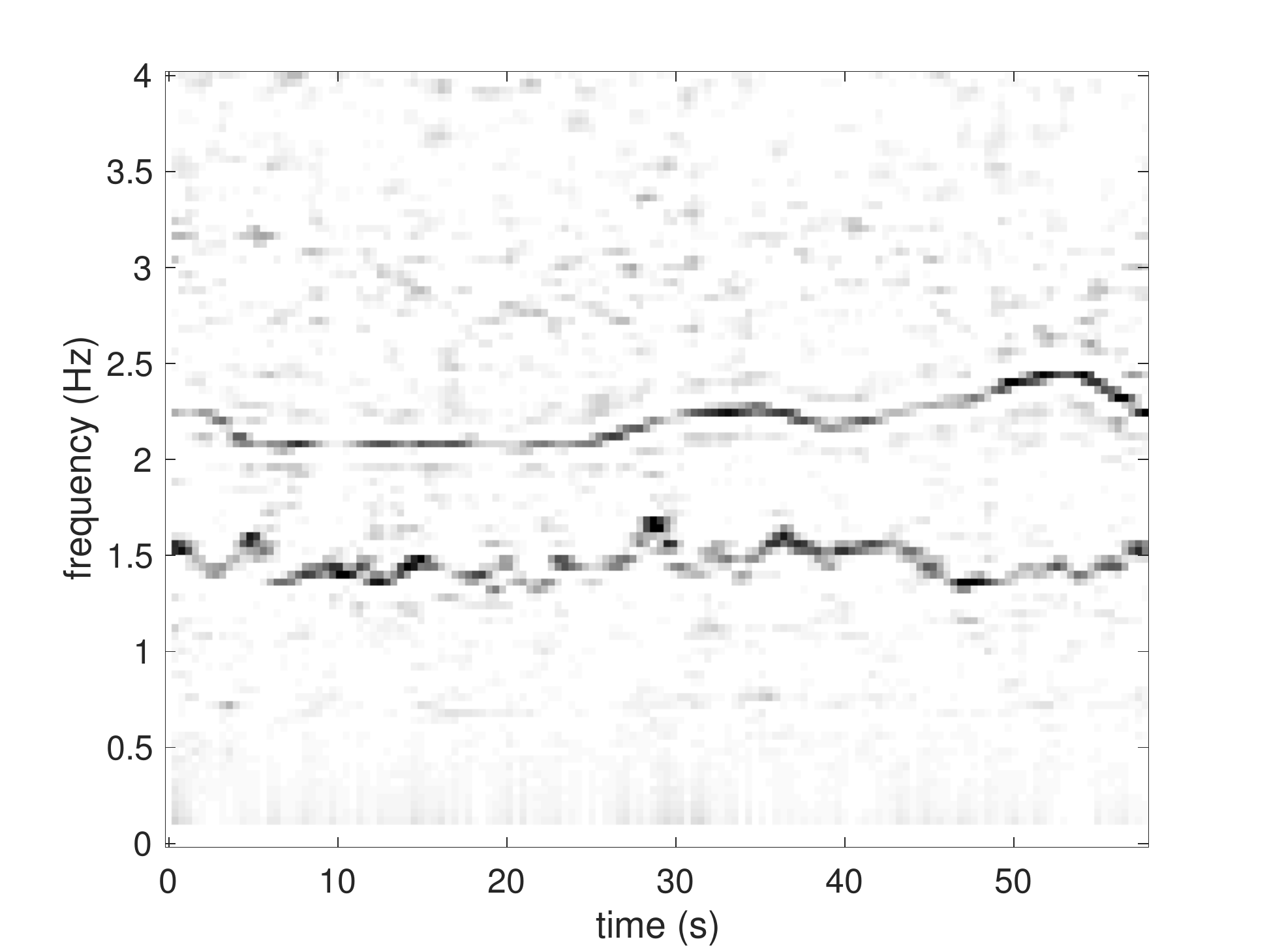}
	\caption{Left: the iSTCT of the ta-mECG signal $f$, $|U_f^{(h)}(t,\xi)|$. Right: the de-shape STFT of $f$, $|W_f^{(h)}(t,\xi)|^2$. Note that there are some light lines below 1.5Hz in the iSTCT, which are masked out in the de-shape STFT. It is clear that there are two curves in the de-shape STFT, one is around 1.5Hz, and one is around 2Hz. The curve around 1.5Hz is the maternal IHR, and the curve around 2Hz is the fetal IHR. To enhance the visibility, we only show the frequency up to 4 Hz.}
	\label{fig:3}
\end{figure}

While the de-shape algorithm has been applied to several problems \cite{deshapeapp,recycling,lin2019unexpected}, however, it has its own limitations. 
\begin{enumerate}
\item [(L1)] The first trouble is the noise. By taking the $\gamma$ power of $|V_f^{(h)}(t,\xi)|$, we might amplify the unwanted noise in a nonlinear way. 
Take the fetal ECG signal in Figure \ref{fig:1} as an example. Its iSTCT and de-shape STFT are displayed in Figure \ref{fig:3}. While we can observe two dark curves in the de-shape STFT (the curve above 2Hz is the fetal IHR, and the other one below 2Hz is the maternal IHR), there are several artifacts around them. 

\item [(L2)] Second, the inversion operation in the iSTCT would flip the noisy short-quefrency content into the high-frequency area. However, since the spectral envelope associated with the wave-shape function can be complicated, even when there is no noise, we generate complicated short-quefrency content in the STCT. Thus, a careful handle of the short-quefrency content is critical in the de-shape STFT framework. 

\item [(L3)] Third, the de-shape STFT might not faithfully reflect the strength of the non-sinusoidal oscillation. Indeed, since the de-shape STFT comes from a direct entrywise multiplication of the STFT and its iSTCT \eqref{Definition deshape STFT}, only the strength of the fundamental component will be preserved. For example, suppose we have two oscillatory signals of the same energy, $f_1(t)=As_1(\xi_0t)$ and $f_2(t)=As_2(\xi_0t)$, where $A>0$ and $\|s_1\|_2=\|s_2\|_2=1$, but $|\hat{s}_1(1)|\ll |\hat{s}_2(1)|$; that is, the fundamental component of $s_1$ is much weaker than that of $s_1$. In this case, the STFT of $f_1$ around $\xi_0$ will be much weaker than that $f_2$, and hence the de-shape STFT. 
\end{enumerate}
The above three issues might generate troubles when we interpret the results for scientific research. We thus propose a novel TF analysis tool motivated by these practical issues.

\subsection{Proposed algorithm -- Ramanujan de-shape }

To handle the above-mentioned limitations (L1)-(L3), we consider the RPT. Note that in general PT is a global time-domain approach, we may not be able to capture the time-varying oscillatory pattern. Thus, in order to capture the time-varying period, or the IF, of our analyzed dataset, we apply the PT in the frequency domain. We coined the algorithm {\em Ramanujan de-shape} (RDS) algorithm if we apply the RPT, or vectorized RDS (vRDS) if we apply the vRPT.

We now describe the RDS and vRDS algorithms. From now on, we follow the convention and assume that the indices of vectors and matrices start at $1$.
Suppose the continuous signal $f$ is sampled at uniform interval $\Delta_t>0$, and we denote its discretization as $\mathbf{f}\in\mathbb{R}^N$, where $\mathbf{f}(n) = f(n\Delta_t)$, where $1 \leq n\leq N$; that is, we ``record'' the signal for $N\Delta_t$ long.  
Choose a discrete window function $\mathbf{h}\in\mathbb{R}^{2K+1}$, where $K\in \mathbb{N}$. For example, $\mathbf{h}$ can be a discretization of a Gaussian, so that $\mathbf{h}(K+1)$ is the center of the Gaussian. 
Fix $M\in\mathbb{N}$, and let $M+1$ be the number of bins in the frequency axis of our targeting TFR. We choose STFT to generate our targeting TFR. Then, the STFT of $\mathbf{f}$ would be a matrix $\mathbf{V_f}\in\mathbb{C}^{(M+1)\times N}$, whose entries are
\[
\mathbf{V_f}(m,n) = \sum_{k=1}^{2K+1}\mathbf{f}(n+k-K-1)\mathbf{h}(k)e^{\frac{-2\pi i (k-1) (m-1)}{2M}},
\]
where we define $\mathbf{f}(n)=0$ when $n<1$ or $n>N$, $n=1,2,\dots,N$ is the time index, and $m=1,2,\dots,M+1$ is the frequency index. 
For a fixed $n$, we apply the program (\ref{lasso}) to each column of the matrix $\mathbf{V_f}(m,n)$, that is, the frequency axis, to obtain a matrix $\mathbf{X_f}\in\mathbb{R}^{\Phi(P_{max})\times N}$, where 
\[
\mathbf{X_f}(\cdot,n) \in \argmin\limits_{x\in \mathbb{R}^{\Phi(P_{max})}} \frac{1}{2}\norm{|\mathbf{V_f}(\cdot,n)|^{\gamma} - Bx }_2^2 + \lambda \norm{ x }_1, \quad \gamma>0.
\]
The parameter $\gamma$ should be chosen small enough like that in STCT in order to flatten the spectral envelope. The dictionary $B$ is the same as in (\ref{lasso}) with the chosen periodicity penalization function $\zeta$ and length $M+1$. We may only utilize the STFT with the corresponding frequency no more than a given maximal frequency $f_{max}$; that is, we only keep the first $1 + \floor{\frac{f_{max}}{2M\Delta_t }}$ entries of $|\mathbf{V_f}(\cdot,n)|^{\gamma}$ and the dictionary $B$ has $1 + \floor{\frac{f_{max}}{2M\Delta_t }}$ rows in the above minimization. We then transform $\mathbf{X_f}$ to a matrix $\mathbf{E_f}\in\mathbb{R}^{P_{max}\times N}$, where 
\begin{equation}\label{RDS-2}
\mathbf{E_f}(j,n) = EOP_{\mathbf{X_f}(\cdot,n)}(j)\,,
\end{equation}
where $j=1,\ldots, P_{max}$.
Then, define $\mathbf{P_f}(m,n)\in\mathbb{C}^{(M+1)\times N}$ as 
\begin{equation}\label{RDS-3}
\mathbf{P_f}(m,n) =\begin{cases}
\mathbf{E_f}(m,n) & 1\leq m\leq P_{max}; \\
0 & P_{max}<m\leq M+1.
\end{cases}
\end{equation}
If we have a priori knowledge about the periods and the periodicity penalization function is designed properly, we may assign $\mathbf{P_f}(m,n)$ to be the final TFR; that is,
\begin{equation}\label{RDS-4}
\mathbf{R_f}(m,n) = \mathbf{P_f}(m,n).
\end{equation}
In general, when we do not have the prior knowledge about the period,
the final TFR would be the entry-wise product
\begin{equation}\label{RDS-5}
\mathbf{R_f}(m,n) = |\mathbf{V_f}(m,n)|^{\gamma'} \mathbf{P_f}(m,n)\,,
\end{equation}
where $\gamma'>0$ is a parameter that the user can choose. Indeed, since there is a nontrivial relationship between Ramanujan subspaces of periods $p$ and $p'$ when $(p,p')\neq 1$, it is possible that the RPT outputs spurious spikes in the divisors of the fundamental IF we have interest. Thus, we count on the STFT to ``mask'' those spurious spikes. This is possible since in the spectrogram, we have the fundamental IF and their multiples. 
The pseudo-code for RDS is shown in Algorithm \ref{Algorithm RDS}. 

\begin{algorithm}[hbt!]\label{Algorithm pseudocode for RDS}
	\caption{Ramanujan de-shape (RDS)}
    \begin{flushleft}
    	\textbf{INPUT:} sampled signal $\mathbf{f}(n)\in\mathbb{R}^N$, window function $\mathbf{h}(n)$, $P_{max}$, $f_{max}$, $\lambda$, $\gamma$, $\gamma'$, and the periodicity penalization function $\zeta$; \\
    	\textbf{OUTPUT:} RDS $\mathbf{R_f}(m,n)$.
    \end{flushleft}
	\begin{algorithmic}[1]\label{Algorithm RDS}
		\STATE Compute the STFT $\mathbf{V_f}(m,n)$;
		\STATE Set $\mathbf{T_f}\in\mathbb{C}^{(\floor{\frac{f_{max}}{2M\Delta_t }}+1)\times N}$, where $\mathbf{T_f}(m,n) = \mathbf{V_f}(m,n)$ for $1\leq m\leq \floor{\frac{f_{max}}{2M\Delta_t }}+1$ and $1\leq n \leq N$, and zero otherwise;
		\STATE for $n$ from $1$ to $N$, solve for 
		\[
		\mathbf{X_f}(\cdot,n) \in \argmin\limits_{x\in \mathbb{R}^{\Phi(P_{max})}} \frac{1}{2}\norm{|\mathbf{T_f}(\cdot,n)|^{\gamma} - Bx }_2^2 + \lambda \norm{ x }_1;
		\]
		\STATE Obtain $\mathbf{E_f}(m,n)$ by $\mathbf{E_f}(j,n) = EOP_{\mathbf{X_f}(\cdot,n)}(j)$;
		\STATE Extend $\mathbf{E_f}(m,n)$ to $\mathbf{P_f}(m,n)$;
		\STATE $\mathbf{R_f}(m,n) = |\mathbf{V_f}(m,n)|^{\gamma'} \mathbf{P_f}(m,n)$ or $\mathbf{R_f}(m,n) = \mathbf{P_f}(m,n)$, depending on the prior knowledge.
	\end{algorithmic}
\end{algorithm}

The computational complexity of Algorithm \ref{Algorithm RDS} is $O(N(F(\floor{\frac{f_{max}}{2M\Delta_t }}+1)+M\log M))$, where $F(\floor{\frac{f_{max}}{2M\Delta_t }}+1)$ is the complexity of running the minimization at step 3 in the algorithm. Thus, if $M$ is fixed and hence $F(\floor{\frac{f_{max}}{2M\Delta_t }}+1)$ is fixed, overall the computational complexity is $O(N)$.

In order to obtain a smoother TFR, instead of applying RPT at each time slot, we apply vRPT on the neighborhood of that time; that is, for each time index $n$ and chosen $k\in\mathbb{N}$, we rather solve for
\[
\mathbf{X_f}(\cdot,n) \in \argmin\limits_{x\in \mathbb{R}^{\Phi(P_{max})}} \frac{1}{2}\norm{|\mathbf{\tilde{V}_f}(\cdot,n)|^{\gamma} - Bx\mathbf{1}^T }_{F}^2 + \lambda \norm{ x }_1, \quad \gamma>0,
\]
where $\mathbf{1}$ is a $(2k+1)$-dim vector with all entries $1$, $|\mathbf{\tilde{V}_f}(\cdot,n)|^{\gamma} \in \mathbb{R}^{(M+1)\times(2k+1)}$, and
\[
|\mathbf{\tilde{V}_f}(\cdot,n)|^{\gamma} := 
\begin{bmatrix}
|\mathbf{V_f}(\cdot,n-k)|^{\gamma} & \dots & |\mathbf{V_f}(\cdot,n)|^{\gamma} & \dots & |\mathbf{V_f}(\cdot,n+k)|^{\gamma}
\end{bmatrix},
\]
except that the smallest time index would be 1 if $n-k<1$ and the largest time index would be $N$ if $n+k>N$, and $\norm{\cdot}_{F}$ is the matrix Frobenius norm. 
With $\mathbf{X_f}$, we can follow the same line as those from \eqref{RDS-2} to \eqref{RDS-5} to construct the final TFR, denoted as $\mathbf{VR_f}$.
We call this proposed algorithm the vRDS, whose pseudo-code is shown in Algorithm \ref{Algorithm vRDS}. Note that when $k=0$, the  vRDS is just RDS.

\begin{algorithm}[hbt!]\label{Algorithm pseudocode for vRDS}
	\caption{Vectorized Ramanujan de-shape (vRDS)}
	\begin{flushleft}
		\textbf{INPUT:} sampled signal $\mathbf{f}(n)\in\mathbb{R}^N$, window function $\mathbf{h}(n)$, $P_{max}$, $f_{max}$, $\lambda$, $\gamma$, $\gamma'$, number of time slots $k$, and the periodicity penalization function $\zeta$; \\
		\textbf{OUTPUT:}  vRDS $\mathbf{VR_f}(m,n)$.
	\end{flushleft}
	\begin{algorithmic}[1]\label{Algorithm  vRDS}
		\STATE Compute the STFT $\mathbf{V_f}(m,n)$;
		\STATE Set $\mathbf{T_f}\in\mathbb{C}^{(\floor{\frac{f_{max}}{2M\Delta_t }}+1)\times N}$, where $\mathbf{T_f}(m,n) = \mathbf{V_f}(m,n)$ for $1\leq m\leq \floor{\frac{f_{max}}{2M\Delta_t }}+1$ and $1\leq n \leq N$, and zero otherwise;
		\STATE for $n$ from $1$ to $N$, solve for 
		\[
		\mathbf{X_f}(\cdot,n) \in \argmin\limits_{x\in \mathbb{R}^{\Phi(P_{max})}} \frac{1}{2}\norm{|\mathbf{\tilde{T}_f}(\cdot,n)|^{\gamma} - Bx*\mathbf{1}^T }_{F}^2 + \lambda \norm{ x }_1,
		\]
		where 
		\[
		|\mathbf{\tilde{T}_f}(\cdot,n)|^{\gamma} = 
		\begin{bmatrix}
		|\mathbf{T_f}(\cdot,n-k)|^{\gamma} & \dots & |\mathbf{T_f}(\cdot,n)|^{\gamma} & \dots & |\mathbf{T_f}(\cdot,n+k)|^{\gamma}
		\end{bmatrix};
		\]
		\STATE Obtain $\mathbf{E_f}(m,n)$ by $\mathbf{E_f}(j,n) = EOP_{\mathbf{X_f}(\cdot,n)}(j)$;
		\STATE Extend $\mathbf{E_f}(m,n)$ to $\mathbf{P_f}(m,n)$;
		\STATE $\mathbf{VR_f}(m,n) = |\mathbf{V_f}(m,n)|^{\gamma'} \mathbf{P_f}(m,n)$ or $\mathbf{VR_f}(m,n) = \mathbf{P_f}(m,n)$, depending on the prior knowledge.
	\end{algorithmic}
\end{algorithm}

\section{Numerical Results}\label{Section Numerical results}
In this section we illustrate the effectiveness of the proposed $l_1$ penalized PT and RDS on several examples. All the numerical examples (as well as those in previous sections) are generated in Matlab R2019a on a 2017 15-inch Macbook Pro with 16GB memory and 2.8 GHz Quad-Core Intel Core i7 processor. The $l^1$ penalized regression is based on Alternating Direction Method of Multipliers (ADMM) \cite{Boyd} in Matlab. 
The implementation of Small to Large algorithm, M-best algorithm, Best Correlation algorithm are downloaded from \url{https://sethares.engr.wisc.edu/downloadper.html} announced by the inventors of those algorithms. 
For the reproducibility purpose, the Matlab code is available upon requests.

\subsection{Simulated signal}

We start by considering two synthesized data. In the first example, we have two periodic signals $y_1(n)$ and $y_2(n)$ of length 500, where $y_1(n) = 5$ if $n$ is a multiple of 15, and 0 otherwise; $y_2(n) = 8$ if $n$ is a multiple of 21, and 0 otherwise. Take two white random processes $\xi_n, \eta_n$, $n=1,\ldots, 250$ that are independently sampled from the uniform distribution on $(-0.3,0.3)$. 
Generate $A_1,A_2\in\mathbb{R}^{250}$ so that $A_1(n) = \xi_n$ if $n$ is a multiple of 15, and 0 otherwise, and $A_2(n) = \eta_n$ if $n$ is a multiple of 21, and 0 otherwise. We also set $A_1(15) = A_2(21)=0$. 
Then $y_1(n)$ and $y_2(n)$ are pointwisely multiplied by $A_1(n)$ and $A_2(n)$ respectively, and set $y(n) = A_1(n)y_1(n)+A_2(n)y_2(n)$. 
To make the example more challenging, we further apply an envelope $E(n)$ to $y(n)$, where $E(n) = e^{\frac{-n^2}{250^2}}$. Then, we apply the proposed RPT to $\tilde{y}(n) = E(n)y(n)$, with $P_{max} = 50$, $\lambda = 1$ and the periodicity penalty function $\zeta(p) = p$. Results are shown in Figure \ref{fig:env1}. In the period estimation, we can find peaks at 15 and 21 as well as their divisors 3, 5 and 7. The arise of 3, 5, 7 is due to the fact that by Proposition \ref{prop4}, a $p$-periodic signal can contain components in Ramanujan subspaces $\mathcal{R}_{p_i,N}$, where $p_i$ are divisors of $p$.

\begin{figure}[!htbp]
	\includegraphics[trim=70 50 70 50, clip, width=1\textwidth]{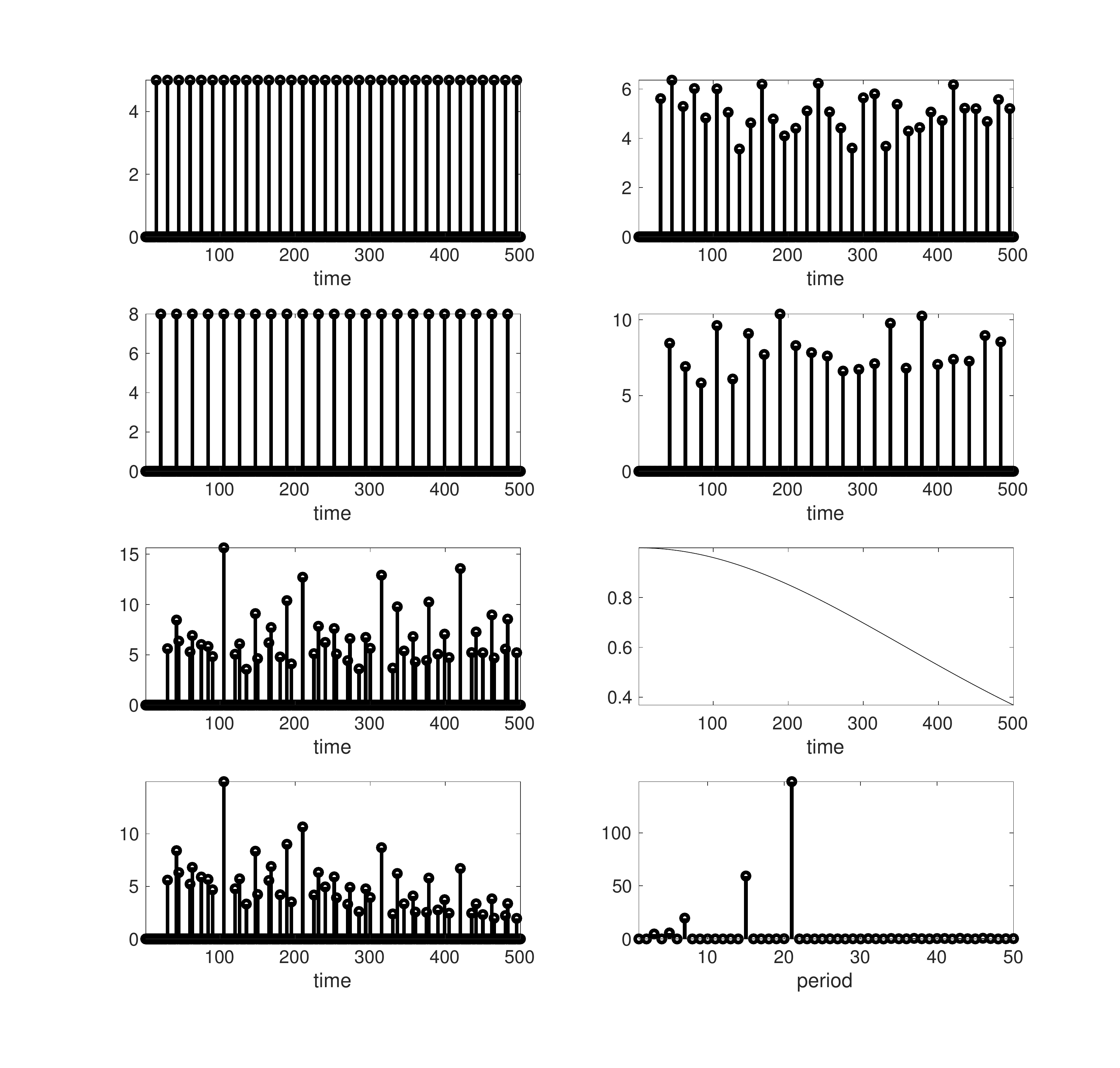}
	\caption{Top row: left panel $y_1(n)$; right panel $A_1(n) y_1(n)$. Second row: left panel $y_2(n)$; right panel $A_2(n) y_2(n)$. Third row: left panel $y(n)$; right panel envelope $E(n)$. Bottom row: left panel $\tilde{y}(n) = E(n)y(n)$; right panel period estimation of $\tilde{y}(n)$. Clearly, two periods $15$ and $21$ are captured, along with some unwanted periods that have small energies.}
	\label{fig:env1}
\end{figure}

The second example is a signal with a series of Gaussian peaks of length 301. $y_1(n)$ and $y_2(n)$ are two periodic signals, where $y_1(n)$ is 27-periodic with Gaussian peaks centered at every 27 points with width 9 points; $y_2(n)$ is 17-periodic with Gaussian peaks centered at every 17 points with width 13 points, and both have peak values 3. 
Set $\xi_n, \eta_n$ to be i.i.d. uniformly distributed on $(-0.2,0.2)$ and $A_1(n) = \xi_n$ and $A_2(n) = \eta_n$.
Define $E(n) := e^{\frac{-n^2}{(2\times 301)^2}}$ and $e(n)$ is i.i.d. Gaussian noise with standard deviation 1.11.
Finally, set $y(n) = A_1(n)y_1(n)+A_2(n)y_2(n)$ and $\tilde{y}(n) = E(n)y(n)+e(n)$. The signal to noise ratior (SNR) is 5.21, where SNR is defined as $20\log\frac{\text{std}(E(n)y(n))}{\text{std}(e(n))}$.
Figure \ref{fig:env2} shows the result of the proposed RPT on $\tilde{y}(n)$ with $P_{max} = 50$, $\lambda = 0.5$ and $\zeta(p)=p^2$. We observe peaks at periods 17 and 27.

\begin{figure}[h]
	\includegraphics[trim=70 50 70 50,clip,width=1\textwidth]{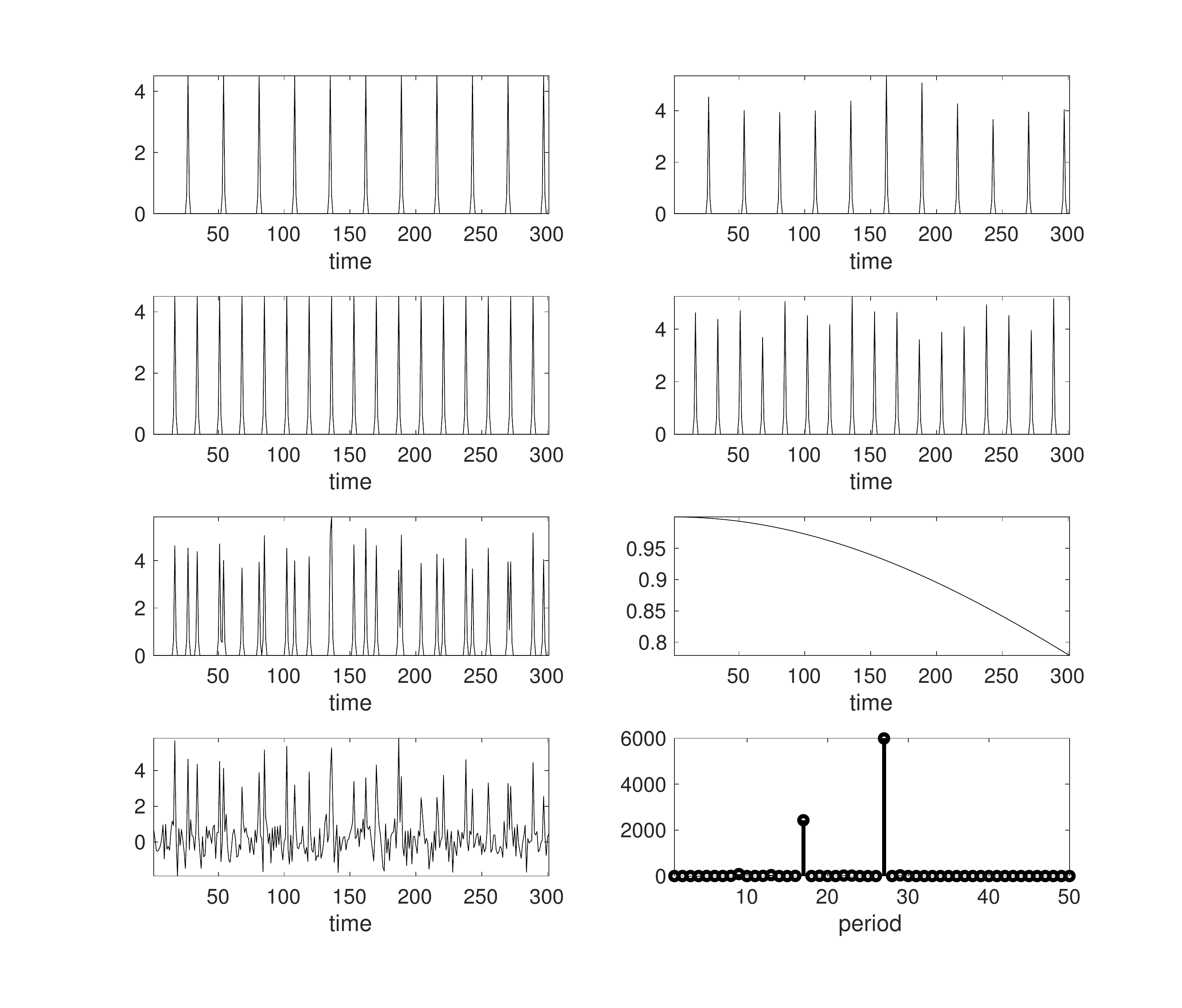}
	\caption{Top row: left panel $y_1(n)$; right panel $A_1(n) y_1(n)$. Second row: left panel $y_2(n)$; right panel $A_2(n) y_2(n)$. Third row: left panel $y(n)$; right panel envelope $E(n)$. Bottom row: left panel $\tilde{y}(n) = E(n)y(n)+e(n)$, SNR = 5.21; right panel period estimation of $\tilde{y}(n)$.}
	\label{fig:env2}
\end{figure}

%
%
%
%

\subsection{Trans-abdominal maternal ECG signal}

We first show the RDS and the  vRDS of the ta-mECG shown in Figure \ref{fig:1}. 
This signal is from the 4th channel of the 17th recording of the Noninvasive Fetal ECG database for the PhysioNet Computing in Cardiology Challenge 2013 (CinC2013) \cite{goldberger2000physiobank},\footnote{\url{https://physionet.org/content/challenge-2013/1.0.0/}}. Each recording in this database contains 4 ta-mECG channels, and each channel is recorded for 60 seconds and sampled at 1000Hz. 
In this case, the signal is downsampled to $250$Hz, that is, $\Delta_t = \frac{1}{250}$ second. We de-trend the signal by the standard median filter technique \cite{median} for the ECG signal. In the STFT, each frequency bin is $0.04$Hz, so there are $M=3126$ frequency bins. 
To run the RDS and  vRDS, we choose $P_{max} = 100$ (corresponds to 4Hz), $f_{max} = 60$Hz, $\zeta(p)=p^2$ as is suggested in \cite[(32)]{PPV3}, $\gamma = 0.1$, and $\gamma'=1$. We choose $\lambda = 0.01$ for the RDS and $k = 1$ and $\lambda = 0.03$ for the vRDS. The result is shown in Figure \ref{fig:3-2}. Note that in vRDS, each time we have 3 time slots in the minimization program, and hence the quadratic term in the convex program would be approximately 3 times larger than that of the RDS. Thus, the $\lambda$ should be 3 times as that of the RDS for a fair comparison.

\begin{figure}[!htbp]
	\includegraphics[trim=20 0 50 0,clip,width=0.49\textwidth]{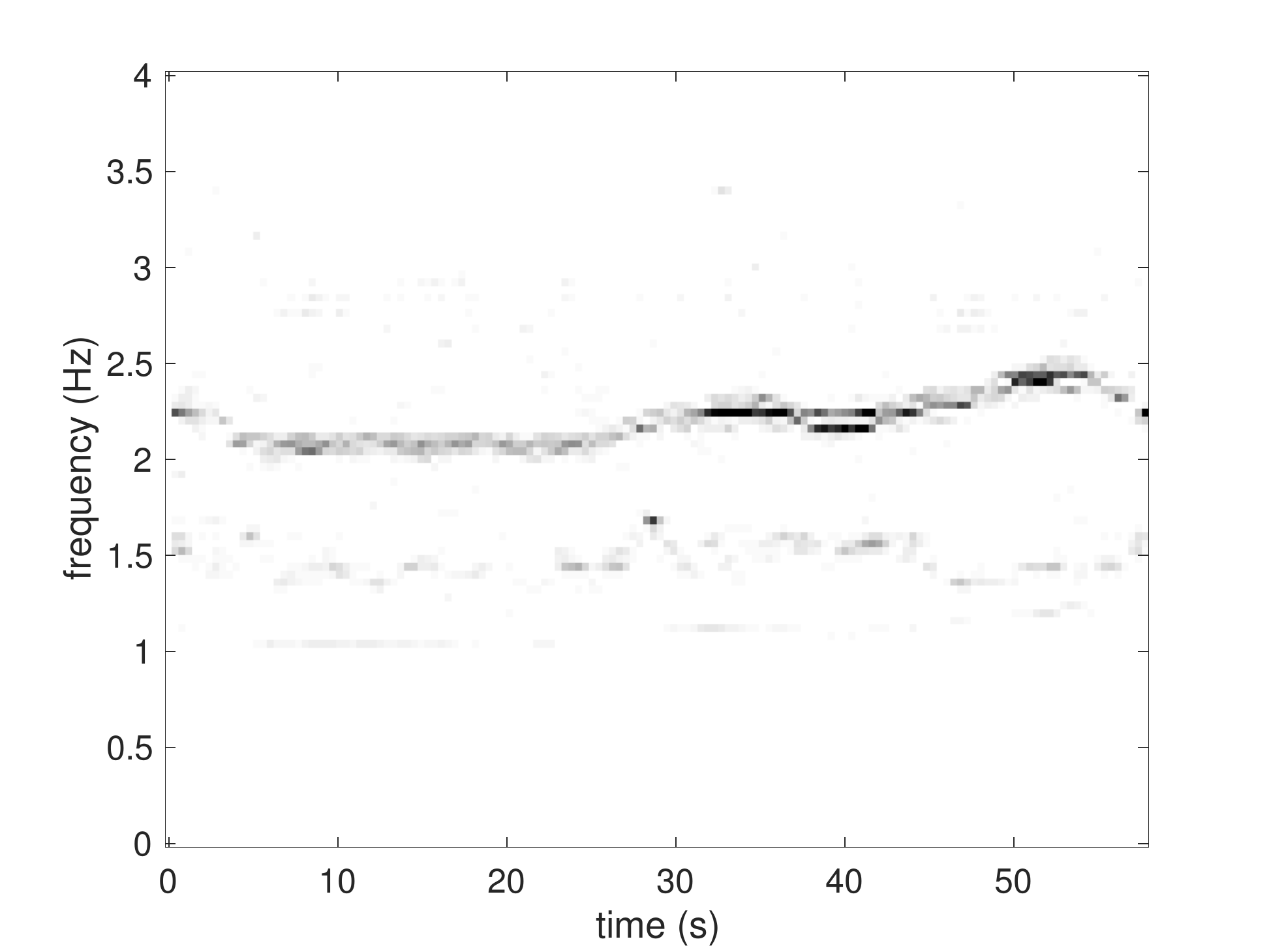}
	\includegraphics[trim=20 0 50 0,clip,width=0.49\textwidth]{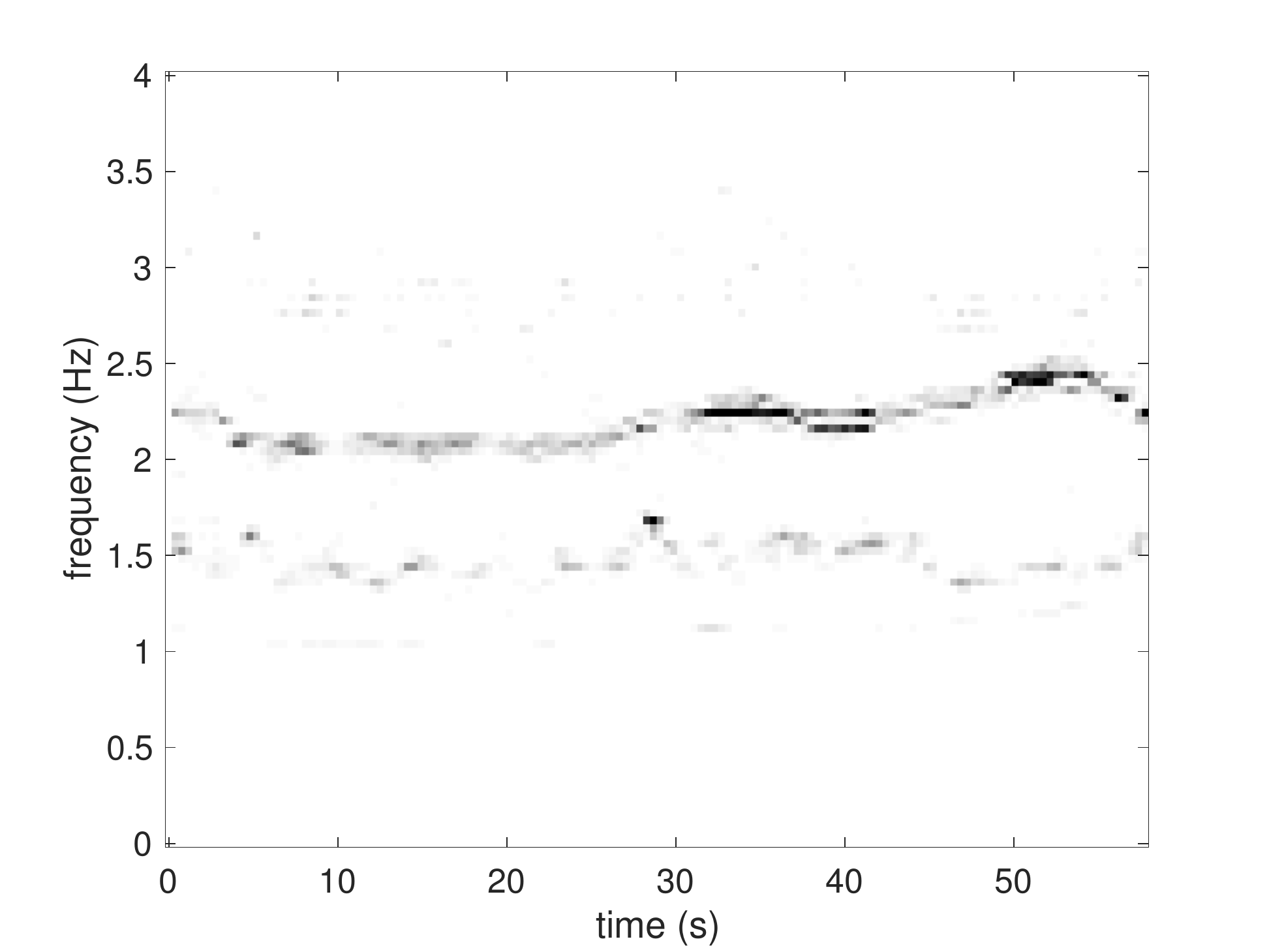}
	\caption{Left: the $\mathbf{P_f}$ of the ta-mECG signal $f$. Right: the Ramanujan de-shape (RDS) $\mathbf{R_f}$. Note that there is a curve around 1Hz in $\mathbf{P_f}$, which is masked out in $\mathbf{R_f}$. It is clear that there are two curves in the RDS, one is around 1.5Hz, and one is around 2Hz. The curve around 1.5Hz is the maternal IHR, and the curve around 2Hz is the fetal IHR. Compared with the de-shape STFT shown in Figure \ref{fig:3}, the TFR generated by RDS is cleaner. To enhance the visibility, we only show the frequency up to 4 Hz.}
	\label{fig:3-2}
\end{figure}

Second, we show that the PT is not suitable to analyze this kind of signal. In Figure \ref{fig:2pt}, the results of RPT with different periodicity penalty functions are shown. In this example, we know that the maternal IHR is higher than 1.4Hz, so we take $P_{max}=200$, which is related to $250/200=1.25$ Hz, to enhance the result. Note that due to the time-varying frequency and the global nature of the PT, there are several spikes in all results and it is not clear how to interpret the results.

\begin{figure}[!htbp]
	\begin{minipage}{19.5 cm}
	\hspace{-90pt}\includegraphics[trim=20 0 40 30,clip,width=0.32\textwidth]{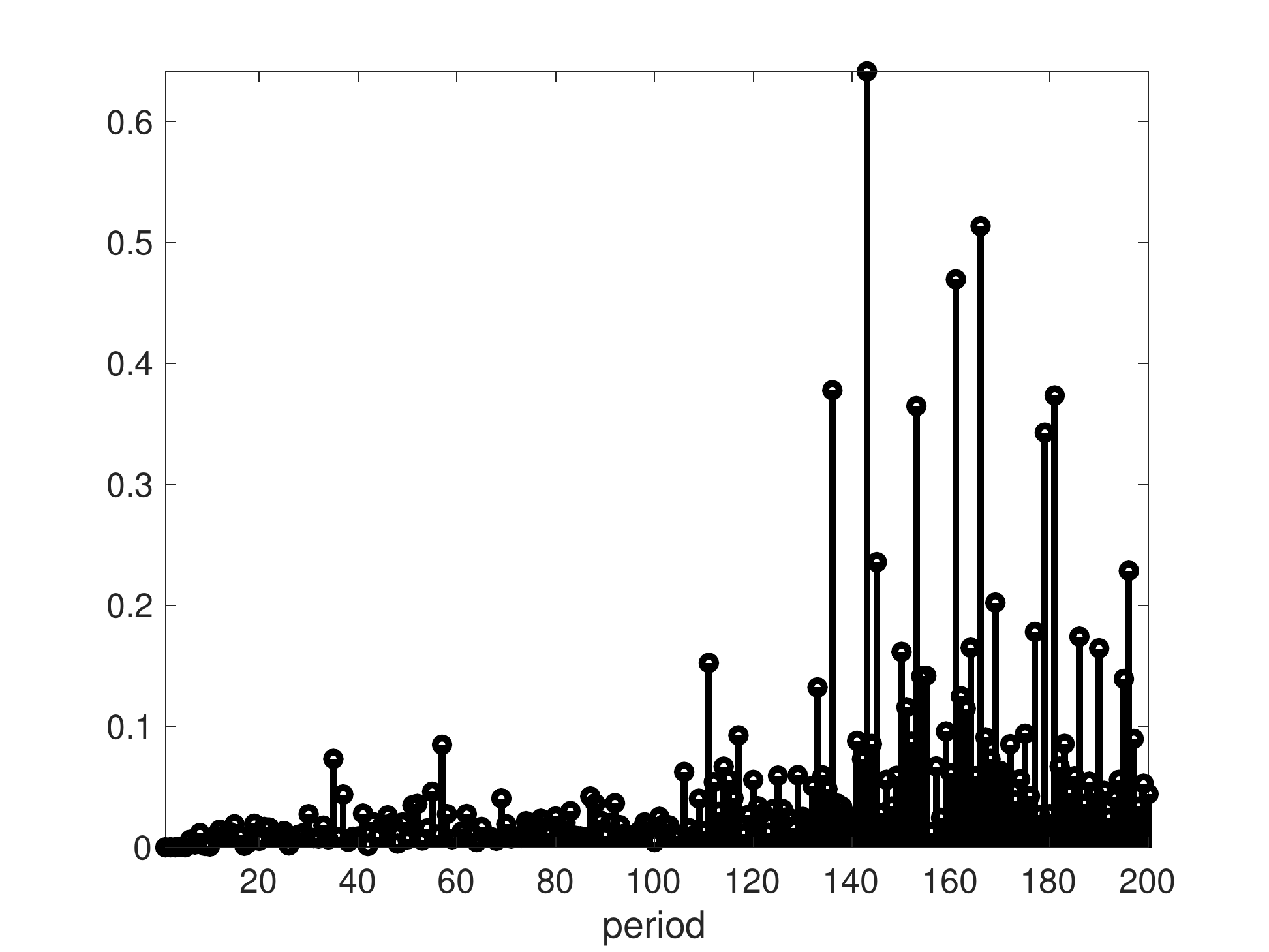}
	\includegraphics[trim=20 0 40 30,clip,width=0.32\textwidth]{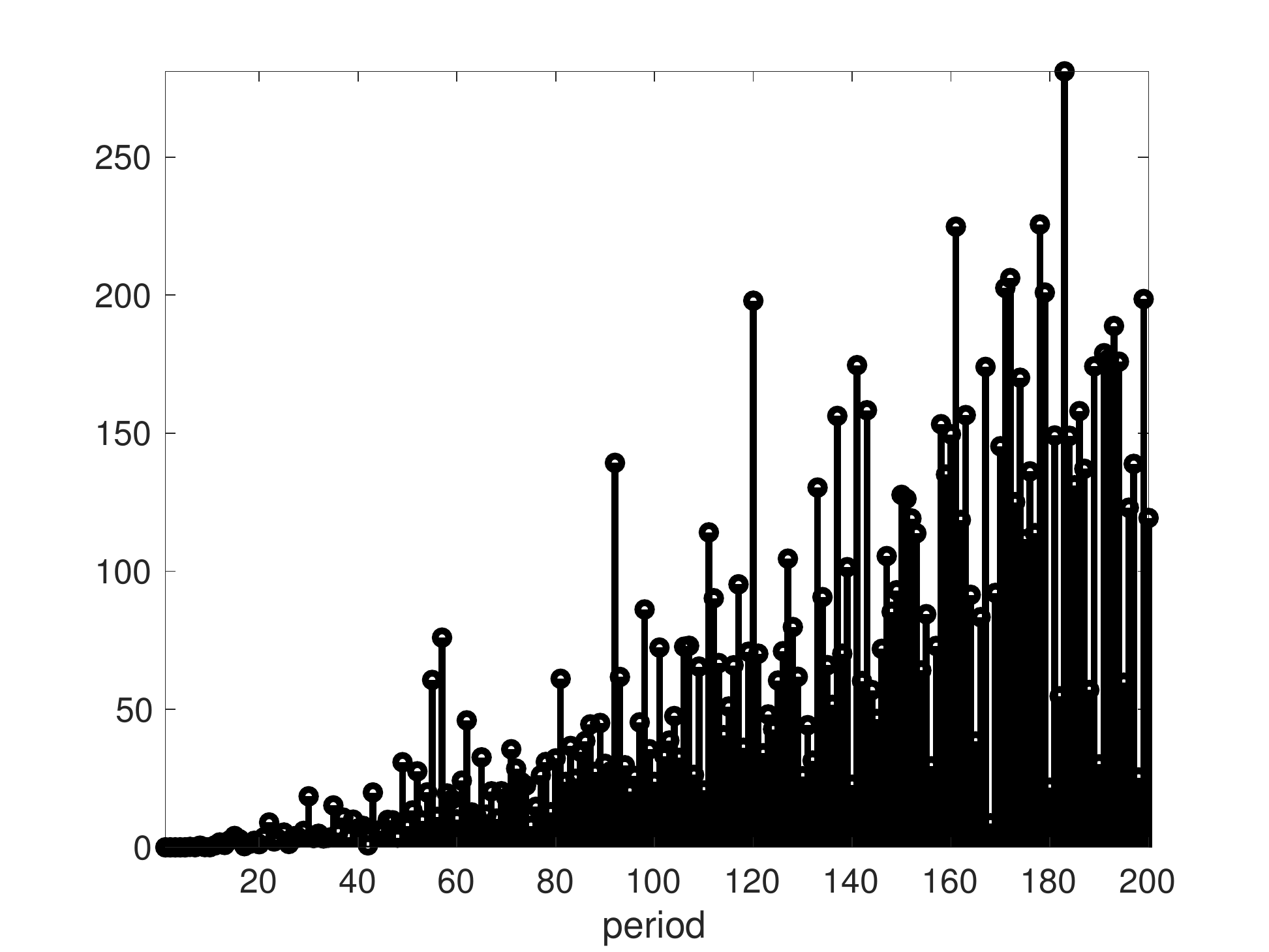}
	\includegraphics[trim=20 0 40 30,clip,width=0.32\textwidth]{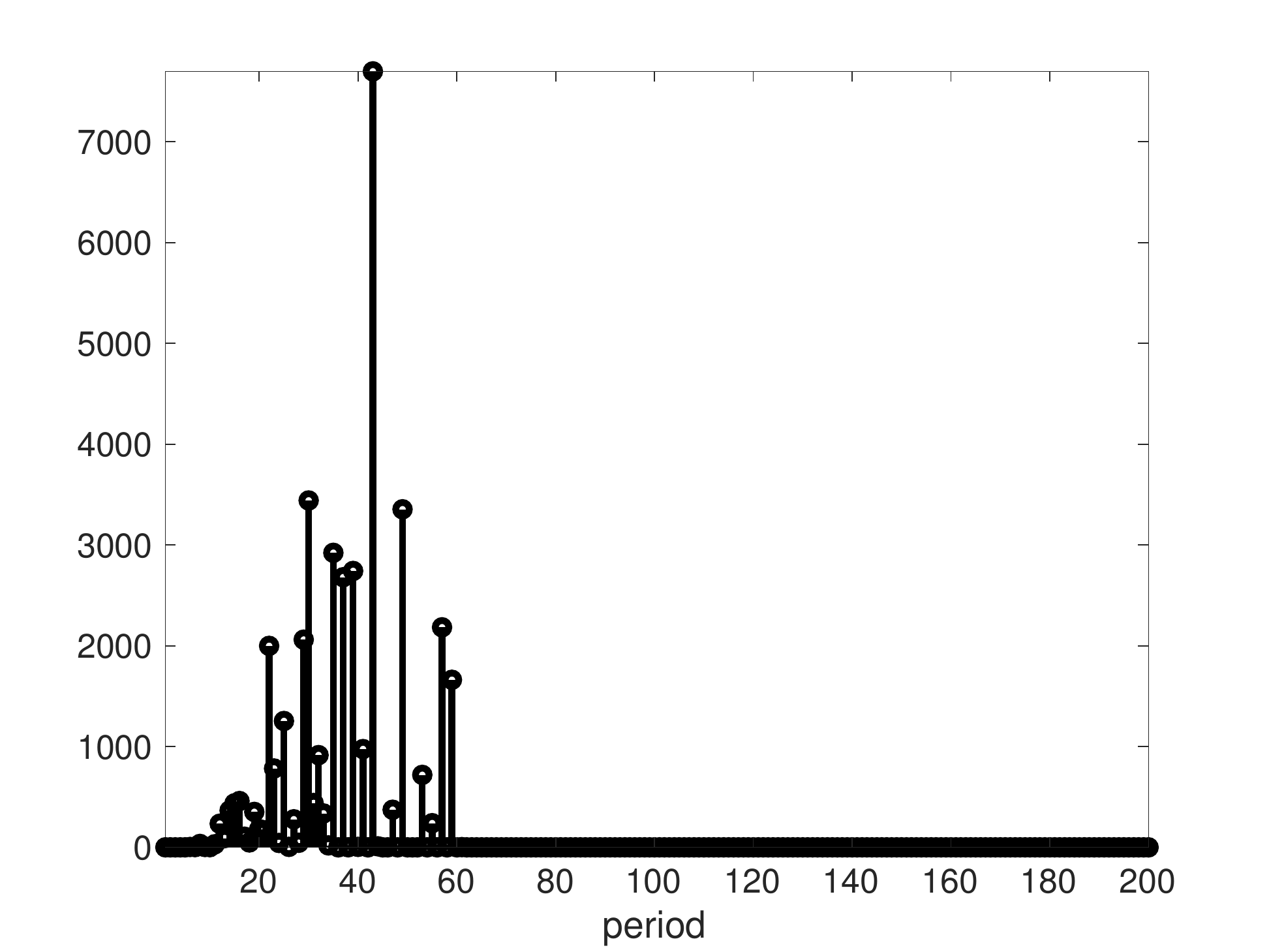}
	\end{minipage}
	\caption{Application of the RPT on the ta-mECG signal $f$ shown in Figure \ref{fig:1}.  
	Left: the application of RPT with $\zeta(p)=1$ and $\lambda=5$. Middle: the application of RPT with $\zeta(p)=p$ with $\lambda=20$. Right: the application of RPT with $\zeta(p)=p^2$ with $\lambda=5$. Since $\Delta_t = 1/250$ second, period $p$ corresponds to frequency $250/p$ Hz. As expected, due to the time-varying frequency, there are several spikes in all results.}
	\label{fig:2pt}
\end{figure}

Third, we show a comparison with other existing PT algorithms. We replace the RPT in Step 3 in Algorithm 1 by other existing PT algorithms, including the Small to Large algorithm, the M-Best algorithm and the Best correlation algorithm. We consider these algorithms since the benchmark code is available online in \url{https://sethares.engr.wisc.edu/downloadper.html}. See Figure \ref{fig:4} for the comparison. We set $P_{max} = 100$ in all cases. In the Small to Large algorithm, we observe that the threshold $T$ should be very small otherwise no periodic component would be extracted. Thus, we set $T = 0.01$ which is close to the threshold when there is at least one output period and we can observe some artifacts in the high-frequency range. In the M-Best algorithm and the Best correlation algorithm, we assume that we know the number of oscillatory components, that is $M=2$. Under this assumption, we can only observe one curve in the TFR. On the other hand, when $M=3$, these algorithms tend to show the maternal IHR but vaguely. We thus show the results of these algorithm for $M=3$. 

\begin{figure}[!htbp]
	\includegraphics[trim=20 0 30 0,clip,width=0.49\textwidth]{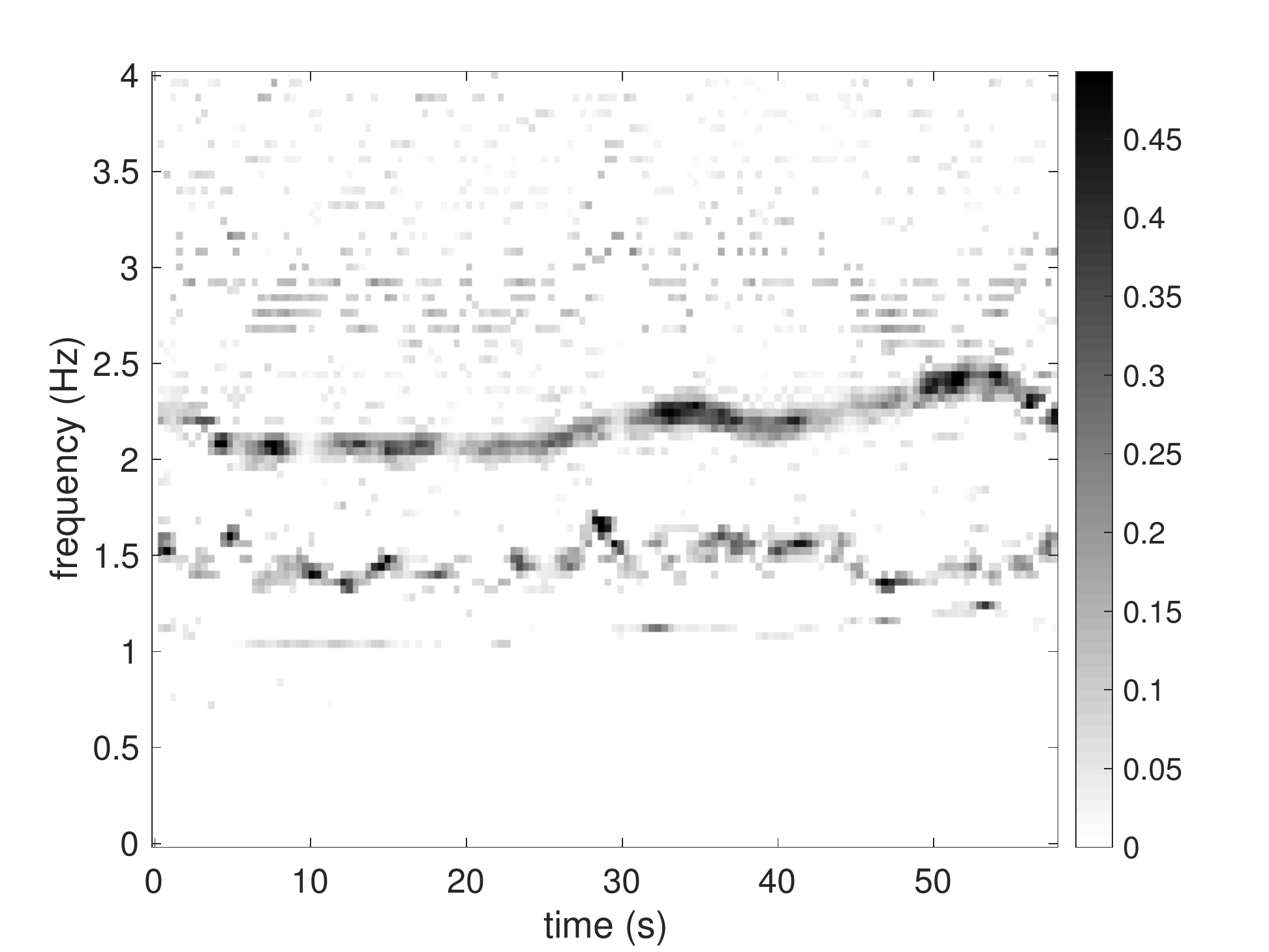}
	\includegraphics[trim=20 0 30 0,clip,width=0.49\textwidth]{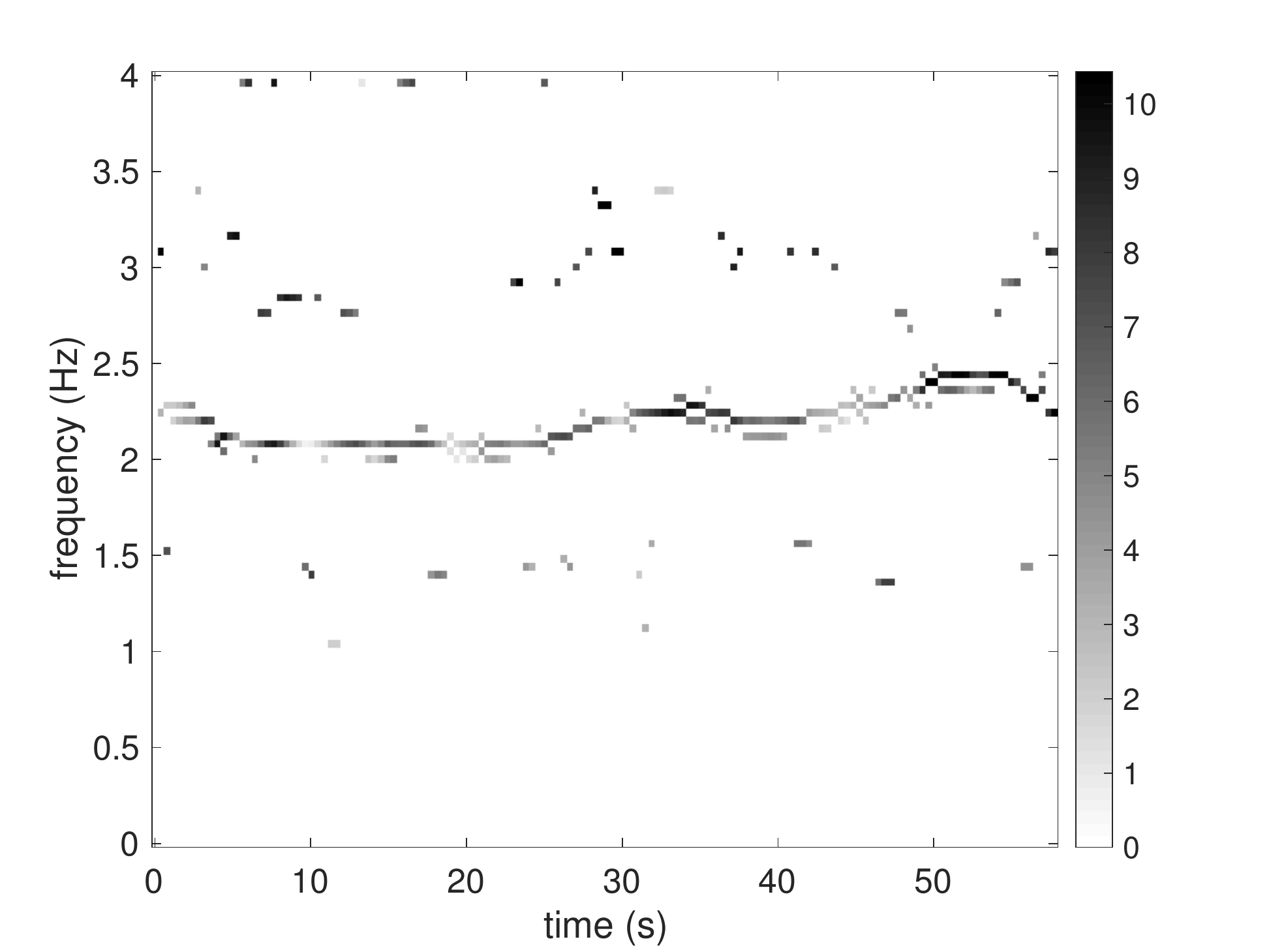}
	\begin{minipage}{19.5 cm}
	\hspace{-90pt}\includegraphics[trim=20 0 30 0,clip,width=0.32\textwidth]{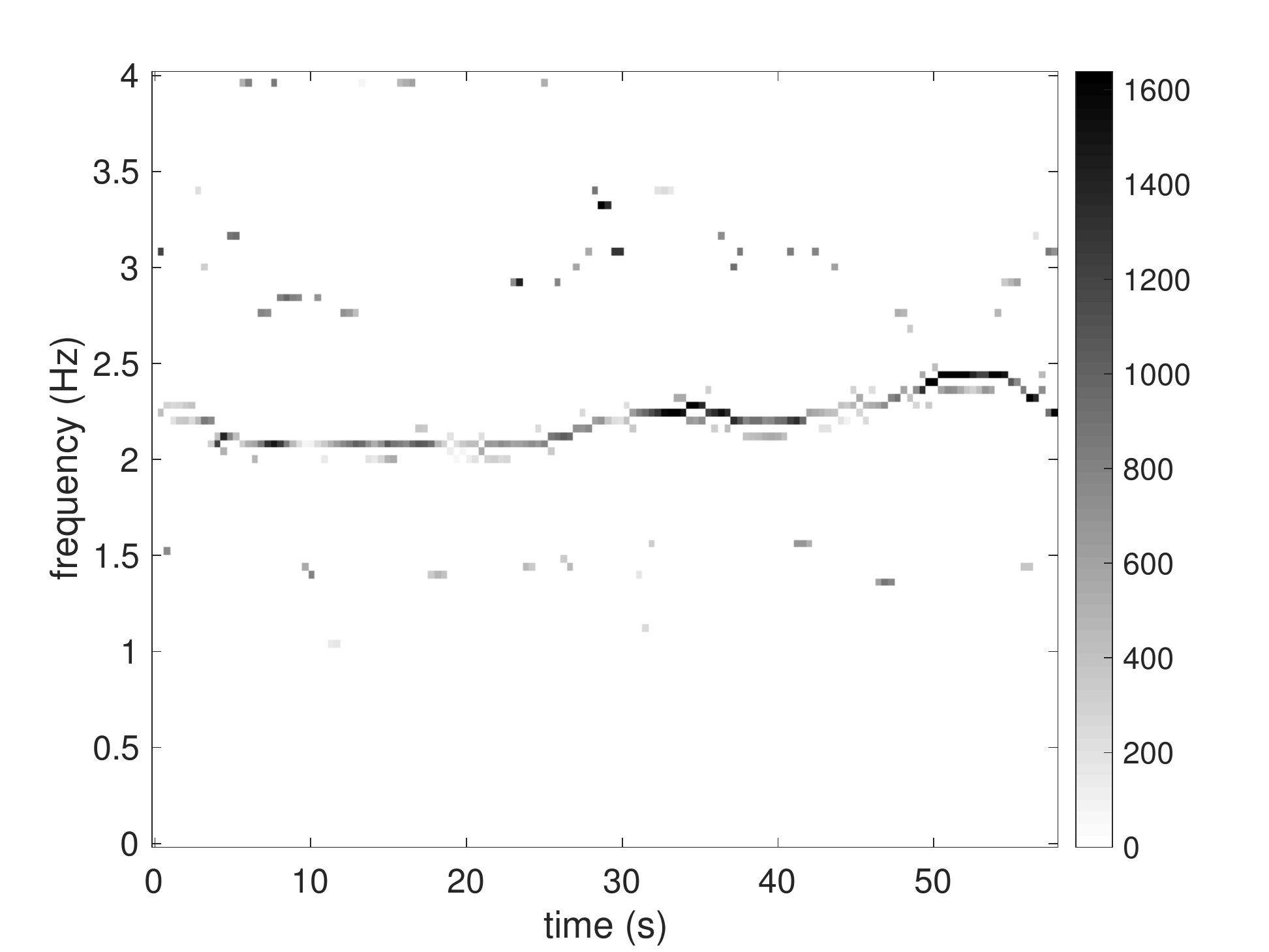}
	\includegraphics[trim=20 0 30 0,clip,width=0.32\textwidth]{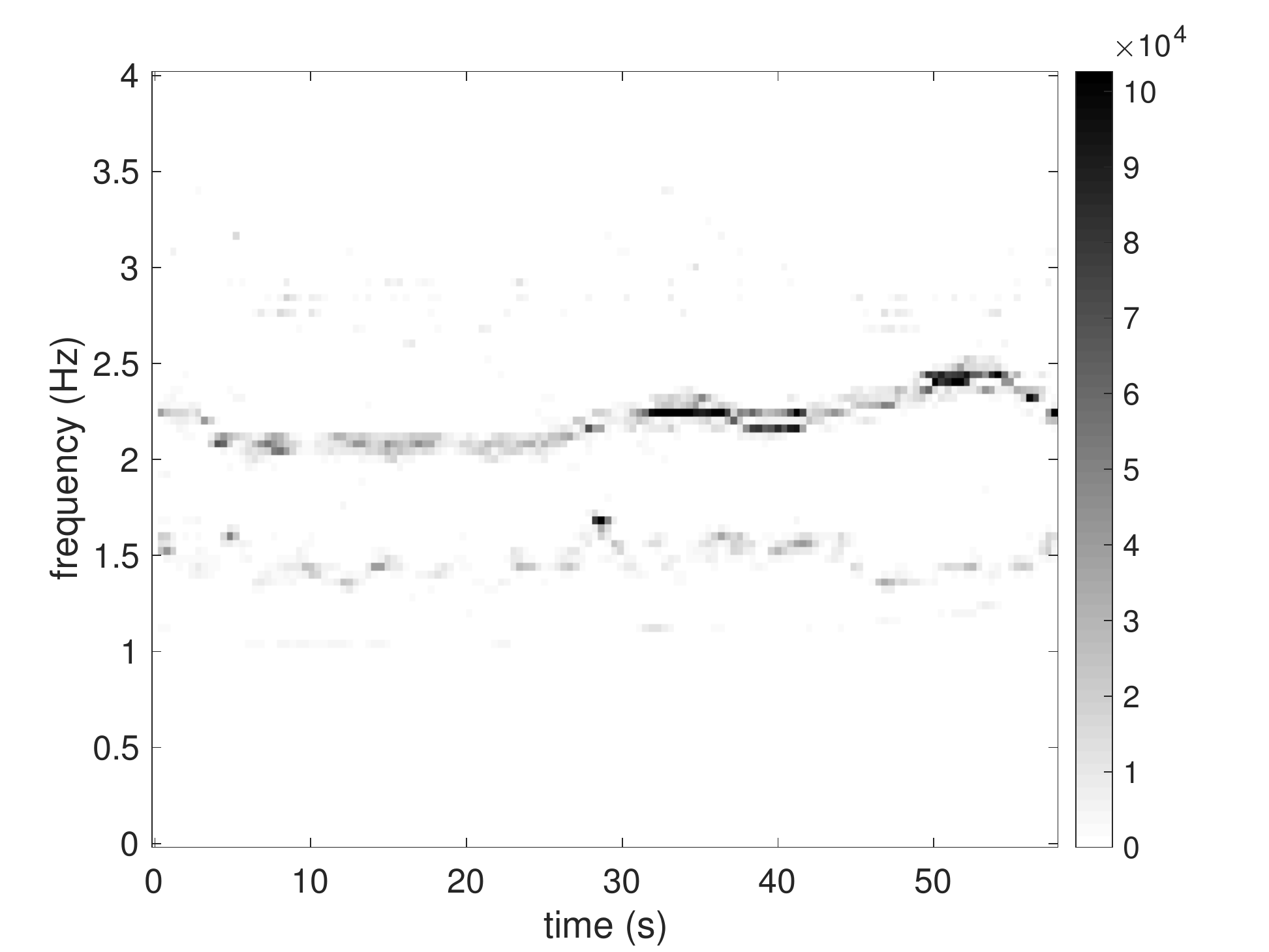}
	\includegraphics[trim=20 0 30 0,clip,width=0.32\textwidth]{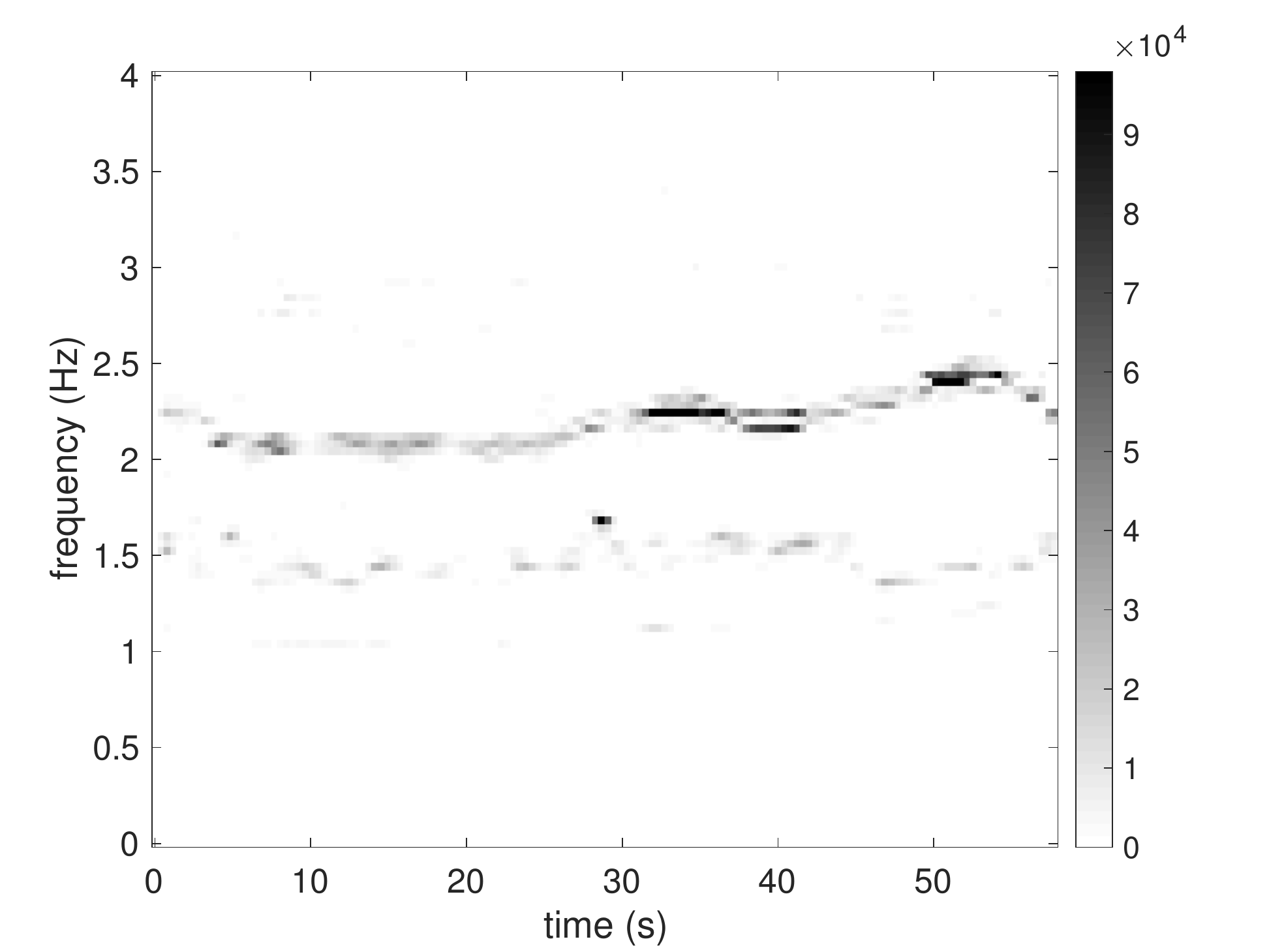}
	\end{minipage}
	\caption{Illustration of replacing the proposed Ramanujan periodic transform (RPT) by different periodicity transform algorithms in the Ramanujan de-shape. The resulting time-frequency representations are shown. Top left: Small to Large, where $T=0.01$. Top right: M-best, where $M=3$. Second row left: Best correlation, where $M=3$. Second row right: 
	our RDS approach. Bottom row: our vRDS approach.}
	\label{fig:4}
\end{figure}

Finally, we compare the RDS with the de-shape algorithm, and point out how the proposed RDS resolves the limitations of the de-shape STFT. We take the 2nd channel of the 31st recording of CinC2013 database. The signal is shown in Figure \ref{fig:fECG-1}.  
Again, we can see that the ta-mECG has non-sinusoidal oscillation pattern. After downsampling the signal to 250 Hz and de-trending it as in the previous example, we apply the RDS, vRDS and the de-shape STFT. The result of de-shape STFT is shown in Figure \ref{fig:fECG-1} and those of RDS and  vRDS are shown in Figure \ref{fig:fECG-2}. Like what is shown in Figure \ref{fig:3}, the de-shape STFT captures both the maternal IHR and fetal IHR, while the resulting TFR is corrupted by noise in the background. We then apply the RDS with $P_{max} = 100$ (corresponds to up to 4Hz), $f_{max} = 49$Hz and $\gamma = 0.1$. When $\lambda = 0.01$, we can successfully extract both the maternal IHR and the fetal IHR. Compared with the de-shape STFT, the denoising effect of the RDS is clear. While less dramatic compared with the improvement of the RDS over the de-shape STFT, the  vRDS can further stabilize the TFR. Therefore, the first and second limitations, (L1) and (L2), of de-shape STFT are resolved. 
When $\lambda = 0.03$, the fetal IHR is suppressed and the maternal IHR is kept. This is due to the property that the energy of the oscillatory signal is well preserved by the PT and the fact that the mECG has a larger energy than that of the fECG. Moreover, since the penalty matrix $D$ penalizes more on high frequency, when we set a larger $l_1$ penalty, the fetal ECG is removed. This is an important feature in that we can more adaptively determine which component is associated with the mECG. Therefore, the third limitation, (L3), of the de-shape STFT is alleviated. We will explore its clinical value in the future work. 
%

\begin{figure}[h]
	\includegraphics[trim=110 0 120 0,clip,width=1\textwidth]{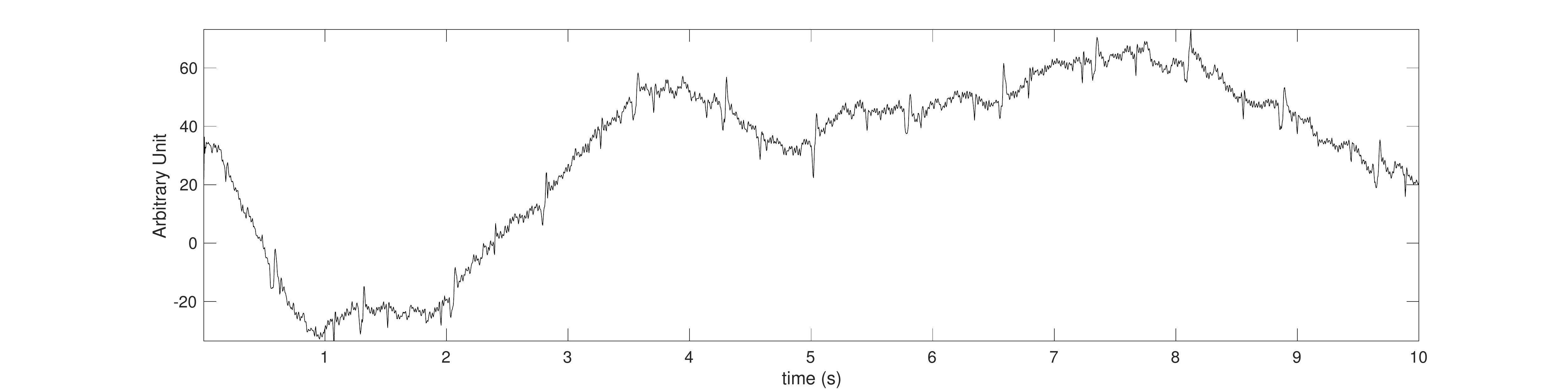}
	\includegraphics[trim=110 0 120 0,clip,width=1\textwidth]{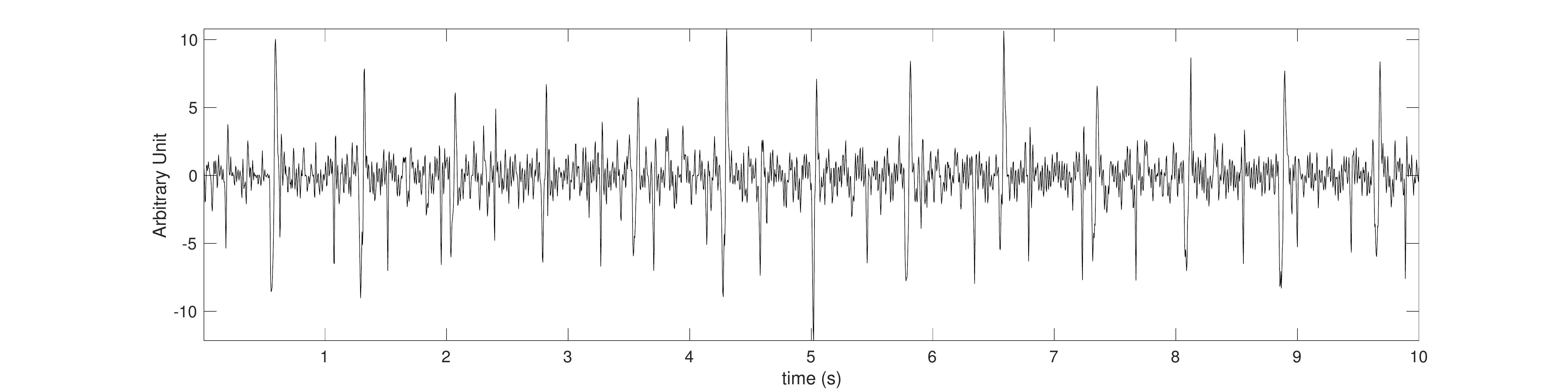}
	\includegraphics[trim=20 0 50 0,clip,width=0.49\textwidth]{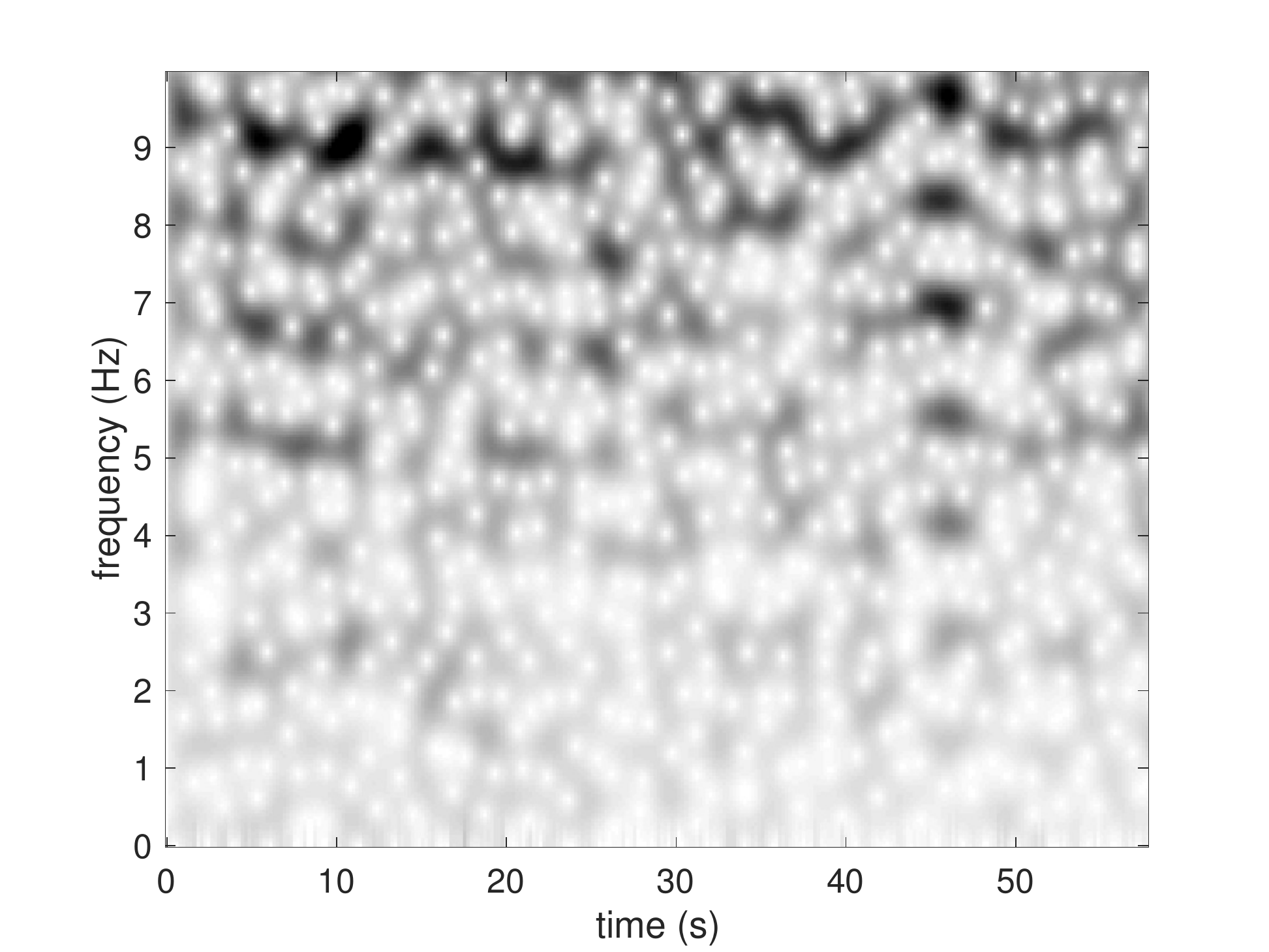}
	\includegraphics[trim=20 0 50 0,clip,width=0.49\textwidth]{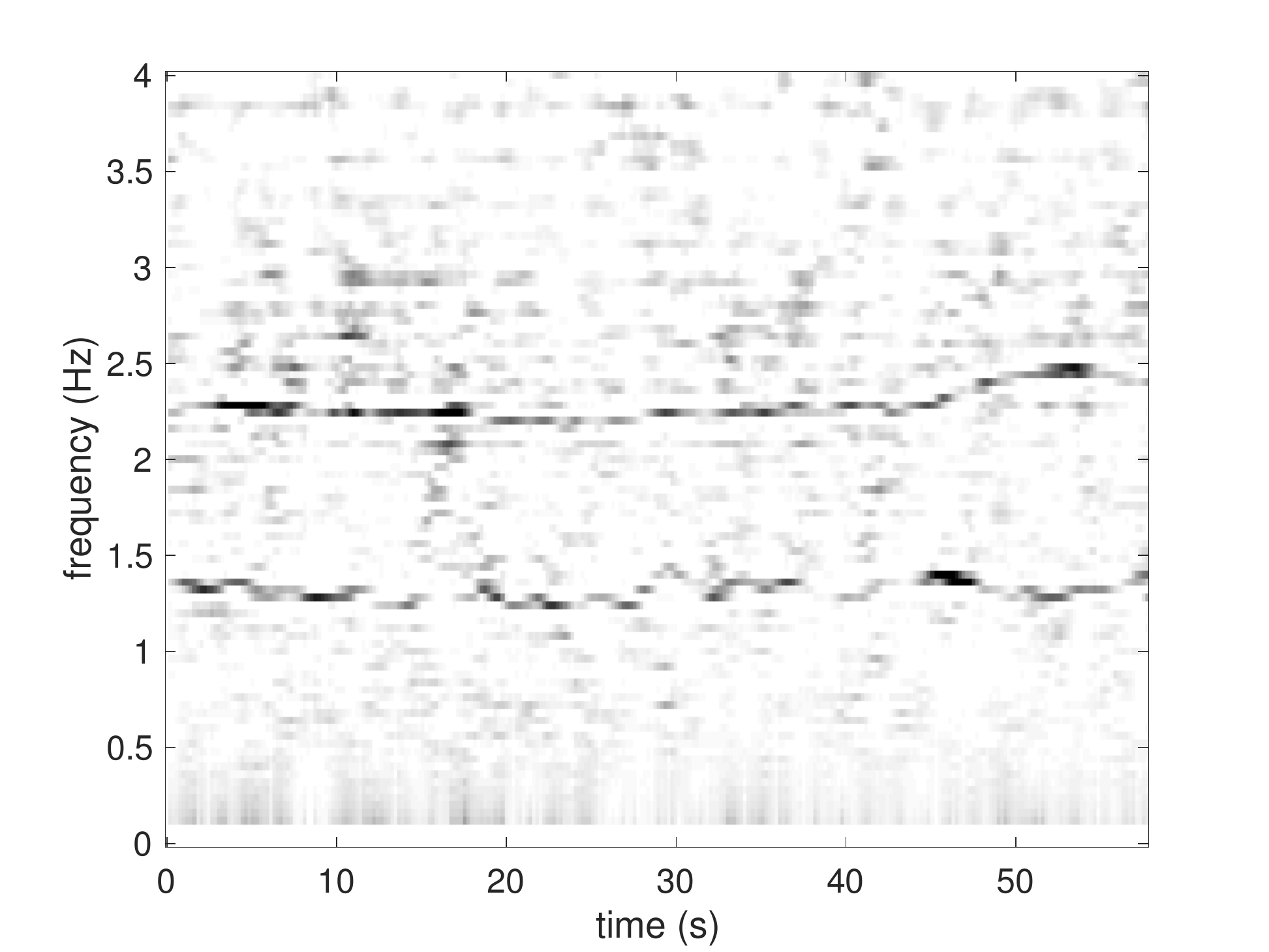}
	\caption{Top row: ta-mECG signal shown for 10 seconds. Second row: ta-mECG signal after removing the trend, denoted as $f$. Third row left: the spectrogram of $f$; third row right: the magnitude of the de-shape STFT of $f$. To enhance the visibility of de-shape STFT, we only show the frequency range up to 4Hz. We can see two clear curves around 1.5Hz and 2.5Hz, which are associated with the maternal instantaneous heart rate (IHR) and fetal IHR respectively.}
	\label{fig:fECG-1}
\end{figure}

\begin{figure}[h]
	\includegraphics[trim=20 0 50 0,clip,width=0.49\textwidth]{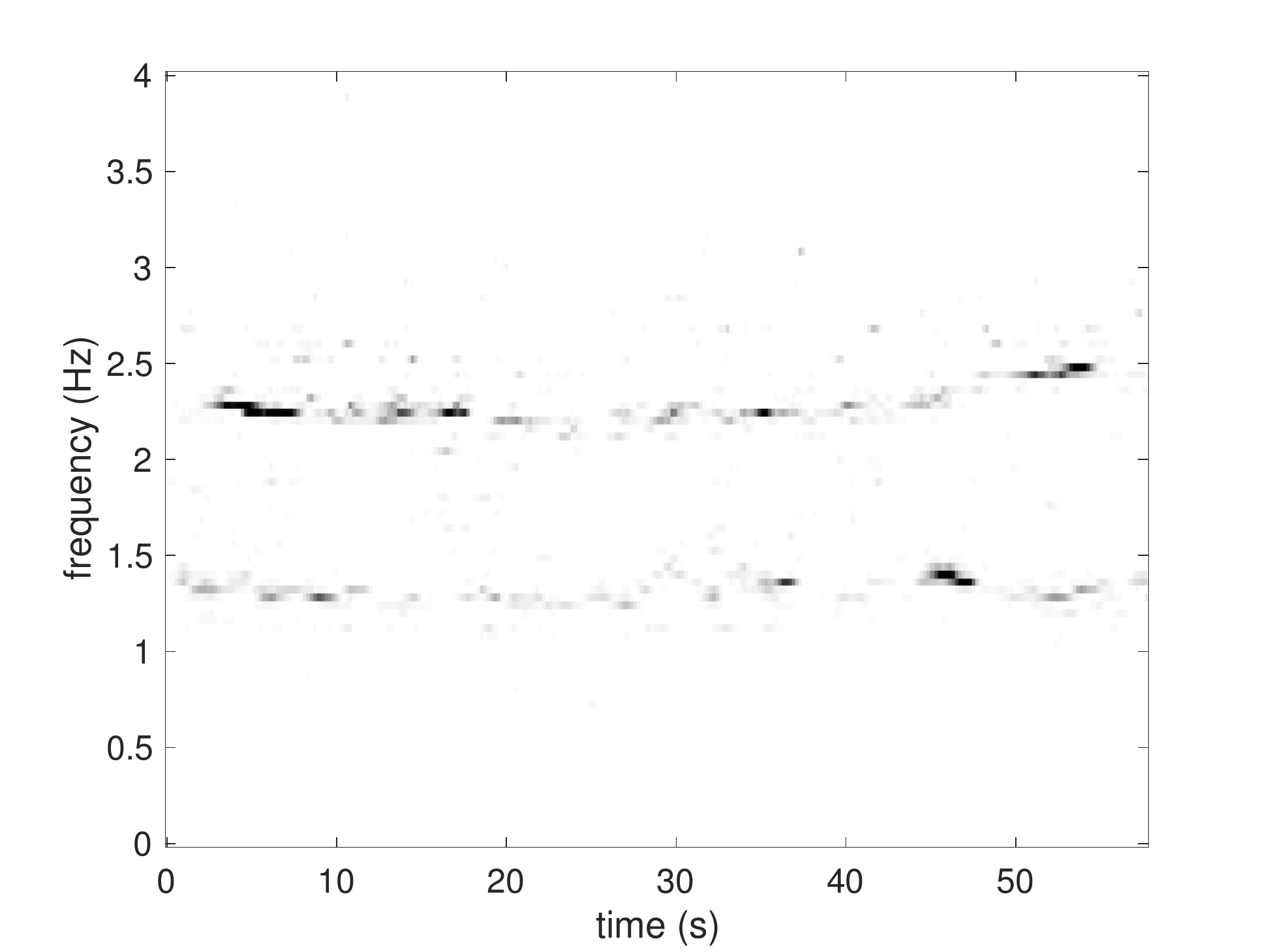}
	\includegraphics[trim=20 0 50 0,clip,width=0.49\textwidth]{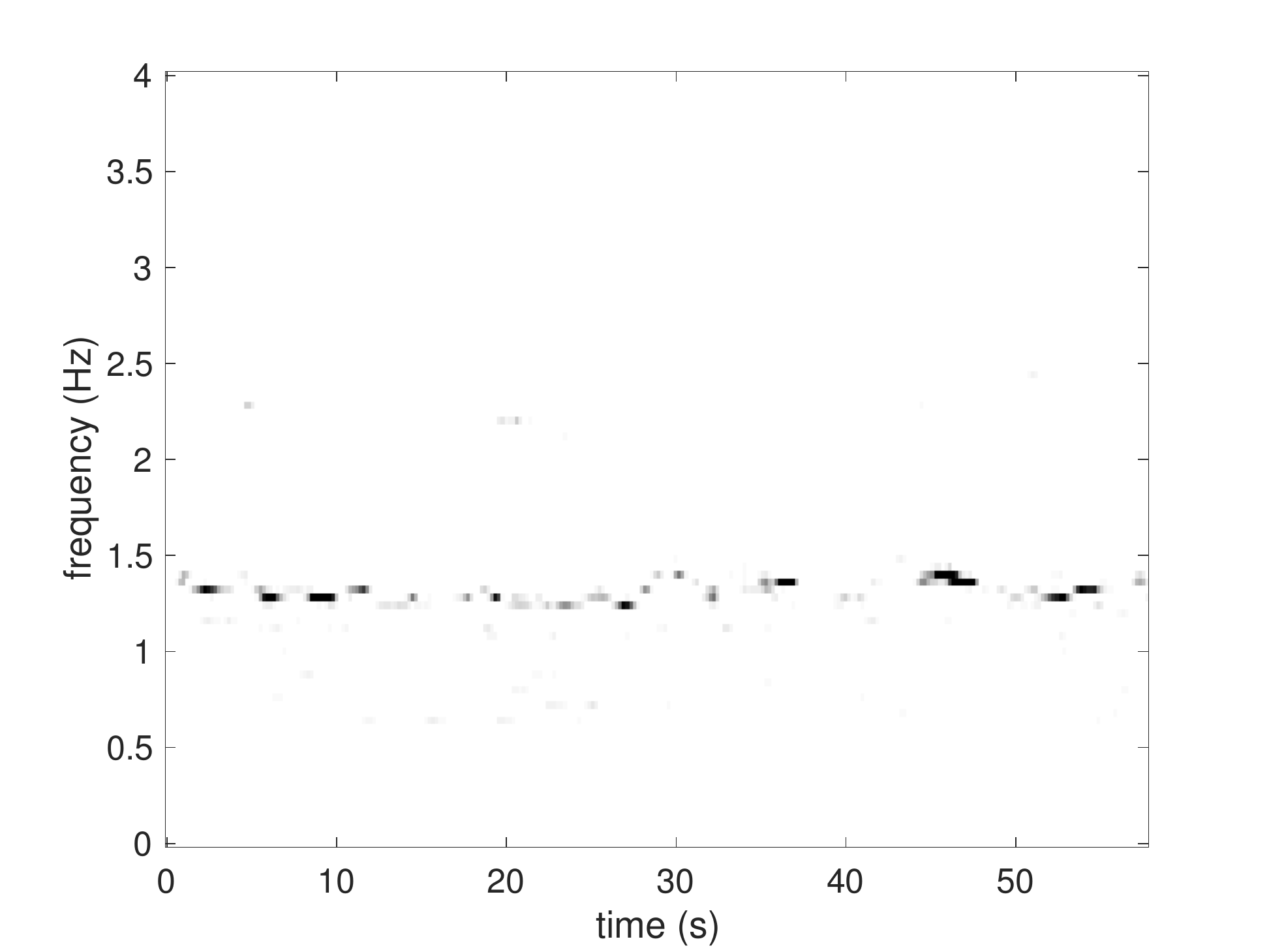}
	\includegraphics[trim=20 0 50 0,clip,width=0.49\textwidth]{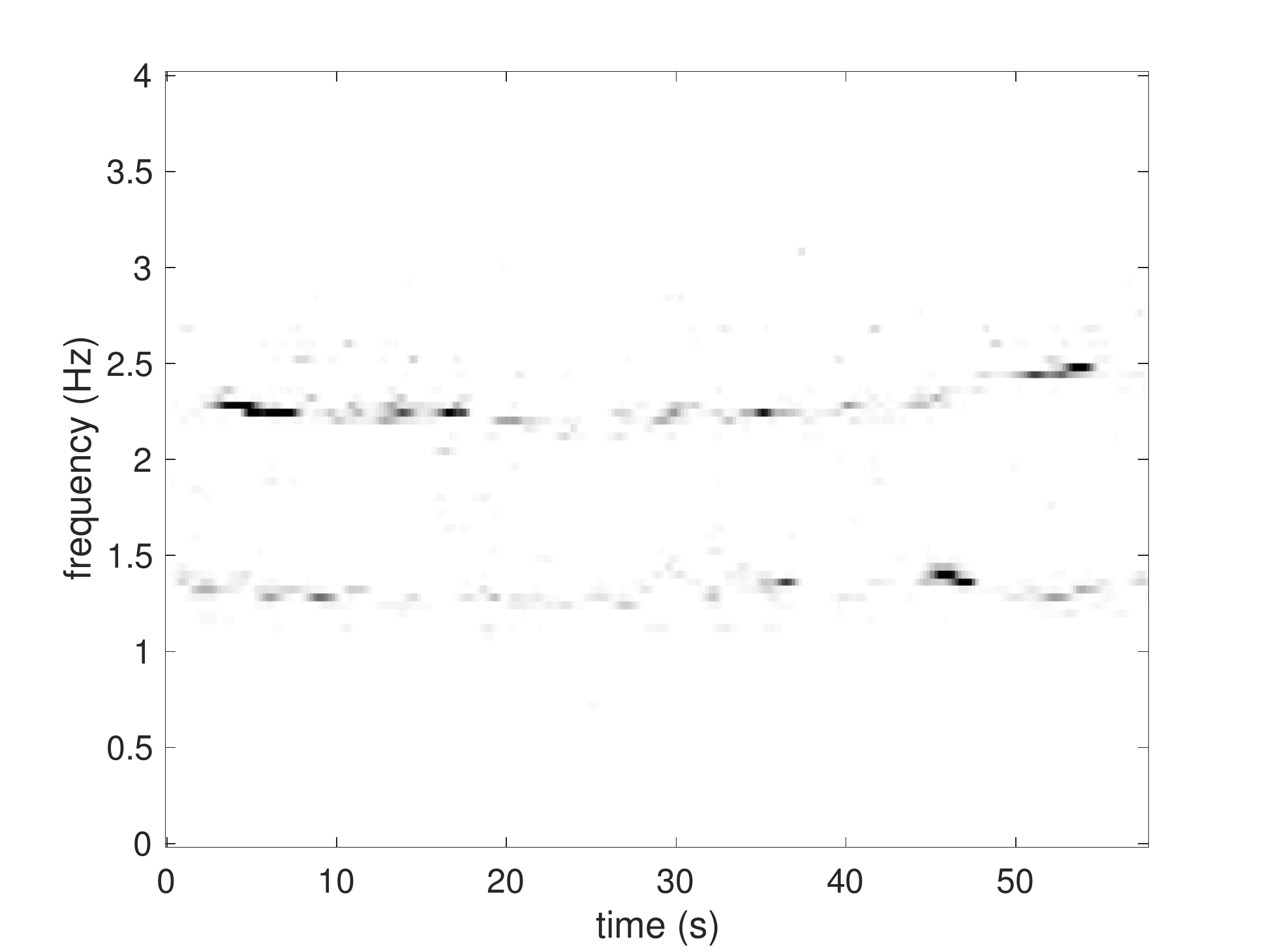}
	\includegraphics[trim=20 0 50 0,clip,width=0.49\textwidth]{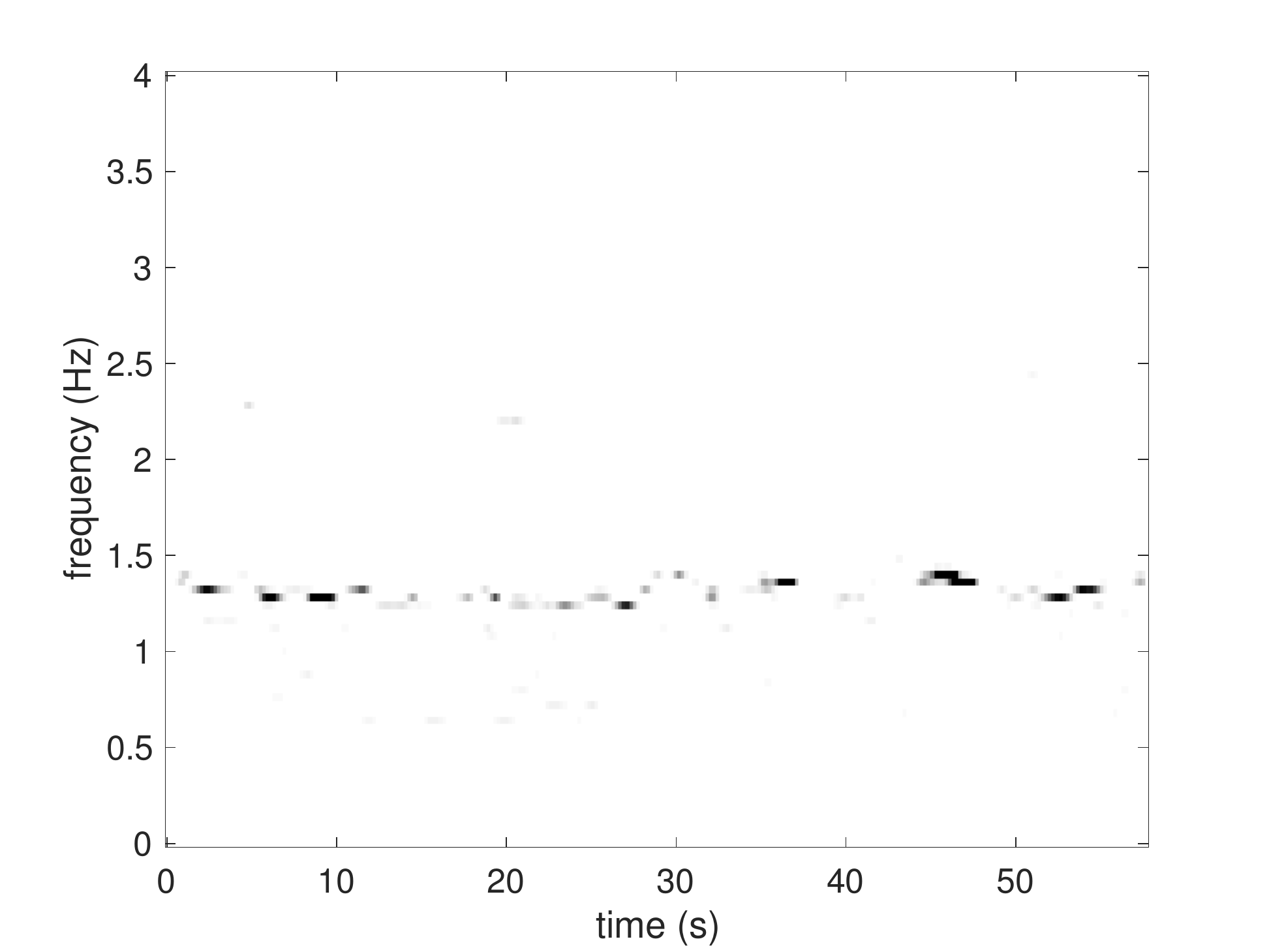}
	\caption{The RDS and  vRDS of the ta-mECG signal $f$ shown in Figure \ref{fig:fECG-1}. Top left: RDS of $f$ when $\lambda = 0.01$; top right: RDS of $f$ when $\lambda = 0.03$; Bottom left:  vRDS of $f$ when $k=1, \lambda = 0.03$; bottom right:  vRDS of $f$ when $k=1, \lambda = 0.09$. It is clear that when we have a larger penalty term, the maternal IHR information is further enhanced. To enhance the visibility, we only show the frequency range up to 4Hz. Note that compared with the de-shape STFT shown in Figure \ref{fig:fECG-1}, the TFR determined by RDS is cleaner; that is, there are less background artifact. Also note that the artifacts in RDS are less enhanced in  vRDS.}
	\label{fig:fECG-2}
\end{figure}

\section{Theoretical results}\label{Section Theoretical results}

In this section, we theoretically analyze several natural questions associated with the proposed RPT, and hence the RDS. The analysis for vRPT is direct, and we omit it. First, the solution to (\ref{lasso}) should be unique (the existence follows from \cite[Lemma 1]{tibshirani2013lasso}); otherwise the construction of $\mathbf{X_f}$ would be meaningless. Second, this approach can accurately reveal the underlying periods (and hence the fundamental frequency) if the input signal, $|V_f^{(h)}(t,\cdot)|^{\gamma}$, is oscillatory but modulated by the spectral envelope and corrupted by noise. Third, it is also important to show that the RPT is robust to the jitter; that is, the oscillation happens around but not precisely on multiples of the period. Last, we want to better know the role of $\lambda$ in the program. 

We start with preparing some notation. Recall the notation $A$, $B$ and $D$ used in \eqref{lasso}. Assume that a signal $y\in\mathbb{R}^{N}$ consists of $l$ oscillatory components coming from $\mathcal{R}_{p_1,N}, \mathcal{R}_{p_2,N}, \dots, \mathcal{R}_{p_l,N}$, where $1\leq p_1<p_2<\ldots<p_l$, so that
\[
y=Bx^*=AD^{-1}x^*\,.
\] 
Suppose 
\[
Bx^{*} = B_{S}x_{S}^{*}, 
\]
where $B_{S}$ is an $N\times (\sum\limits_{i=1}^{l}\phi(p_i))$ submatrix of $B$ with the column index $SOP(x^{*})$, and $x_{S}^{*}$ is a subvector of $x^{*}$ with the index $SOP(x^{*})$.  Denote 
\[
x_S^*=\begin{bmatrix}x^*_{p_1}\\\vdots\\x^*_{p_l}\end{bmatrix}, 
\]
where $x^*_{p_i}\in \mathbb{R}^{\phi(p_i)}$.
Now, for $i=1,\ldots,l$, set 
\[
b_{p_i,N}^{(k)} := \frac{1}{\zeta(p_i)} c_{p_i,N}^{(k)},
\]
where $0\leq k\leq \phi(p_{i})-1$. With the above setup, we thus have
\[
y = \sum\limits_{i=1}^{l}\sum\limits_{j=0}^{\phi(p_i)-1}b_{p_i,N}^{(j)} x_{p_i}^*(j+1), 
\] 
where at least one of the $x_{p_i}^*(j)$, $1\leq j \leq \phi(p_i)$, for each $p_i$ is nonzero. 
Throughout this paper, we further assume that $N$ is sufficiently large compared with $p_l$ so that $B_S$ has full column rank, according to Proposition \ref{prop4}. We mention that this is a reasonable and practical assumption since the signal should have enough length, compared to the underlying periods, so that we can be sure that there exists an oscillation. 
Under this assumption, we would have that $B_S^T B_S$ is invertible.

We use $IP_{x^*}^c$ to denote the complement set of $IP_{x^*}$ in $\{ 1, 2, \dots, P_{max}\}$ and $S^c$ to denote the complement of the index set $S$ in $\{ 1, 2, \dots \Phi(P_{max})\}$. As a result, $y$ can be written as 
\[
y = \begin{bmatrix} B_S & B_{S^c} \end{bmatrix} \begin{bmatrix} x_S^* \\ 0 \end{bmatrix}. 
\]
Define
\[
M_{S,S} = \max_{\substack{{i,j:} \\ {p_i, p_j\in IP_{x^*}}}}\max \left\{ \left|\frac{\sum\limits_{n = 1}^{m}c_{p_i}^{(s)}(n)c_{p_j}^{(t)}(n)}{\zeta(p_i)\zeta(p_j)}\right| \;
\begin{tabular}{|l}
$0\leq m\leq lcm(p_i,p_j)-1$; \\
$0\leq s \leq \phi(p_i)-1$; \\
$0\leq t\leq \phi(p_j)-1$.  
\end{tabular}
\right\},
\] 
and
\[
M_{S,S^c} = \max_{\substack{{i,j:} \\ {p_i\in IP_{x^*}} \\ {p_j\in IP_{x^*}^c}}}\max \left\{ \left|\frac{\sum\limits_{n = 1}^{m}c_{p_i}^{(s)}(n)c_{p_j}^{(t)}(n)}{\zeta(p_i)\zeta(p_j)}\right| \;
\begin{tabular}{|l}
$0\leq m\leq lcm(p_i,p_j)-1$; \\
$0\leq s \leq \phi(p_i)-1$; \\
$0\leq t\leq \phi(p_j)-1$.  
\end{tabular}
\right\}.
\] 
$M_{S,S}$ is the upper bound for the truncated correlation within $SOP(x^*)$, while $M_{S,S^c}$ is the upper bound for the truncated correlation between columns in the $SOP(x^{*})$ and those not in $SOP(x^{*})$.

In the following analysis, we assume $P_{max}\in \mathbb{N}^{+}$ is fixed. Suppose $\tilde{y}$ is some noisy version of $y$. Below, we will show the robustness of the proposed Ramanujan-based PT algorithm, and see that the solution to 
\begin{equation}\label{ptlasso}
\argmin\limits_{x\in \mathbb{R}^{\Phi(P_{max})}} \frac{1}{2}\norm{\tilde{y} - Bx }_2^2 + \lambda \norm{ x }_1
\end{equation}
could reveal those periods associated with $y$.
We prepare some technical lemmas. 

\begin{lemma}\label{lemma:1}
	Let $M = \max\{M_{S,S}, M_{S,S^c}\}$, $\lambda_i$ be the minimal eigenvalue of $C_{p_i}^T C_{p_i}$ and 
	\[
	K := \min\limits_i \floor*{\frac{N}{p_i}} \frac{\lambda_i}{\zeta^2(p_i)} - n_{S}M, 
   \]
	where $n_S := \sum\limits_{i=1}^{l}\phi(p_i)$. Then, the minimal eigenvalue of $B_{S}^{T}B_{S}$ has a lower bound $\lambda_{min}(B_{S}^{T}B_{S})\geq K$.
\end{lemma}
\begin{proof}
	Write $B_S = \begin{bmatrix} B_1 & B_2 & \cdots & B_l \end{bmatrix}$, where $B_i = \frac{1}{\zeta(p_i)}C_{p_i,N}$. Then, $B_{S}^{T}B_{S}$ is an $l\times l$ block matrix whose $(i,j)$-th block is $B_{i}^{T}B_{j}$. Let $\tilde{B}_i$ be the upper $p_i\floor*{\frac{N}{p_i}}\times\phi(p_i)$ submatrix of $B_i$, and set $H$ to be 
	\[  
	H = 
	\begin{bmatrix}
	\tilde{B}_1^T \tilde{B}_1 & & \\
	& \tilde{B}_2^T \tilde{B}_2  &\\
	&  & \ddots \\
	& & & \tilde{B}_l^T \tilde{B}_l
	\end{bmatrix};
	\]
	that is, a block diagonal matrix of the same size as $B_{S}^{T}B_{S}$ whose diagonal entries are $\tilde{B}_i^T \tilde{B}_i$.
	By definition, the minimum eigenvalue $\lambda_{min}(H)$ of $H$ is 
	\[ 
	\lambda_{min}(H) = \min\limits_i \, \floor*{\frac{N}{p_i}} \frac{\lambda_i}{\zeta^2(p_i)}, 
	\]
	where $i = 1, 2, \dots, l$.
	Denote 
	\[
	P := B_{S}^{T}B_{S} - H. 
	\]
	By the orthogonality property of the Ramanujan sums, any entry of $P$ is either a truncated correlation within $SOP(x^*)$ (in diagonal blocks) or a truncated correlation between some column in the $SOP(x^{*})$ and the other not in $SOP(x^{*})$ (off diagonal block). Hence any entry of $P$ is bounded by $M$ so that $\norm{P}_2 \leq (\sum\limits_{i=1}^{l}\phi(p_i))M$. Therefore, by Weyl's inequality we have 
	\[ 
	|\lambda_{min}(B_S^T B_S) - \lambda_{min}(H)| \leq   n_{S}M. 
	\]
	As a result, we have 
	\[ 
	\lambda_{min}(B_S^T B_S)\geq \lambda_{min}(H) - (\sum\limits_{i=1}^{l}\phi(p_i))M = \min_i \floor*{\frac{N}{p_i}} \frac{\lambda_i}{\zeta^2(p_i)} - n_{S}M.
	\]
\end{proof}

Since the term $n_{S}M$ does not depend on $N$ and $p_1,\ldots,p_l$ are fixed, $\min_i \floor*{\frac{N}{p_i}} \frac{\lambda_i}{\zeta^2(p_i)} - n_{S}M$ grows asymptotically as $N$ increases.
Based on Lemma \ref{lemma:1}, we have 

\begin{lemma}\label{lemma:2}
	Suppose the signal length N is sufficiently large such that $K =  \min_i \floor*{\frac{N}{p_i}} \frac{\lambda_i}{\zeta^2(p_i)} - n_{S}M > 0$. Then we have
	\[ 
	\norm{B_{S^c}^T B_S(B_S^T B_S)^{-1}}_{\infty}\leq \frac{n_{S}^2 M_{S,S^c}}{K}. 
	\]
\end{lemma}
\begin{proof}
	Since $B_S^T B_S$ is positive definite when $N$ is sufficient large, we have the eigen-decomposition as
	\[
	B_S^T B_S = U\Lambda U^T =
	\begin{bmatrix}
	u_1 & u_2 & \cdots & u_{n_S}
	\end{bmatrix}
	\begin{bmatrix}
	\lambda_1 & &  & \\
	& \lambda_2 & & \\
	& & \ddots & \\
	& & & \lambda_{n_S}
	\end{bmatrix}
	\begin{bmatrix}
	u_1^T \\
	u_2^T \\
	\vdots \\
	u_{n_S}^T 
	\end{bmatrix},
	\]
	where $U\in O(n_S)$, $u_i$ is the $i$-th eigenvector associated with the $i$-th eigenvalue $\lambda_i$, and $\lambda_1\geq \lambda_2\geq\ldots\geq \lambda_{n_S}\geq K>0$. Therefore,  
	\[
	(B_S^T B_S)^{-1} =
	\begin{bmatrix}
	u_1 & u_2 & \cdots & u_{n_S}
	\end{bmatrix}
	\begin{bmatrix}
	\frac{1}{\lambda_1} & &  & \\
	& \frac{1}{\lambda_2} & & \\
	& & \ddots & \\
	& & & \frac{1}{\lambda_{n_S}}
	\end{bmatrix}
	\begin{bmatrix}
	u_1^T \\
	u_2^T \\
	\vdots \\
	u_{n_S}^T 
	\end{bmatrix}.
	\]
	Moreover, let 
	\[
	B_S =
	\begin{bmatrix}
	\tilde{b}_1 & \tilde{b}_2 & \dots & \tilde{b}_{n_S},
	\end{bmatrix}
	\]
	and the $m$-th row of $B_{S^c}^T$ be denoted as $a^T$, where $a$ is the $m$-th column of $B_{S^c}$. Then, the $(m,k)$-th entry of $B_{S^c}^T B_S(B_S^T B_S)^{-1}$ is equal to
	\[
	\begin{bmatrix}
	a^T\tilde{b}_1 & \dots & a^T\tilde{b}_{n_S}
	\end{bmatrix}
	\begin{bmatrix}
	\frac{1}{\lambda_1}u_1 & \dots & \frac{1}{\lambda_{n_S}}u_{n_S} 
	\end{bmatrix}
	\begin{bmatrix}
	u_{1,k} \\
	\vdots \\
	u_{n_S,k}
	\end{bmatrix}
	=
	\sum\limits_{j=1}^{n_S}\sum\limits_{i=1}^{n_S}\frac{1}{\lambda_j}a^T\tilde{b}_i u_{j,i} u_{j,k},
	\]
	where $u_{i,j}$ is the $j$-th entry of $u_i$. Hence, the $l^1$ norm of the $m$-th row of $B_{S^c}^T B_S(B_S^T B_S)^{-1}$ is $\sum\limits_{k=1}^{n_S}\left| \sum\limits_{j=1}^{n_S}\sum\limits_{i=1}^{n_S}\frac{1}{\lambda_j}a^T\tilde{b}_i u_{j,i} u_{j,k}\right|$, and
	\begin{align*}
	\sum\limits_{k=1}^{n_S}\left| \sum\limits_{j=1}^{n_S}\sum\limits_{i=1}^{n_S}\frac{1}{\lambda_j}a^T\tilde{b}_i u_{j,i} u_{j,k}\right|
	&\leq \frac{1}{K}\sum\limits_{k=1}^{n_S}\sum\limits_{j=1}^{n_S}\sum\limits_{i=1}^{n_S} \left| a^T \tilde{b}_i \right| \left| u_{j,i}\right| \left| u_{j,k}\right| \quad\text{(by Lemma \ref{lemma:1})} \\
	&\leq \frac{M_{S,S^c}}{K} \sum\limits_{k=1}^{n_S}\sum\limits_{i=1}^{n_S}(\sum\limits_{j=1}^{n_S} \left| u_{j,i}\right| \left| u_{j,k}\right|) \\
	&\leq \frac{M_{S,S^c}}{K} \sum\limits_{k=1}^{n_S}\sum\limits_{i=1}^{n_S}1 \quad\text{(by Cauchy-Schwarz)} \\
	&= \frac{n_{S}^2 M_{S,S^c}}{K}.
	\end{align*}
	Therefore, 
	\[ 
	\norm{B_{S^c}^T B_S(B_S^T B_S)^{-1}}_{\infty}\leq\frac{n_{S}^2 M_{S,S^c}}{K}.
	\]
\end{proof}

Note that since $n_{S}$ and $M_{S,S^c}$ are fixed and $K=\Theta(N)$ when $N\to \infty$, $\norm{B_{S^c}^T B_S(B_S^T B_S)^{-1}}_{\infty}=O(N^{-1})$ when $N\to \infty$. We proceed to prove the main theorems using the technique similar to the Primal-Dual Witness construction in \cite{wainwright}, where $\norm{B_{S^c}^T B_S(B_S^T B_S)^{-1}}_{\infty} < 1$ is presumed and the probability of successful support recovery by Lasso under additive Gaussian noise is estimated.
\begin{theorem}[Robustness to envelope perturbations]\label{thm1}
	Suppose $N$ is sufficiently large such that $K>n_S^2 M_{S,S^c} $ and $\tilde{y} = Ey$, where $E = diag(e_1, e_2, \dots, e_N)$ is  the envelope. 
	Let 
	\begin{equation}\label{definition MSy1}
	M_{S,y}^1 := \norm{B_{S^c}^T(I-B_S (B_S^T B_S)^{-1}B_S^T)}_{\infty} \norm{y}_{\infty}
	\end{equation}
	and
	\begin{equation}\label{definition MSy2}
	M_{S,y}^2 := \norm{(B_S^T B_S)^{-1}B_S^T}_{\infty}\norm{y}_{\infty}.
	\end{equation}
	Then, if 
	\begin{equation}\label{condition MSy1}
	\max\limits_{i} |e_i -1|< \frac{\lambda}{M_{S,y}^1}\Big(1-\frac{n_S^2 M_{S,S^c}}{K}\Big)\,, 
	\end{equation}
	where $\lambda>0$ is the chosen penalty in (\ref{ptlasso}), the program (\ref{ptlasso}) has a unique solution $\hat{x}$ such that if $\hat{x}_i \neq 0$, then $i\in SOP(x^*)$. 
	
	Moreover, if for $i=1,2,\dots,l$ we denote 
	\begin{equation}\label{definition betai}
	\beta_i := \max\limits_j | x_{p_i}^*(j) |\,,
	\end{equation}
	then the solution $\hat{x}$ can successfully identify periods $p_i$ when $\beta_i$ satisfies 
	\begin{equation}\label{condition betai}
	\beta_i > \max\limits_{i} |e_i -1|\cdot M_{S,y}^2 + \lambda\norm{(B_S^T B_S)^{-1}}_{\infty}.
	\end{equation}
\end{theorem}

Before proving the theorem, let us have a closer look at this theorem. First, if $E=I$, that is, when there is no envelop variation, $\max\limits_{i} |e_i -1|=0$. Thus, for any $\lambda>0$, the program (\ref{ptlasso}) always has a unique solution $\hat{x}$ such that $SOP(\hat{x}) \subseteq SOP(x^*)$, so there is no spurious period detected. Moreover, if $\lambda$ is sufficiently small such that $\lambda < \frac{\min_i \beta_i}{\norm{(B_S^T B_S)^{-1}}_{\infty}}$, we can detect all the periods $p_i$, $1\leq i \leq l$.

Second, \eqref{definition MSy1} might rings the bell that $B_S (B_S^T B_S)^{-1}B_S^T\tilde {y}$ is the solution of minimizing $\min_{x_S} \|\tilde {y}-B_Sx_S\|_2$. Thus, $B_{S^c}^T(I-B_S (B_S^T B_S)^{-1}B_S^T)\tilde {y}$ is quantifying the deviation of $\tilde{y}$ from $y$ using the basis not indexed by $S$. However, as we will see in the proof, \eqref{definition MSy1} comes from controlling $\tilde y-y$ instead of $\tilde{y}$.

Third, for fixed coefficients $x_{p_i}^*(j)$, $\norm{y}_{\infty}$ is fixed. Hence by the definition of $B_S$ and $B_{S^c}$, it is clear that $M_{S,y}^{1}=\Theta(N)$, and $M_{S,y}^{2}=\Theta(1)$ when $N\to \infty$. So, if the strength of envelope modulation $\max_i |e_i -1|$ is fixed in \eqref{condition MSy1}, to ensure the uniqueness solution of (\ref{ptlasso}) and to avoid spurious periods, $\lambda$ should be $\Omega(N)$. This result provides a rule-of-thumb criterion to choose $\lambda$ when we solve the program. 

Finally, note that $\norm{(B_S^T B_S)^{-1}}_{\infty}=\Theta(\frac{1}{N})$, and a bound for $\norm{(B_S^T B_S)^{-1}}_{\infty}$ is $\norm{(B_S^T B_S)^{-1}}_{\infty} \leq \frac{1}{K}$, which can be proved in a way similar to that of Lemma \ref{lemma:2} by using its eigenvalue decomposition. As a result, by \eqref{condition betai}, the lower bound for $\beta_i$ in order to to successfully identify the period $p_i$ is $\Theta(\frac{\lambda}{N})$.

\begin{proof}
	To study \eqref{ptlasso}, we first look at the following restricted problem: 
	\begin{equation}\label{restrictlasso}
	\argmin\limits_{x\in \mathbb{R}^{n_S}} \frac{1}{2}\norm{\tilde{y} - B_{S}x }_2^2 + \lambda \norm{ x }_1 ,
	\end{equation}      
	where $\lambda>0$.
	By the assumption, $B_{S}$ is of full rank, and hence we have $B_{S}^{T}B_{S}$ is invertible. Therefore, the minimized function in (\ref{restrictlasso}) is strictly convex and the solution $\hat{x}_{S}$ to (\ref{restrictlasso}) is unique. Let 
	\[
	\hat{z}_{S} = \frac{1}{\lambda}(B_{S}^{T}\tilde{y} - B_{S}^{T}B_{S}\hat{x}_{S})
	\]
	be a subgradient of $\norm{x}_1$ at $\hat{x}_{S}$ by (\ref{characterization}). Set $\hat{x} = \begin{bmatrix} \hat{x}_{S} & 0 \end{bmatrix}^{T}$ and $\hat{z} = \begin{bmatrix} \hat{z}_{S} & \hat{z}_{S^c} \end{bmatrix}^{T}$, where $\hat{z}_{S^c}$ is to be determined so that $\hat{z}\in \partial\|\hat{x}\|_1$. We thus plug $\hat{x}$ and $\hat{z}$ into equation (\ref{characterization}), and get 
	\begin{equation}\label{EQ: BTBx-BTEBx+lambdaz=0}
	B^{T}B\hat{x}-B^{T}EBx^{*}+\lambda \hat{z} = 0,
	\end{equation}
	where we use the assumption that $y=Bx^{*}$.
	Let $I$ be the $N\times N$ identity matrix. Then \eqref{EQ: BTBx-BTEBx+lambdaz=0} can be rewritten as
	\[
	B^{T}B(\hat{x} - x^{*}) + B^{T}(I-E)Bx^{*} + \lambda\hat{z} = 0,
	\]
	which can also be written as
	\[
	\begin{bmatrix}
	B_{S}^{T} \\
	B_{S^{c}}^{T}
	\end{bmatrix}
	\begin{bmatrix}
	B_{S} & B_{S^{c}}
	\end{bmatrix}
	\begin{bmatrix}
	\hat{x}_{S} - x_{S}^{*}\\
	0
	\end{bmatrix}
	+
	\begin{bmatrix}
	B_{S}^{T} \\
	B_{S^{c}}^{T}
	\end{bmatrix}
	(I-E)
	\begin{bmatrix}
	B_{S} & B_{S^{c}}
	\end{bmatrix}
	\begin{bmatrix}
	x_{S}^{*} \\
	0
	\end{bmatrix}
	+
	\lambda
	\begin{bmatrix}
	\hat{z}_{S} \\
	\hat{z}_{S^c}
	\end{bmatrix}
	= 0,
	\]
	where we use the fact that $x^*_{S^c}=0$.
	We have from the top block that
	\[ 
	B_{S}^{T}B_{S}(\hat{x}_{S} - x_{S}^{*}) + B_{S}^{T}(I-E)B_{S}x_{S}^{*}+\lambda\hat{z}_{S} = 0, 
	\]
	which indicates
	\begin{equation}\label{thm1eq1}
	\hat{x}_{S} - x_{S}^{*} = -(B_{S}^{T}B_{S})^{-1}[B_{S}^{T}(I-E)B_{S}x_{S}^{*}+\lambda \hat{z}_{S}].
	\end{equation}
	On the other hand, we have from the bottom block that
	\begin{equation}\label{thm1eq1.1}
	B_{S^{c}}^{T}B_{S}(\hat{x}_{S}-x_{S}^{*}) + B_{S^{c}}^{T}(I-E)B_{S}x_{S}^{*}+\lambda \hat{z}_{S^{c}} = 0.  
	\end{equation}
	By plugging $\hat{x}_{S}-x_{S}^{*}$ from (\ref{thm1eq1}) to \eqref{thm1eq1.1}, since $B_{S}^{T}B_{S}$ is invertible and $\lambda>0$, we have solved $\hat{z}_{S^c}$ and get
	\begin{equation*}
	\hat{z}_{S^c} = B_{S^c}^{T}[B_{S}(B_{S}^{T}B_{S})^{-1}\hat{z}_{S}-\frac{1}{\lambda}(B_S(B_S^T B_S)^{-1}B_S^T-I)(E-I)B_S x_{S}^{*}].
	\end{equation*}
	We need to show that $\hat{z}$ is a subgradient of the $l^1$ norm at $\hat{x}$. Note that $\|\hat{z}_{S}\|_\infty\leq 1$. Thus, by a direct bound, we have
	\begin{align*}
	\norm{\hat{z}_{S^c}}_{\infty} 
	&\leq \norm{B_{S^c}^T B_S (B_S^T B_S)^{-1}}_{\infty} + \frac{1}{\lambda}\norm{B_{S^c}^T(I-B_S (B_S^T B_S)^{-1}B_S^T)}_{\infty}\norm{E-I}_{\infty} \norm{y}_{\infty} \\
	&< \frac{n_S^2 M_{S,S^c}}{K} + \frac{1}{\lambda}M_{S,y}^1\frac{\lambda}{M_{S,y}^1}(1-\frac{n_S^2 M_{S,S^c}}{K}) \quad\text{(by Lemma \ref{lemma:2})}\\
	&< 1.
	\end{align*}
	Hence $\hat{z}$ is a subgradient of the $l^1$ norm at $\hat{x}$ and $\hat{x}$ is a solution to (\ref{ptlasso}), and any solution $x$ to (\ref{ptlasso}) satisfies $x_i = 0$, if $i\in S^c$; that is, $x = \begin{bmatrix} x_S & 0 \end{bmatrix}^T$ where $x_S$ is a solution to program (\ref{restrictlasso}). Therefore, $\hat{x}$ is the unique solution to (\ref{ptlasso}) since $\hat{x}_S$ is the unique solution to (\ref{restrictlasso}).
	
	For the second part, we have from (\ref{thm1eq1}),
	\begin{align*}
	\norm{\hat{x} - x^*}_{\infty} = \norm{\hat{x}_S - x_S^*}_{\infty} 
	&\leq \norm{(B_S^T B_S)^{-1}B_S^T}_{\infty}\norm{E-I}_{\infty}\norm{y}_{\infty} + \lambda\norm{(B_S^T B_S)^{-1}}_{\infty} \\
	&\leq \max\limits_{i} |e_i -1|\cdot M_{S,y}^2 + \lambda\norm{(B_S^T B_S)^{-1}}_{\infty}
	\end{align*}
	Therefore, if for some period $p_i$, there exists a coefficient $| x_{p_i}^*(j)| > \max\limits_{i} |e_i -1|\cdot M_{S,y}^2 + \lambda\norm{(B_S^T B_S)^{-1}}_{\infty}$, then $EOP_{\hat{x}}(p_i)>0$.
\end{proof}

Next, we show the robustness of (\ref{ptlasso}) to jitter; that is, when the periods between two consecutive cycles are not fixed and deviated by some perturbation. For example, a $p$-periodic signal with peaks at $p, 2p, 3p, \dots$ is jittered (or its peak locations are perturbed) so that the peaks arise at $p+\epsilon_1,2p+\epsilon_2,3p+\epsilon_3,\dots$, where $\epsilon_i$ is bounded and $|\epsilon_i|<p/2$. 

\begin{theorem}[Robustness to jitter]\label{Theorem robustness to jitter}
	Suppose $N$ is sufficiently large such that $K>n_S^2 M_{S,S^c}$, and $\tilde{y} = \sum\limits_{i=1}^l V_i y_i$, where $y_i = \sum\limits_{j=0}^{\phi(p_i)-1}b_{p_i,N}^{(j)} x_{p_i}^*(j+1)$, and $V_i$ is a diagonal block matrix:
	\[ V_i = 
	\begin{bmatrix} 
	V_{i,1} & & & \\
	& V_{i,2} & & \\
	& & \ddots & \\
	& & & V_{i,k_i} \\
	\end{bmatrix},\]
	where each $V_{i,j}$ is either a identity matrix or a permutation matrix.  
	Then, if
	\begin{equation}\label{condition MSyi1} 
	\lambda > \frac{2K\sum_{i=1}^l M_{S,y_i}^1}{K-n_S^2 M_{S,S^c}},
	\end{equation}
	 the program (\ref{ptlasso}) has a unique solution $\hat{x}$ such that if $\hat{x}_i \neq 0$, then $i\in SOP(x^*)$. Moreover, if for $i=1,2,\dots,l$ we denote 
	 \begin{equation}
	 \beta_i = \max\limits_j | x_{p_i}^*(j)|,
	 \end{equation}
	 then the solution $\hat{x}$ can successfully identify periods $p_i$ if $\beta_i$ satisfies
	 \begin{equation}
	  \beta_i > \sum_{i=1}^l 2 M_{S,y_i}^2 + \lambda\norm{(B_S^T B_S)^{-1}}_{\infty}.
	  \end{equation}
\end{theorem}
\begin{proof}
	First we note that $\norm{V_i - I}_{\infty}\leq 2$ for any $i$. We can write each $y_i$ as $y_i = Bx_{S_i}^* = B_{S_i}x_{S_i}^*$, where $S_i$ are indices in $B$ that correspond to $\mathcal{R}_{p_i, N}$ so that $S = \bigcup_{i=1}^{l}S_i$. We follow the same procedure an notations as in the proof of previous theorem to obtain
	\[
	\begin{bmatrix}
	B_{S}^{T} \\
	B_{S^{c}}^{T}
	\end{bmatrix}
	\begin{bmatrix}
	B_{S} & B_{S^{c}}
	\end{bmatrix}
	\begin{bmatrix}
	\hat{x}_{S} - x_{S}^{*}\\
	0
	\end{bmatrix}
	+
	\sum_{i=1}^l
	\begin{bmatrix}
	B_{S}^{T} \\
	B_{S^{c}}^{T}
	\end{bmatrix}
	(I-V_i)
	\begin{bmatrix}
	B_{S_i} & B_{S_i^{c}}
	\end{bmatrix}
	\begin{bmatrix}
	x_{S_i}^{*} \\
	0
	\end{bmatrix}
	+
	\lambda
	\begin{bmatrix}
	\hat{z}_{S} \\
	\hat{z}_{S^c}
	\end{bmatrix}
	= 0.
	\]
	We have from the top block that
	\[ 
	B_{S}^{T}B_{S}(\hat{x}_{S} - x_{S}^{*}) + \sum_{i=1}^l B_{S}^{T}(I-V_i)B_{S_i}x_{S_i}^{*}+\lambda\hat{z}_{S} = 0, 
	\]
	which indicates
	\begin{equation}\label{thm2eq1}
	\hat{x}_{S} - x_{S}^{*} = -(B_{S}^{T}B_{S})^{-1}[\sum_{i=1}^lB_{S}^{T}(I-V_i)B_{S_i}x_{S_i}^{*}+\lambda \hat{z}_{S}].
	\end{equation}
	On the other hand, we have from the bottom block that
	\begin{equation}\label{thm2eq1.1}
	B_{S^{c}}^{T}B_{S}(\hat{x}_{S}-x_{S}^{*}) + \sum_{i=1}^l B_{S^{c}}^{T}(I-V_i)B_{S_i}x_{S_i}^{*}+\lambda \hat{z}_{S^{c}} = 0.  
	\end{equation}
	By plugging $\hat{x}_{S}-x_{S}^{*}$ from (\ref{thm2eq1}) to \eqref{thm2eq1.1}, we get
	\begin{equation*}
	\hat{z}_{S^c} = B_{S^c}^{T}[B_{S}(B_{S}^{T}B_{S})^{-1}\hat{z}_{S}+\sum_{i=1}^l\frac{1}{\lambda}(B_S(B_S^T B_S)^{-1}B_S^T-I)(I-V_i)B_{S_i} x_{S_i}^{*}].
	\end{equation*}
	Thus, by a direct bound, we have 
	\begin{align*}
	\norm{\hat{z}_{S^c}}_{\infty} 
	&\leq \norm{B_{S^c}^T B_S (B_S^T B_S)^{-1}}_{\infty} + \sum_{i=1}^l\frac{1}{\lambda}\norm{B_{S^c}^T(I-B_S (B_S^T B_S)^{-1}B_S^T)}_{\infty}\norm{I-V_i}_{\infty} \norm{y_i}_{\infty} \\
	&< \frac{n_S^2 M_{S,S^c}}{K} + 2\sum_{i=1}^l M_{S,y_i}^1 \frac{K-n_S^2 M_{S,S^c}}{2K\sum_{i=1}^l M_{S,y_i}^1} \\
	&= 1.
	\end{align*}
	On the other hand, we have from \eqref{thm2eq1}
	\begin{align*}
	\norm{\hat{x} - x^*}_{\infty} = \norm{\hat{x}_S - x_S^*}_{\infty} 
	&\leq \sum_{i=1}^l \norm{(B_S^T B_S)^{-1}B_S^T}_{\infty}\norm{V_i-I}_{\infty}\norm{y_i}_{\infty} + \lambda\norm{(B_S^T B_S)^{-1}}_{\infty} \\
	&\leq \sum_{i=1}^l 2 M_{S,y_i}^2 + \lambda\norm{(B_S^T B_S)^{-1}}_{\infty}.
	\end{align*}
	
\end{proof}

%
%

\begin{theorem}[Robustness to additive noise]
Suppose $N$ is sufficiently large such that $K>n_S^2 M_{S,S^c}$ and $\tilde{y} = y + e$, where the entries of $e$ are i.i.d. Gaussian with mean 0 and variance $\sigma^2$, i.e. $e \sim \mathcal{N}(0,\,\sigma^2 I)$. Let 
	\begin{equation}
	C := \max_{1\leq i \leq \Phi(P_{max})} \norm{b_i}_2^2.
	\end{equation}
	Denote
	\begin{equation}
	\beta_i = \max\limits_j | x_{p_i}^*(j) |
	\end{equation}
	for $i=1,2,\dots,l$. 
	Then, with probability at least 
	\begin{equation}
	1 - (\Phi(P_{max})-n_S)e^{-\frac{\lambda^2 (K-n_S^2 M_{S,S^c})^2}{2C\sigma^2 K^2}},
	\end{equation}
	we have the following:
	
	a) the program (\ref{ptlasso}) has a unique solution $\hat{x}$ such that if $\hat{x}_i \neq 0$, then $i\in SOP(x^*)$. 
	
	b) the solution $\hat{x}$ can successfully identify the period $p_i$ if $\beta_i$ satisfies
    \begin{equation}\label{thm3betabound}
    \beta_i > \sqrt{\frac{\lambda^2(K-n_S^2 M_{S,S^c})^2}{K^4} + \frac{2C\sigma^2(\log n_S - \log(\Phi(P_{max})-n_S))}{K^2}} + \lambda\norm{(B_S^T B_S)^{-1}}_{\infty}.
    \end{equation}

\end{theorem}

Before we prove the theorem, we mention that if we choose $\lambda = \Omega(\sqrt{N})$, then the above probability tends to 1 as $N \to \infty$ since $K = \Theta(N)$ and $C = \Theta(N)$. Moreover, with the choice of $\lambda = \Omega(\sqrt{N})$, the right-hand side of \eqref{thm3betabound} will be $\Omega(1/\sqrt{N})$. The proof below is similar to the proof of \cite[Theorem 1]{wainwright}.

\begin{proof}
	First we note that for any $i$, $1\leq i \leq N$ and $t>0$ we have
	\begin{align*}
	\mathbb P(|e_i|\geq t) 
	& = \frac{2}{\sqrt{2\pi}\sigma} \int_t^{\infty} e^{-\frac{x^2}{2\sigma^2}}\diff{x} = \frac{2}{\sqrt{2\pi}\sigma} \int_0^{\infty} e^{-\frac{(x+t)^2}{2\sigma^2}}\diff{x} \\
	& \leq \frac{2}{\sqrt{2\pi}\sigma} e^{-\frac{t^2}{2\sigma^2}} \int_0^{\infty} e^{-\frac{x^2}{2\sigma^2}}\diff{x}  = e^{-\frac{t^2}{2\sigma^2}}.
	\end{align*}
	As at the beginning of the proof of Theorem \ref{thm1}, we have from (\ref{characterization}) that 
	\[
	B^TB\hat{x} - B^T(Bx^* +e) + \lambda\hat{z} = 0, 
	\]
	so that
	\[
	B^T B(\hat{x}-x^*)-B^T e + \lambda\hat{z} = 0.
	\]
	Consequently, 
	\[
	\begin{bmatrix}
	B_{S}^{T} \\
	B_{S^{c}}^{T}
	\end{bmatrix}
	\begin{bmatrix}
	B_{S} & B_{S^{c}}
	\end{bmatrix}
	\begin{bmatrix}
	\hat{x}_{S} - x_{S}^{*}\\
	0
	\end{bmatrix}
	-
	\begin{bmatrix}
	B_{S}^{T} \\
	B_{S^{c}}^{T}
	\end{bmatrix}
	e
	+
	\lambda
	\begin{bmatrix}
	\hat{z}_{S} \\
	\hat{z}_{S^c}
	\end{bmatrix}
	= 0.
	\]
	For the first part of the theorem, it is sufficient to prove $\norm{\hat{z}_{S^c}}_{\infty} < 1$. 
	
	We have from the top block that
	\[ B_{S}^{T}B_{S}(\hat{x}_{S} - x_{S}^{*}) - B_{S}^{T}e+\lambda\hat{z}_{S} = 0, \]
	which indicates
	\begin{equation}\label{thm3eq1}
	\hat{x}_{S} - x_{S}^{*} = (B_{S}^{T}B_{S})^{-1}(B_{S}^{T}e-\lambda \hat{z}_{S}).
	\end{equation}
	On the other hand, we have from the bottom block that
	\[ B_{S^{c}}^{T}B_{S}(\hat{x}_{S}-x_{S}^{*}) - B_{S^{c}}^{T}e+\lambda \hat{z}_{S^{c}} = 0.  \]
	Substitute $\hat{x}_{S}-x_{S}^{*}$ to have 
	\begin{equation}\label{thm3eq2}
	\hat{z}_{S^c} = B_{S^c}^{T}B_{S}(B_{S}^{T}B_{S})^{-1}\hat{z}_{S} + \frac{1}{\lambda}B_{S^c}^{T}(I-B_S(B_S^T B_S)^{-1}B_S^T)e.
	\end{equation}
	Note that each entry of $\frac{1}{\lambda}B_{S^c}^{T}(I-B_S(B_S^T B_S)^{-1}B_S^T)e$ has mean zero and variance at most $\frac{C\sigma^2}{\lambda^2}$, since $I-B_S(B_S^T B_S)^{-1}B_S^T$ is an orthogonal projection an thus has spectral norm equal to 1. Therefore, by a direct union bound we have 
	\begin{equation}
	\mathbb P\left(\,\norm{\frac{1}{\lambda}B_{S^c}^{T}(I-B_S(B_S^T B_S)^{-1}B_S^T)e}_{\infty}\geq t\right) \leq (\Phi(P_{max})-n_S)e^{-\frac{\lambda^2 t^2}{2C\sigma^2}},
	\end{equation}
	since there are $\Phi(P_{max})-n_S$ entries in $\frac{1}{\lambda}B_{S^c}^{T}(I-B_S(B_S^T B_S)^{-1}B_S^T)e$. 
	Let $t = 1 - \frac{n_S^2 M_{S,S^c}}{K}$, and we have 
	\begin{equation}
	\mathbb P\left(\,\norm{\frac{1}{\lambda}B_{S^c}^{T}(I-B_S(B_S^T B_S)^{-1}B_S^T)e}_{\infty}\geq 1 - \frac{n_S^2 M_{S,S^c}}{K}\right) \leq (\Phi(P_{max})-n_S)e^{-\frac{\lambda^2 (K-n_S^2 M_{S,S^c})^2}{2C\sigma^2 K^2}}\,,\nonumber
	\end{equation}
	so that with probability at least $1-(\Phi(P_{max})-n_S)e^{-\frac{\lambda^2 (K-n_S^2 M_{S,S^c})^2}{2C\sigma^2 K^2}}$, we have
	\begin{align*}
	\norm{\hat{z}_{S^c}}_{\infty} 
	&\leq \norm{B_{S^c}^T B_S (B_S^T B_S)^{-1}}_{\infty} + \norm{\frac{1}{\lambda}B_{S^c}^{T}(I-B_S(B_S^T B_S)^{-1}B_S^T)e}_{\infty} \\
	&< \frac{n_S^2 M_{S,S^c}}{K} + 1 - \frac{n_S^2 M_{S,S^c}}{K}\\
	&= 1.
	\end{align*}
	To prove the second part of the theorem, we first note that by Lemma \ref{lemma:2}, each entry of $(B_S^T B_S)^{-1}B_S^T e$ has mean zero and variance at most $\frac{C\sigma^2}{K^2}$. Therefore, by a direct union bound we have
	\begin{equation}
	\mathbb P\Big(\,\norm{(B_S^T B_S)^{-1}B_S^T e}_{\infty}> t\Big) \leq n_S e^{-\frac{K^2 t^2}{2C\sigma^2}}.
	\end{equation}
	Let 
	\[t = \sqrt{\frac{\lambda^2(K-n_S^2 M_{S,S^c})^2}{K^4} + \frac{2C\sigma^2(\log n_S - \log(\Phi(P_{max})-n_S))}{K^2}},
	\] 
	then with probability at least $1-(\Phi(P_{max})-n_S)e^{-\frac{\lambda^2 (K-n_S^2 M_{S,S^c})^2}{2C\sigma^2 K^2}}$, we have
	\begin{align*}
	&\norm{\hat{x} - x^*}_{\infty} 
	= \norm{\hat{x}_S - x_S^*}_{\infty} \\
	\leq&\, \norm{(B_S^T B_S)^{-1}B_S^T e}_{\infty} + \lambda\norm{(B_S^T B_S)^{-1}}_{\infty} \\
	\leq&\, \sqrt{\frac{\lambda^2(K-n_S^2 M_{S,S^c})^2}{K^4} + \frac{2C\sigma^2(\log n_S - \log(\Phi(P_{max})-n_S))}{K^2}} + \lambda\norm{(B_S^T B_S)^{-1}}_{\infty}\,.
	\end{align*}
\end{proof}

\section{Discussion and Conclusion}\label{Section Discussion}

We have proposed a novel TF analysis algorithm, RDS, which incorporates the RPT and $l^1$ penalized linear regression, to analyze complicated time series with non-sinusoidal oscillatory patterns. Numerical experiments demonstrate its usefulness in the extraction of the fundamental IF of a non-stationary signal. It alleviates some limitations of the traditional cepstrum-based de-shape algorithm, and it is supported by the provided robustness analysis of RPT to different perturbations. 
Specifically, the $l^1$ approach helps stabilize the impact of inevitable noise; the robustness of the RPT to the spectral envelop fluctuation helps reduce the inevitable low quefrency artifact issue when we apply the STCT; the fact that the PT preserves the oscillatory component energy helps resolve the ``weak fundamental component'' issue. In addition to these benefits, the computational complexity of this algorithm is reasonably fast, and can be easily scaled to long time series.

We have a comment here regarding directly applying the PT to the high frequency time series, particularly the biomedical signals. In general, any oscillatory component in a biomedical signal does not oscillate with a fixed period. Instead, its period changes from time to time. Since the PT is a global method, directly applying the PT to extract periods from such time series might not be suitable. Note that this does not contradict the provided robustness property of RPT against jitter shown in Theorem \ref{Theorem robustness to jitter}. This is because the time-varying period cannot be well approximated by the jitter effect. 

It is natural to ask if we can generalize the PT to short periods, and come out with the idea of ``short-time PT''; that is, we run the PT over a short period around the time we want to study how fast it oscillates, and stitch all PT together and generate a two-dimensional matrix that we can call the {\em time-period representation}. See Figure \ref{fig:STPT} for a preliminary result of this idea when we apply it to the ta-mECG $\mathbf f$ shown in Figure \ref{fig:1}, where $\mathbf f(i)=f(i\Delta_t)$, $f$ is the ta-mECG, $\Delta_t=1/250$ second, and $i=1,\ldots, N$. Numerically, for the $i$-th sampling time, $[\mathbf f(i-500) \,\, \mathbf f(i-499) \,\, \mathbf f(i-500) \ldots \mathbf f(i) \ldots \mathbf f(i+499) \,\, \mathbf f(i+500)]$; that is, the window is chosen to be of 4 seconds long. Then, based on the knowledge of the maternal IHR, we choose $P_{max}=200$, which is associated with $250/200=1.25$Hz. After tuning parameters (although not optimized via a grid search), the final ``time-period representation'' results with different periodicity penalization functions are shown in Figure \ref{fig:STPT}. It is clear that we can get some useful information both for the fetal IHR, but the quality of the maternal IHR information is less ideal compared with the proposed RDS. We mention that we can choose other window lengths, like 10 seconds long, but the results are worse and we do not show them.  
%
We emphasize that we do not claim that the short-time PT does not work, since we have not extensively explored its possibility. This is the direction we would explore in our future work.

Another comment following the about discussion about the short-time PT, the RDS solution we propose can be viewed as a solution sitting between PT and short-time PT in the following sense. We first delegate the time-varying oscillatory periods information to the TF domain, where the oscillation is related to the spectrum of the non-sinusoidal oscillation, and this oscillation in the TF domain is more regular and can be well modeled by the jitter effect. Then, we can benefit from this nice periodicity property in the TF domain and apply the PT to extract the fundamental IF, and hence the instantaneous period.

\begin{figure}[!htbp]
	\begin{minipage}{19.5 cm}
	\hspace{-90pt}\includegraphics[trim=20 0 40 10,clip,width=0.32\textwidth]{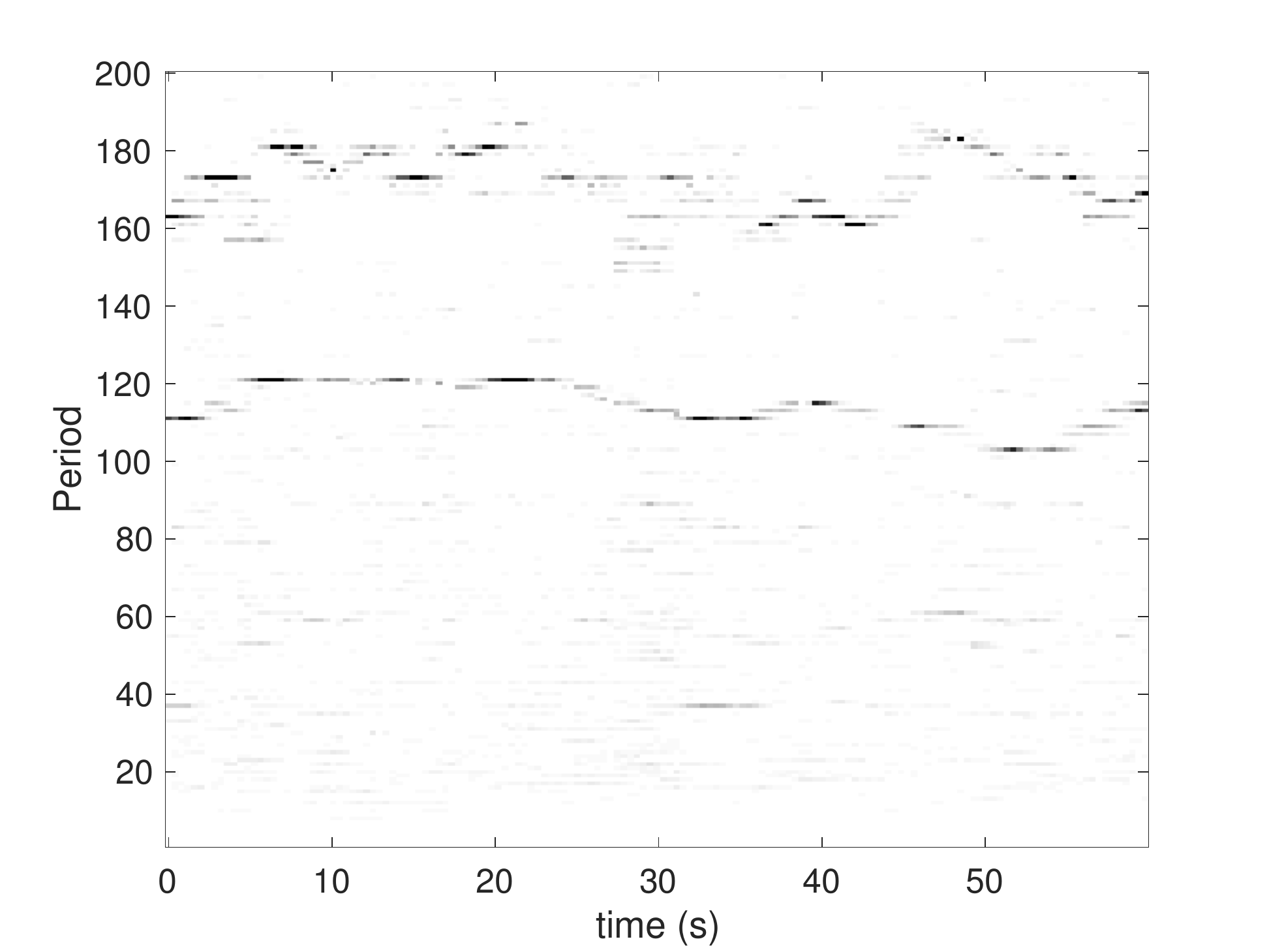}
	\includegraphics[trim=20 0 40 10,clip,width=0.32\textwidth]{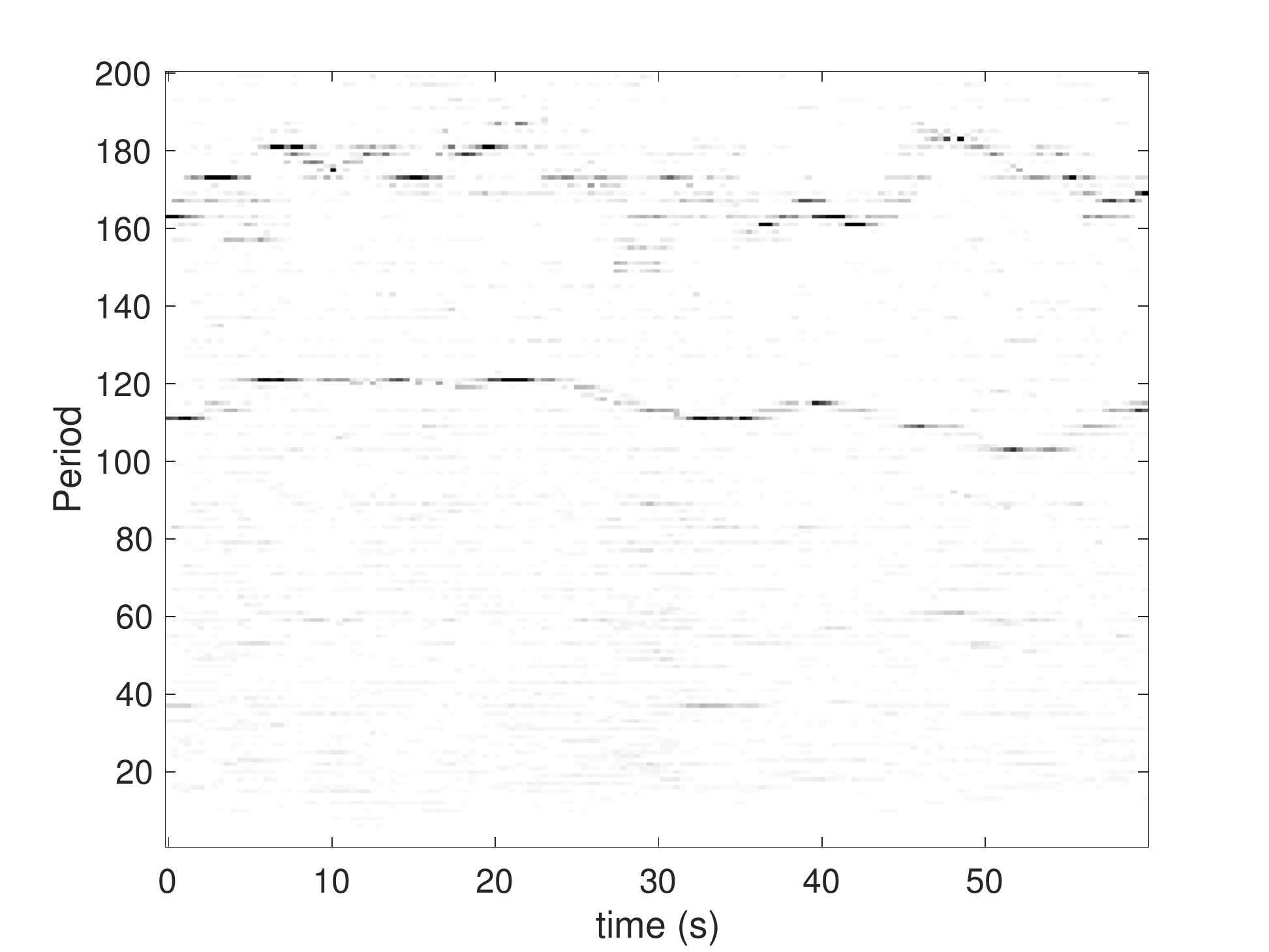}
	\includegraphics[trim=20 0 40 10,clip,width=0.32\textwidth]{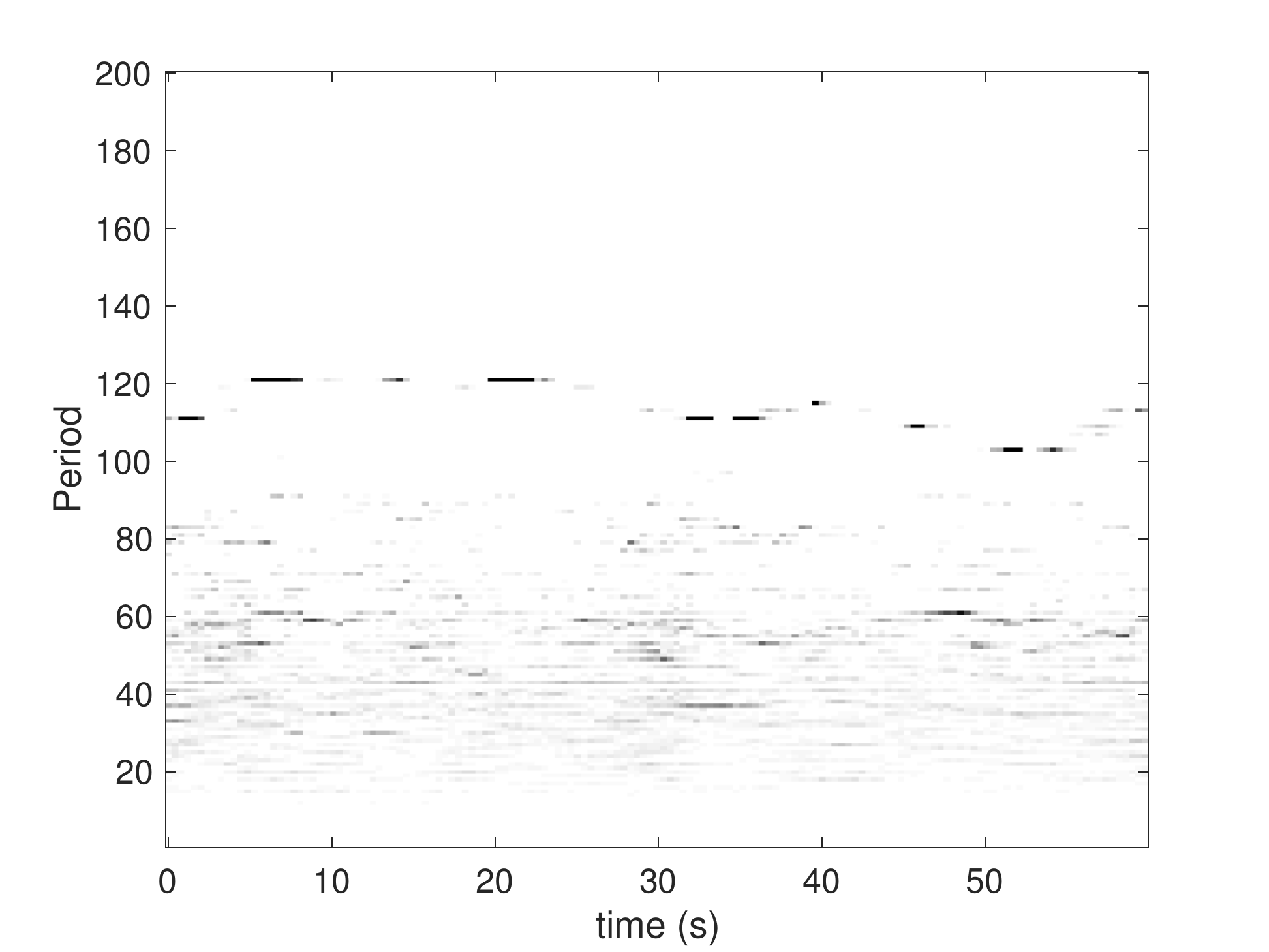}
	\end{minipage}
	\caption{A preliminary result of the short-time periodic transform (STPT) for the ta-mECG signal $f$ shown in Figure \ref{fig:1}. Here, for the $i$-th sampling time, we run the RPT on $[f((i-500)\Delta_t) \,\, f((i-499)\Delta_t) \,\, f((i-500)\Delta_t) \ldots f(i\Delta_t) \ldots f((i+499)\Delta_t) \,\, f((i+500)\Delta_t)]$ with $P_{max}=200$, which is associated with $250/200=1.25$Hz since $\Delta_t=250$. Left: $\zeta(p)=1$ with $\lambda=500$. Middle: $\zeta(p)=p$ with $\lambda=200$. Right: $\zeta(p)=p^2$ with $\lambda=10$. The curve around period 180 (about 1.39Hz) is associated with the maternal IHR, and the curve around period 120 (about 2.08Hz) is associated with the fetal IHR.}
	\label{fig:STPT}
\end{figure}

There are several future directions we will explore. From the theoretical perspective, we need to establish the statistical property of the RPT, and hence the RDS, so that we can carry out statistical inference. To the best of our knowledge, there is limited work in this direction, even for the basic PT. Note that the statistical property of de-shape algorithm is also missing.
Although we do not discuss extensively in the main article, we shall indicate the boundary problem, which is universal to all kernel-based algorithms. Around the boundary, the STFT would get distorted, and hence the RDS. How to handle this boundary effect is important and will be studied in the future work, particularly for the real-time implementation of the RDS. 
Although our theoretical results and preliminary numerical results support that the RDS resolves several limitations of the traditional de-shape approach, it is not clear how useful it is to real data. We will explore its performance in the real database and clinical problems in future work.

\section{Acknowledgement}
The authors would like to thank Shahar Kovalsky and John Malik for the discussion of numerical implementations.

\bibliographystyle{siam}
\bibliography{mybibfile.bib}

\begin{thebibliography}{10}

\bibitem{baumert2016qt}
{\sc M.~Baumert, A.~Porta, M.~A. Vos, M.~Malik, J.-P. Couderc, P.~Laguna,
  G.~Piccirillo, G.~L. Smith, L.~G. Tereshchenko, and P.~G. Volders}, {\em Qt
  interval variability in body surface ecg: measurement, physiological basis,
  and clinical value: position statement and consensus guidance endorsed by the
  european heart rhythm association jointly with the esc working group on
  cardiac cellular electrophysiology}, Europace, 18 (2016), pp.~925--944.

\bibitem{cepstrum}
{\sc B.~P. Bogert}, {\em The quefrency alanysis of time series for echoes;
  cepstrum, pseudo-autocovariance, cross-cepstrum and saphe cracking}, Time
  series analysis,  (1963), pp.~209--243.

\bibitem{Boyd}
{\sc S.~Boyd, N.~Parikh, E.~Chu, B.~Peleato, J.~Eckstein, et~al.}, {\em
  Distributed optimization and statistical learning via the alternating
  direction method of multipliers}, Foundations and Trends in Machine learning,
  3 (2011), pp.~1--122.

\bibitem{DNA1}
{\sc A.~K. Brodzik}, {\em Quaternionic periodicity transform: an algebraic
  solution to the tandem repeat detection problem}, Bioinformatics, 23 (2007),
  pp.~694--700.

\bibitem{DNA2}
{\sc M.~Buchner and S.~Janjarasjitt}, {\em Detection and visualization of
  tandem repeats in {DNA} sequences}, IEEE Transactions on Signal Processing,
  51 (2003), pp.~2280--2287.

\bibitem{bp}
{\sc S.~Chen and D.~Donoho}, {\em Basis pursuit}, in Proceedings of 1994 28th
  Asilomar Conference on Signals, Systems and Computers, vol.~1, IEEE, 1994,
  pp.~41--44.

\bibitem{Chen_Cheng_Wu:2014}
{\sc Y.-C. Chen, M.-Y. Cheng, and H.-T. Wu}, {\em Non-parametric and adaptive
  modelling of dynamic periodicity and trend with heteroscedastic and dependent
  errors}, J. R. Stat. Soc. Ser. B. Stat. Methodol., 76 (2014), pp.~651--682.

\bibitem{cicone2017nonlinear}
{\sc A.~Cicone and H.-T. Wu}, {\em How nonlinear-type time-frequency analysis
  can help in sensing instantaneous heart rate and instantaneous respiratory
  rate from photoplethysmography in a reliable way}, Frontiers in physiology, 8
  (2017), p.~701.

\bibitem{median}
{\sc G.~D. Clifford, F.~Azuaje, P.~McSharry, et~al.}, {\em Advanced methods and
  tools for ECG data analysis}, Artech house Boston, 2006.

\bibitem{SST}
{\sc I.~Daubechies, J.~Lu, and H.-T. Wu}, {\em Synchrosqueezed wavelet
  transforms: An empirical mode decomposition-like tool}, Applied and
  Computational Harmonic Analysis, 30 (2011), pp.~243--261.

\bibitem{ConceFT}
{\sc I.~Daubechies, Y.~Wang, and H.-t. Wu}, {\em Conce{FT}: Concentration of
  frequency and time via a multitapered synchrosqueezed transform},
  Philosophical Transactions of the Royal Society A: Mathematical, Physical and
  Engineering Sciences, 374 (2016), p.~20150193.

\bibitem{goldberger2000physiobank}
{\sc A.~L. Goldberger, L.~A. Amaral, L.~Glass, J.~M. Hausdorff, P.~C. Ivanov,
  R.~G. Mark, J.~E. Mietus, G.~B. Moody, C.-K. Peng, and H.~E. Stanley}, {\em
  Physiobank, physiotoolkit, and physionet: components of a new research
  resource for complex physiologic signals}, circulation, 101 (2000),
  pp.~e215--e220.

\bibitem{HouShi}
{\sc T.~Y. Hou and Z.~Shi}, {\em Extracting a shape function for a signal with
  intra-wave frequency modulation}, Philosophical Transactions of the Royal
  Society A: Mathematical, Physical and Engineering Sciences, 374 (2016),
  p.~20150194.

\bibitem{cepstrum1}
{\sc T.~Kobayashi and S.~Imai}, {\em Spectral analysis using generalised
  cepstrum}, IEEE Transactions on Acoustics, Speech, and Signal Processing, 32
  (1984), pp.~1235--1238.

\bibitem{SSTAPP3}
{\sc C.~Li and M.~Liang}, {\em Time-frequency signal analysis for gearbox fault
  diagnosis using a generalized synchrosqueezing transform}, Mechanical Systems
  and Signal Processing, 26 (2012), pp.~205--217.

\bibitem{li2019non}
{\sc F.~Li, W.~Song, C.~Li, and A.~Yang}, {\em Non-harmonic analysis based
  instantaneous heart rate estimation from photoplethysmography}, in ICASSP
  2019-2019 IEEE International Conference on Acoustics, Speech and Signal
  Processing (ICASSP), IEEE, 2019, pp.~1279--1283.

\bibitem{SuLi}
{\sc H.-W. Liao and L.~Su}, {\em Monaural source separation using ramanujan
  subspace dictionaries}, IEEE Signal Processing Letters, 25 (2018),
  pp.~1156--1160.

\bibitem{Deshape}
{\sc C.-Y. Lin, L.~Su, and H.-T. Wu}, {\em Wave-shape function analysis--when
  cepstrum meets time-frequency analysis}, Journal of Fourier Analysis and
  Applications, 24 (2018), pp.~451--505.

\bibitem{lin2019unexpected}
{\sc Y.-T. Lin, Y.-L. Lo, C.-Y. Lin, M.~G. Frasch, and H.-T. Wu}, {\em
  Unexpected sawtooth artifact in beat-to-beat pulse transit time measured from
  patient monitor data}, PloS one, 14 (2019).

\bibitem{lin2019wave}
{\sc Y.-T. Lin, J.~Malik, and H.-T. Wu}, {\em Wave-shape oscillatory model for
  biomedical time series with applications}, arXiv preprint arXiv:1907.00502,
  (2019).

\bibitem{SSTAPP2}
{\sc Y.-T. LIN, H.-T. WU, J.~Tsao, H.-W. YIEN, and S.-S. HSEU}, {\em
  Time-varying spectral analysis revealing differential effects of sevoflurane
  anaesthesia: non-rhythmic-to-rhythmic ratio}, Acta Anaesthesiologica
  Scandinavica, 58 (2014), pp.~157--167.

\bibitem{lobmaier2019fetal}
{\sc S.~M. Lobmaier, A.~M{\"u}ller, C.~Zelgert, C.~Shen, P.~Su, G.~Schmidt,
  B.~Haller, G.~Berg, B.~Fabre, J.~Weyrich, et~al.}, {\em Fetal heart rate
  variability responsiveness to maternal stress, non-invasively detected from
  maternal transabdominal ecg}, Archives of gynecology and obstetrics,  (2019),
  pp.~1--10.

\bibitem{recycling}
{\sc Y.~Lu, H.-t. Wu, and J.~Malik}, {\em Recycling cardiogenic artifacts in
  impedance pneumography}, Biomedical Signal Processing and Control, 51 (2019),
  pp.~162--170.

\bibitem{mainardi2008analysis}
{\sc L.~Mainardi, M.~Bertinelli, and R.~Sassi}, {\em Analysis of t-wave
  alternans using the ramanujan transform}, in 2008 Computers in Cardiology,
  IEEE, 2008, pp.~605--608.

\bibitem{mainardi2007application}
{\sc L.~Mainardi, L.~Pattini, and S.~Cerutti}, {\em Application of the
  ramanujan fourier transform for the analysis of secondary structure content
  in amino acid sequences}, Methods of information in medicine, 46 (2007),
  pp.~126--129.

\bibitem{malik2002relation}
{\sc M.~Malik, P.~F{\"a}rbom, V.~Batchvarov, K.~Hnatkova, and A.~Camm}, {\em
  Relation between qt and rr intervals is highly individual among healthy
  subjects: implications for heart rate correction of the qt interval}, Heart,
  87 (2002), pp.~220--228.

\bibitem{cepreview}
{\sc A.~V. Oppenheim and R.~W. Schafer}, {\em From frequency to quefrency: A
  history of the cepstrum}, IEEE signal processing Magazine, 21 (2004),
  pp.~95--106.

\bibitem{Planat2002}
{\sc M.~Planat}, {\em {Ramanujan sums for signal processing of low frequency
  noise}}, in IEEE International frequency control symposium and PDA
  exhibition, 2002, pp.~5--10.

\bibitem{Planat2009}
{\sc M.~Planat, M.~Minarovjech, and M.~Saniga}, {\em {Ramanujan sums analysis
  of long-period sequences and 1/f noise}}, Epl, 85 (2009).

\bibitem{video}
{\sc N.~Rea, G.~Lacey, R.~Dahyotit, and R.~Dahyot}, {\em Multimodal periodicity
  analysis for illicit content detection in videos}, in The 3rd European
  Conference on Visual Media Production (CVMP 2006)-Part of the 2nd Multimedia
  Conference 2006, IET, 2006, pp.~106--114.

\bibitem{schmidt2017ecg}
{\sc M.~Schmidt, A.~Schumann, J.~M{\"u}ller, K.-J. B{\"a}r, and G.~Rose}, {\em
  Ecg derived respiration: comparison of time-domain approaches and application
  to altered breathing patterns of patients with schizophrenia}, Physiological
  measurement, 38 (2017), p.~601.

\bibitem{sethares1999}
{\sc W.~A. Sethares and T.~W. Staley}, {\em Periodicity {T}ransforms}, IEEE
  Transactions on Signal Processing, 47 (1999), pp.~2953--2964.

\bibitem{music2001}
\leavevmode\vrule height 2pt depth -1.6pt width 23pt, {\em Meter and
  periodicity in musical performance}, Journal of New Music Research, 30
  (2001), pp.~149--158.

\bibitem{deshapeapp}
{\sc L.~Su and H.-T. Wu}, {\em Extract fetal {ECG} from single-lead abdominal
  {ECG} by de-shape short time {F}ourier transform and nonlocal median},
  Frontiers in Applied Mathematics and Statistics, 3 (2017), p.~2.

\bibitem{su2019recovery}
{\sc P.-C. Su, S.~Miller, S.~Idriss, P.~Barker, and H.-T. Wu}, {\em Recovery of
  the fetal electrocardiogram for morphological analysis from two
  trans-abdominal channels via optimal shrinkage}, Physiological measurement,
  40 (2019), p.~115005.

\bibitem{electrophysiology1996heart}
{\sc {Task Force of the European Society of Cardiology and the North American
  Society of Pacing and Electrophysiology}}, {\em Heart rate variability:
  standards of measurement, physiological interpretation and clinical use. task
  force of the european society of cardiology and the north american society of
  pacing and electrophysiology}, Circulation, 93 (1996), pp.~1043--1065.

\bibitem{cepstrum3}
{\sc T.~Taxt}, {\em Comparison of cepstrum-based methods for radial blind
  deconvolution of ultrasound images}, IEEE transactions on ultrasonics,
  ferroelectrics, and frequency control, 44 (1997), pp.~666--674.

\bibitem{dictionary}
{\sc S.~Tenneti and P.~Vaidyanathan}, {\em Dictionary approaches for
  identifying periodicities in data}, in 2014 48th Asilomar Conference on
  Signals, Systems and Computers, IEEE, 2014, pp.~1967--1971.

\bibitem{PPV3}
{\sc S.~V. Tenneti and P.~Vaidyanathan}, {\em Nested periodic matrices and
  dictionaries: New signal representations for period estimation}, IEEE
  Transactions on Signal Processing, 63 (2015), pp.~3736--3750.

\bibitem{DNA3}
\leavevmode\vrule height 2pt depth -1.6pt width 23pt, {\em Detecting tandem
  repeats in {DNA} using {R}amanujan filter bank}, in 2016 IEEE International
  Symposium on Circuits and Systems (ISCAS), IEEE, 2016, pp.~21--24.

\bibitem{unified}
\leavevmode\vrule height 2pt depth -1.6pt width 23pt, {\em A unified theory of
  union of subspaces representations for period estimation}, IEEE Transactions
  on Signal Processing, 64 (2016), pp.~5217--5231.

\bibitem{iMUSIC}
{\sc S.~V. Tenneti and P.~P. Vaidyanathan}, {\em imusic: A family of music-like
  algorithms for integer period estimation}, IEEE Transactions on Signal
  Processing, 67 (2018), pp.~367--382.

\bibitem{tibshirani1996}
{\sc R.~Tibshirani}, {\em Regression shrinkage and selection via the lasso},
  Journal of the Royal Statistical Society: Series B (Methodological), 58
  (1996), pp.~267--288.

\bibitem{tibshirani2013lasso}
{\sc R.~J. Tibshirani et~al.}, {\em The lasso problem and uniqueness},
  Electronic Journal of statistics, 7 (2013), pp.~1456--1490.

\bibitem{cepstrum2}
{\sc K.~Tokuda, T.~Kobayashi, T.~Masuko, and S.~Imai}, {\em Mel-generalized
  cepstral analysis-a unified approach to speech spectral estimation}, in Third
  International Conference on Spoken Language Processing, 1994.

\bibitem{tropp}
{\sc J.~A. Tropp}, {\em Just relax: Convex programming methods for identifying
  sparse signals in noise}, IEEE Transactions on Information Theory, 52 (2006),
  pp.~1030--1051.

\bibitem{PPV1}
{\sc P.~Vaidyanathan}, {\em Ramanujan sums in the context of signal
  processing--{P}art {I}: Fundamentals}, IEEE Transactions on Signal
  Processing, 62 (2014), pp.~4145--4157.

\bibitem{PPV2}
\leavevmode\vrule height 2pt depth -1.6pt width 23pt, {\em Ramanujan sums in
  the context of signal processing--{P}art {II}: {FIR} representations and
  applications}, IEEE Transactions on Signal Processing, 62 (2014),
  pp.~4158--4172.

\bibitem{wainwright}
{\sc M.~J. Wainwright}, {\em Sharp thresholds for {H}igh-{D}imensional and
  noisy sparsity recovery using $l_1$-{C}onstrained {Q}uadratic {P}rogramming
  ({L}asso)}, IEEE transactions on Information Theory, 55 (2009),
  pp.~2183--2202.

\bibitem{SSTAPP6}
{\sc P.~Wang, J.~Gao, and Z.~Wang}, {\em Time-frequency analysis of seismic
  data using synchrosqueezing transform}, IEEE Geoscience and Remote Sensing
  Letters, 11 (2014), pp.~2042--2044.

\bibitem{wang2019novel}
{\sc S.-C. Wang, H.-T. Wu, P.-H. Huang, C.-H. Chang, C.-K. Ting, and Y.-T.
  Lin}, {\em Novel imaging revealing inner dynamics for cardiovascular waveform
  analysis via unsupervised manifold learning}, Anesthesia \& Analgesia (in
  press),  (2020).

\bibitem{Wu2013}
{\sc H.-T. Wu}, {\em Instantaneous frequency and wave shape functions ({I})},
  Applied and Computational Harmonic Analysis, 35 (2013), pp.~181--199.

\bibitem{Yang2}
{\sc J.~Xu, H.~Yang, and I.~Daubechies}, {\em Recursive diffeomorphism-based
  regression for shape functions}, SIAM Journal on Mathematical Analysis, 50
  (2018), pp.~5--32.

\bibitem{zhaoyu}
{\sc P.~Zhao and B.~Yu}, {\em On model selection consistency of {L}asso},
  Journal of Machine Learning Research, 7 (2006), pp.~2541--2563.

\bibitem{zhou2013heteroscedasticity}
{\sc Z.~Zhou}, {\em Heteroscedasticity and autocorrelation robust structural
  change detection}, J AM STAT ASSOC, 108 (2013), pp.~726--740.

\end{thebibliography}

\appendix

\end{document}